\documentclass[10pt]{article}
\usepackage[margin=.8in]{geometry}
\usepackage[T1]{fontenc}
\usepackage[utf8]{inputenc}
\usepackage{lmodern}
\usepackage{microtype}
\usepackage{amsmath,amssymb,amsthm,mathtools}
\usepackage{aliascnt}
\usepackage{enumitem}
\usepackage{booktabs,array}
\usepackage{tikz}
\usepackage{float}
\usetikzlibrary{arrows.meta,calc,positioning}
\usepackage[hidelinks]{hyperref}
\hypersetup{
  pdftitle={The Complexity of Kemeny Aggregation with Three Rankings},
  pdfauthor={Péter Madarasi},
  pdfsubject={Computational complexity of Kemeny aggregation, score, winner, precedence, and recognition problems, support thresholds, and center objectives},
  pdfkeywords={Kemeny rank aggregation, three rankings, Kendall--tau distance, NP-completeness, Theta-2-p completeness, coNP-completeness, Max-Cut, feedback arc set in tournaments, pairwise support threshold, permutation median, permutation center, Slater order}
}
\usepackage[nameinlink,noabbrev,capitalise,sort&compress]{cleveref}
\makeatletter
\edef\ReferenceTool@CleverefVersion{\csname ver@cleveref.sty\endcsname}%
\def\ReferenceTool@BrokenCleverefVersion{2018/03/27 v0.21.4 Intelligent cross-referencing}%
\ifx\ReferenceTool@CleverefVersion\ReferenceTool@BrokenCleverefVersion
  \let\cpageref\relax
  \DeclareRobustCommand{\cpageref}{%
    \@ifstar{\@crefstar{cpageref}}{\@cref{cpageref}}}%
  \let\Cpageref\relax
  \DeclareRobustCommand{\Cpageref}{%
    \@ifstar{\@crefstar{Cpageref}}{\@cref{Cpageref}}}%
  \providecommand*{\@setcpagerefrange}[3]{%
    \@@setcpagerefrange{#1}{#2}{cref}{#3}}%
  \providecommand*{\@setCpagerefrange}[3]{%
    \@@setcpagerefrange{#1}{#2}{Cref}{#3}}%
  \providecommand*{\@setlabelcpagerefrange}[3]{%
    \@@setcpagerefrange{#1}{#2}{labelcref}{#3}}%
\fi
\makeatother

\AtBeginDocument{%
}

\newcommand{\DeclareEquationCrefFormat}[1]{%
  \crefformat{#1}{##2\textup{(##1)}##3}%
  \Crefformat{#1}{##2\textup{(##1)}##3}%
  \crefrangeformat{#1}{##3\textup{(##1)}##4\nobreakdash--##5\textup{(##2)}##6}%
  \Crefrangeformat{#1}{##3\textup{(##1)}##4\nobreakdash--##5\textup{(##2)}##6}%
  \crefmultiformat{#1}%
    {##2\textup{(##1)}##3}%
    { and~##2\textup{(##1)}##3}%
    {, ##2\textup{(##1)}##3}%
    {, and~##2\textup{(##1)}##3}%
  \Crefmultiformat{#1}%
    {##2\textup{(##1)}##3}%
    { and~##2\textup{(##1)}##3}%
    {, ##2\textup{(##1)}##3}%
    {, and~##2\textup{(##1)}##3}%
  \crefrangemultiformat{#1}%
    {##3\textup{(##1)}##4\nobreakdash--##5\textup{(##2)}##6}%
    { and~##3\textup{(##1)}##4\nobreakdash--##5\textup{(##2)}##6}%
    {, ##3\textup{(##1)}##4\nobreakdash--##5\textup{(##2)}##6}%
    {, and~##3\textup{(##1)}##4\nobreakdash--##5\textup{(##2)}##6}%
  \Crefrangemultiformat{#1}%
    {##3\textup{(##1)}##4\nobreakdash--##5\textup{(##2)}##6}%
    { and~##3\textup{(##1)}##4\nobreakdash--##5\textup{(##2)}##6}%
    {, ##3\textup{(##1)}##4\nobreakdash--##5\textup{(##2)}##6}%
    {, and~##3\textup{(##1)}##4\nobreakdash--##5\textup{(##2)}##6}%
  \labelcrefformat{#1}{##2\textup{(##1)}##3}%
  \labelcrefrangeformat{#1}{##3\textup{(##1)}##4\nobreakdash--##5\textup{(##2)}##6}%
  \labelcrefmultiformat{#1}%
    {##2\textup{(##1)}##3}%
    { and~##2\textup{(##1)}##3}%
    {, ##2\textup{(##1)}##3}%
    {, and~##2\textup{(##1)}##3}%
  \labelcrefrangemultiformat{#1}%
    {##3\textup{(##1)}##4\nobreakdash--##5\textup{(##2)}##6}%
    { and~##3\textup{(##1)}##4\nobreakdash--##5\textup{(##2)}##6}%
    {, ##3\textup{(##1)}##4\nobreakdash--##5\textup{(##2)}##6}%
    {, and~##3\textup{(##1)}##4\nobreakdash--##5\textup{(##2)}##6}%
}
\DeclareEquationCrefFormat{equation}
\DeclareEquationCrefFormat{subequation}

\newcommand{\DeclareItemCrefFormat}[1]{%
  \crefformat{#1}{##2\textup{##1}##3}%
  \Crefformat{#1}{##2\textup{##1}##3}%
  \crefrangeformat{#1}{##3\textup{##1}##4\nobreakdash--##5\textup{##2}##6}%
  \Crefrangeformat{#1}{##3\textup{##1}##4\nobreakdash--##5\textup{##2}##6}%
  \crefmultiformat{#1}%
    {##2\textup{##1}##3}{ and~##2\textup{##1}##3}%
    {, ##2\textup{##1}##3}{, and~##2\textup{##1}##3}%
  \Crefmultiformat{#1}%
    {##2\textup{##1}##3}{ and~##2\textup{##1}##3}%
    {, ##2\textup{##1}##3}{, and~##2\textup{##1}##3}%
  \crefrangemultiformat{#1}%
    {##3\textup{##1}##4\nobreakdash--##5\textup{##2}##6}%
    { and~##3\textup{##1}##4\nobreakdash--##5\textup{##2}##6}%
    {, ##3\textup{##1}##4\nobreakdash--##5\textup{##2}##6}%
    {, and~##3\textup{##1}##4\nobreakdash--##5\textup{##2}##6}%
  \Crefrangemultiformat{#1}%
    {##3\textup{##1}##4\nobreakdash--##5\textup{##2}##6}%
    { and~##3\textup{##1}##4\nobreakdash--##5\textup{##2}##6}%
    {, ##3\textup{##1}##4\nobreakdash--##5\textup{##2}##6}%
    {, and~##3\textup{##1}##4\nobreakdash--##5\textup{##2}##6}%
  \labelcrefformat{#1}{##2\textup{##1}##3}%
  \labelcrefrangeformat{#1}{##3\textup{##1}##4\nobreakdash--##5\textup{##2}##6}%
  \labelcrefmultiformat{#1}%
    {##2\textup{##1}##3}{ and~##2\textup{##1}##3}%
    {, ##2\textup{##1}##3}{, and~##2\textup{##1}##3}%
  \labelcrefrangemultiformat{#1}%
    {##3\textup{##1}##4\nobreakdash--##5\textup{##2}##6}%
    { and~##3\textup{##1}##4\nobreakdash--##5\textup{##2}##6}%
    {, ##3\textup{##1}##4\nobreakdash--##5\textup{##2}##6}%
    {, and~##3\textup{##1}##4\nobreakdash--##5\textup{##2}##6}%
}
\DeclareItemCrefFormat{enumi}
\DeclareItemCrefFormat{enumii}
\DeclareItemCrefFormat{enumiii}
\DeclareItemCrefFormat{enumiv}

\theoremstyle{plain}
\newtheorem{theorem}{Theorem}

\newaliascnt{lemma}{theorem}
\newtheorem{lemma}[lemma]{Lemma}
\aliascntresetthe{lemma}

\newaliascnt{proposition}{theorem}
\newtheorem{proposition}[proposition]{Proposition}
\aliascntresetthe{proposition}

\newaliascnt{corollary}{theorem}
\newtheorem{corollary}[corollary]{Corollary}
\aliascntresetthe{corollary}

\newaliascnt{claim}{theorem}
\newtheorem{claim}[claim]{Claim}
\aliascntresetthe{claim}

\theoremstyle{definition}
\newaliascnt{definition}{theorem}

\aliascntresetthe{definition}

\theoremstyle{remark}
\newaliascnt{remark}{theorem}
\newtheorem{remark}[remark]{Remark}
\aliascntresetthe{remark}

\newcommand{\DeclareNamedCrefType}[3]{%
  \crefname{#1}{#2}{#3}%
  \Crefname{#1}{#2}{#3}%
  \crefrangelabelformat{#1}{##3##1##4\nobreakdash--##5##2##6}%
}

\DeclareNamedCrefType{section}{Section}{Sections}
\DeclareNamedCrefType{subsection}{Section}{Sections}
\DeclareNamedCrefType{subsubsection}{Section}{Sections}
\DeclareNamedCrefType{appendix}{Appendix}{Appendices}
\DeclareNamedCrefType{figure}{Figure}{Figures}
\DeclareNamedCrefType{table}{Table}{Tables}
\DeclareNamedCrefType{footnote}{Footnote}{Footnotes}
\DeclareNamedCrefType{page}{Page}{Pages}

\DeclareNamedCrefType{theorem}{Theorem}{Theorems}
\DeclareNamedCrefType{lemma}{Lemma}{Lemmas}
\DeclareNamedCrefType{proposition}{Proposition}{Propositions}
\DeclareNamedCrefType{corollary}{Corollary}{Corollaries}
\DeclareNamedCrefType{claim}{Claim}{Claims}
\DeclareNamedCrefType{definition}{Definition}{Definitions}
\DeclareNamedCrefType{remark}{Remark}{Remarks}

\makeatletter
\@onlypreamble\DeclareEquationCrefFormat
\@onlypreamble\DeclareItemCrefFormat
\@onlypreamble\DeclareNamedCrefType
\makeatother

\newcommand{\KTdist}{d_{\mathrm{KT}}}
\newcommand{\Kemeny}{\mathcal K}
\newcommand{\MaxCut}{\textnormal{\textsc{Max-Cut}}}
\newcommand{\cC}{\mathcal C}
\newcommand{\before}{\prec}
\newcommand{\maxcutvalue}{\operatorname{maxcut}}
\newcommand{\cut}{\operatorname{cut}}
\newcommand{\fas}{\operatorname{fas}}
\newcommand{\support}{\operatorname{supp}}
\newcommand{\ThetaTwoP}{\Theta_2^p}
\newcommand{\KemenyScore}{\textnormal{\textsc{Kemeny Score}}}
\newcommand{\KemenyWinner}{\textnormal{\textsc{Kemeny Winner}}}
\newcommand{\KemenyUniqueWinner}{\textnormal{\textsc{Kemeny Unique Winner}}}
\newcommand{\KemenyPossiblePrecedence}{\textnormal{\textsc{Kemeny Possible Precedence}}}
\newcommand{\KemenyNecessaryPrecedence}{\textnormal{\textsc{Kemeny Necessary Precedence}}}
\newcommand{\KemenyConsensusRecognition}{\textnormal{\textsc{Kemeny Consensus Recognition}}}
\newcommand{\UniqueKemenyConsensusRecognition}{\textnormal{\textsc{Unique Kemeny Consensus Recognition}}}
\newcommand{\SlaterWinner}{\textnormal{\textsc{Slater Winner}}}
\newcommand{\SlaterUniqueWinner}{\textnormal{\textsc{Slater Unique Winner}}}
\newcommand{\SlaterPossiblePrecedence}{\textnormal{\textsc{Slater Possible Precedence}}}
\newcommand{\SlaterNecessaryPrecedence}{\textnormal{\textsc{Slater Necessary Precedence}}}
\newcommand{\SlaterConsensusRecognition}{\textnormal{\textsc{Slater Consensus Recognition}}}
\newcommand{\UniqueSlaterConsensusRecognition}{\textnormal{\textsc{Unique Slater Consensus Recognition}}}
\newcommand{\TournamentFAS}{\textnormal{\textsc{Feedback Arc Set in Tournaments}}}
\newcommand{\KendallCenter}{\textnormal{\textsc{Kendall--Tau Center}}}
\newcommand{\KOpt}{\operatorname{Opt}_{\mathrm K}}

\title{The Complexity of Kemeny Aggregation with Three Rankings}
\author{P\'eter Madarasi\thanks{HUN-REN Alfr\'ed R\'enyi Institute of Mathematics, Re\'altanoda u.\ 13--15., Budapest H-1053, Hungary; and Department of Operations Research, ELTE E\"otv\"os Lor\'and University, P\'azm\'any P.\ s.\ 1/c, Budapest H-1117, Hungary. E-mail: \texttt{madarasi@renyi.hu}}}
\date{}

\begin{document}
\maketitle

\begin{abstract}
The Kemeny rule aggregates rankings by choosing an order that minimizes the sum of its Kendall--tau distances to the input rankings.
We prove that \KemenyScore{} is NP-complete for exactly three unweighted rankings, even when every candidate pair is split $2$\nobreakdash-to\nobreakdash-$1$.
On the same restricted profiles, \KemenyWinner{}, \KemenyUniqueWinner{}, \KemenyPossiblePrecedence{}, and \KemenyNecessaryPrecedence{} are $\ThetaTwoP$-complete, while \KemenyConsensusRecognition{} and \UniqueKemenyConsensusRecognition{} are coNP-complete; the hard instances induce tournaments of majority dimension exactly $3$.
The reduction also determines the exact maximum-cut value from the optimal Kemeny score and recovers a maximum cut from any Kemeny-optimal aggregate in polynomial time.
For every fixed $q\geq3$ and $\lceil q/2\rceil\leq s\leq q$, requiring pairwise support at least $s$ yields a sharp dichotomy: if $3s\leq2q$, \KemenyScore{} is NP-complete, the four winner and precedence problems are $\ThetaTwoP$-complete, and the two recognition problems are coNP-complete; if $3s>2q$, the majority tournament is transitive and its unique topological order is the unique Kemeny-optimal aggregate.
In the hard case, exact support $s$ suffices when $s>q/2$, while supports in $\{s,s+1\}$ suffice when $s=q/2$.
These results give complete fixed-profile-size classifications and transfer to Slater orders, permutation medians, and maximum-likelihood central rankings in the Mallows model with fixed dispersion.
Finally, a six-copy construction proves NP-completeness of both \KemenyScore{} and \KendallCenter{} for three pairwise-equidistant rankings that still split every candidate pair $2$\nobreakdash-to\nobreakdash-$1$.
For outputs on $N$ candidates, the common pairwise distance attains the maximum possible value $\frac23\binom N2$.
The construction also gives affine formulas for the optimal Kemeny score and center radius, describes all Kemeny-optimal output orders, and shows that the output has a unique Kemeny-optimal order and a unique center exactly when the input has a unique Kemeny-optimal order.

\medskip
\noindent\textbf{Keywords.}
Kemeny rank aggregation; computational complexity; three rankings; Kendall--tau distance; pairwise support thresholds; majority dimension; pairwise-equidistant rankings; Kendall--tau center; feedback arc set in tournaments; permutation median; Slater order; Mallows model.
\end{abstract}

\section{Introduction}\label{sec:introduction}

The Kemeny rule aggregates full rankings on a set of candidates.
A full ranking is a strict total order, and the Kendall--tau distance between two orders is the number of candidate pairs on which they disagree.
Given a list of input rankings, the Kemeny rule selects an aggregate order that minimizes the sum of its Kendall--tau distances to the inputs.
This sum is the Kemeny objective of the aggregate.
This paper studies the computational complexity of minimizing the Kemeny objective and of the associated winner, precedence, and recognition problems when the number of input rankings is fixed.

With one input ranking, the unique Kemeny-optimal aggregate is that ranking.
With two input rankings, each pair on which the rankings disagree contributes one unit to every aggregate, while the comparisons on which they agree form a partial order; the Kemeny-optimal aggregates are exactly its linear extensions.
Before the present work, among fixed profile sizes $q\geq3$, only $q=3$ and $q=5$ remained unresolved for the Kemeny objective; hardness was known for every other such $q$~\cite{BachmeierEtAl2019,BiedlBrandenburgDeng2009}.
The three-ranking case had also remained open in the permutation-median literature~\cite{BlinEtAl2011,JainThakur2026}.
This paper resolves the three-ranking case.\footnote{Peters independently proved the same result using a reduction found by GPT~5.6 Sol Ultra and subsequently simplified~\cite{Peters2026}. The results in the present paper were obtained independently, and this manuscript was essentially complete when his preprint appeared.}

The threshold problem \KemenyScore{} asks whether some aggregate has total Kendall--tau distance at most a given nonnegative integer.
The \KemenyWinner{} and \KemenyUniqueWinner{} problems ask whether a designated candidate is first in some Kemeny-optimal aggregate or is the only candidate that can be first in a Kemeny-optimal aggregate, respectively.
The \KemenyPossiblePrecedence{} and \KemenyNecessaryPrecedence{} problems ask whether some or every Kemeny-optimal aggregate places one designated candidate before another.
The \KemenyConsensusRecognition{} and \UniqueKemenyConsensusRecognition{} problems ask whether a supplied aggregate is Kemeny-optimal or uniquely Kemeny-optimal.
We write $\ThetaTwoP$ for polynomial time with nonadaptive access to an NP oracle.

A three-ranking profile is \emph{all\nobreakdash-$2$\nobreakdash-to\nobreakdash-$1$} if every candidate pair is ordered one way by exactly two rankings and the other way by the third.
Its strict majority relation is a tournament, and the \emph{majority dimension} of a tournament is the minimum number of rankings whose strict majority relation realizes it.
\begin{theorem}\label{thm:main}
\KemenyScore{} is NP-complete for exactly three unweighted full rankings, even when the profile is all\nobreakdash-$2$\nobreakdash-to\nobreakdash-$1$.
\end{theorem}

\begin{theorem}\label{thm:three-ranking-optimality}
On all\nobreakdash-$2$\nobreakdash-to\nobreakdash-$1$ profiles consisting of exactly three unweighted full rankings, the following statements hold.
\begin{enumerate}[label=\textup{(\roman*)}]
\item \KemenyWinner{}, \KemenyUniqueWinner{}, \KemenyPossiblePrecedence{}, and \KemenyNecessaryPrecedence{} are $\ThetaTwoP$-complete.
\item \KemenyConsensusRecognition{} and \UniqueKemenyConsensusRecognition{} are coNP-complete.
\end{enumerate}
\end{theorem}

The majority tournament in every hard instance constructed for these two theorems contains a directed triangle.
Since the supplied three rankings realize the tournament and every tournament realizable by at most two rankings is transitive, these tournaments have majority dimension exactly $3$.
This property is used in the later support and Slater consequences.

The proof of \cref{thm:main} reduces from \MaxCut{} on simple $4$-regular graphs.
For every source graph, the optimal Kemeny score determines the maximum-cut value, and any Kemeny-optimal aggregate can be transformed into a maximum cut in polynomial time; see \cref{prop:maxcut-recovery}.
The reductions proving \cref{thm:three-ranking-optimality} adapt the construction used to prove \cref{thm:main}.
The winner and precedence reduction compares the maximum independent-set sizes of two graphs, while the recognition reductions test whether a supplied independent set is maximum.

The second classification is stated in terms of pairwise support.
For a profile of $q$ rankings, the support of a candidate pair is the larger of the two numbers of rankings that order the pair in each direction.
For every fixed $q\geq3$ and minimum pairwise support $s$ with $\lceil q/2\rceil\leq s\leq q$, \KemenyScore{} is NP-complete, the winner and precedence problems are $\ThetaTwoP$-complete, and the two recognition problems are coNP-complete exactly when $3s\leq2q$.
When $3s>2q$, the majority tournament is transitive, and its unique topological order is the unique Kemeny-optimal aggregate, computable in $O(qN^2)$ time on $N$ candidates.
In the case $3s\leq2q$, exact support $s$ suffices whenever $s>q/2$.
In the boundary case $s=q/2$, supports in $\{s,s+1\}$ suffice.
For four rankings, \KemenyScore{} is NP-complete, the four winner and precedence problems are $\ThetaTwoP$-complete, and the two recognition problems are coNP-complete when pairwise supports belong to $\{2,3\}$; all seven problems are polynomial-time solvable when every pair has support at least $3$.
The classification also has fractional-support and uniform-margin forms and yields the fixed-profile-size classification.

A further contribution concerns the Kendall--tau minimax (center) objective, also called maximum rank aggregation or the permutation-center problem~\cite{AlvinChakraborty2023,BachmaierEtAl2015,BiedlBrandenburgDeng2009}.
The threshold problem \KendallCenter{} asks whether some aggregate lies within a prescribed radius of every input ranking.
A six-copy construction ensures that no candidate pair is ordered unanimously while making the three output rankings pairwise equidistant.

\begin{theorem}\label{thm:center}
Both \KemenyScore{} and \KendallCenter{} are NP-complete for profiles consisting of exactly three pairwise-equidistant rankings, even when the profile is all\nobreakdash-$2$\nobreakdash-to\nobreakdash-$1$.
On the hard instances with $N$ candidates, the common pairwise distance is $\frac23\binom N2$; no three pairwise-equidistant rankings on $N$ candidates can have a larger common distance.
On these instances, the optimal Kemeny score is three times the center radius.
\end{theorem}

The six-copy construction gives explicit affine formulas for the optimal Kemeny score and center radius of the output profile and characterizes every Kemeny-optimal output order through its restriction to each copy.
Every optimal center is Kemeny-optimal, and the original profile has a unique Kemeny-optimal order exactly when the output profile has a unique Kemeny-optimal order and a unique center.

\paragraph{Relation to prior work.}
Kemeny introduced the aggregation rule~\cite{Kemeny1959}; it uses Kendall's rank distance~\cite{Kendall1938} and has a classical axiomatic characterization due to Young and Levenglick~\cite{YoungLevenglick1978}.
Unrestricted NP-hardness was established by Bartholdi, Tovey, and Trick~\cite{BartholdiToveyTrick1989}.
Hemaspaandra, Spakowski, and Vogel proved $\ThetaTwoP$-completeness of \KemenyWinner{}~\cite{HemaspaandraSpakowskiVogel2005}, while Fitzsimmons and Hemaspaandra proved coNP-completeness of recognizing a supplied Kemeny consensus~\cite{FitzsimmonsHemaspaandra2021}.
For the related Slater rule, constant-voter realizations and winner complexity have been studied by Bachmeier et al.\ and Lampis~\cite{BachmeierEtAl2019,Lampis2022}.
The support result uses the classical fact that if every pairwise majority is supported by more than two thirds of the rankings, then the majority tournament is transitive~\cite{Vidu2000}.
Fixed-parameter algorithms for Kemeny aggregation and approximation algorithms for weighted feedback arc set in tournaments provide complementary positive results~\cite{BetzlerEtAl2008,BetzlerEtAl2009,KenyonMathieuSchudy2007}.

\paragraph{Organization.}
\Cref{sec:preliminaries} gives the problem definitions and pairwise score identities.
\Cref{sec:score-reduction} proves \cref{thm:main}, including exact value and solution recovery.
\Cref{sec:other-base-reductions} proves the winner, precedence, and recognition results.
\Cref{sec:support} gives the transformations used to prescribe the number of rankings and pairwise supports, the constructions for four rankings, and the support-threshold classifications.
\Cref{sec:consequences} gives equivalent formulations, the six-copy construction for pairwise-equidistant profiles, consequences for the center and other distance objectives, and the fixed-profile-size classifications.

\section{Preliminaries}\label{sec:preliminaries}

Throughout, an \emph{order} means a strict total order.
A \emph{ranking} is an order used as an input, and an \emph{aggregate order}, shortened to \emph{aggregate} when no confusion is possible, is a candidate solution.
We use \emph{permutation} only when stating formulations from the permutation literature.
Let $\cC$ be a finite candidate set and let $\pi_1,\dots,\pi_q$ be rankings on $\cC$.
Their list $\Pi=(\pi_1,\dots,\pi_q)$ is a \emph{profile}; repetitions are allowed.
Write $x\before_\pi y$ when the order $\pi$ places $x$ before $y$.
Profiles are explicitly represented, and every input ranking has unit weight.
All decision-problem hardness reductions are polynomial-time many-one reductions.
All undirected graphs are finite and simple.
When a theorem fixes the profile size $q$ or support threshold $s$, these quantities are constants, whereas the fractional-support-threshold theorem allows $q$ to vary with the input.
In every tournament or Slater formulation, a realizing profile is supplied as part of the instance.
The paper proves the majority dimension of its constructed tournaments but does not address recognition of that dimension when no realization is given.
For two orders $\pi$ and $\sigma$, their Kendall--tau distance is
\[
\KTdist(\pi,\sigma)
=
\#\bigl\{\{x,y\}\subseteq\cC:\pi\text{ and }\sigma\text{ order }x,y\text{ oppositely}\bigr\}.
\]
The Kemeny objective and optimal Kemeny score are
\[
\Kemeny_\Pi(\sigma)=\sum_{r=1}^{q}\KTdist(\pi_r,\sigma),
\qquad
\Kemeny_\Pi^{\star}=\min_{\sigma}\Kemeny_\Pi(\sigma).
\]
The set of Kemeny-optimal orders is $\KOpt(\Pi)=\arg\min_\sigma\Kemeny_\Pi(\sigma)$.
When the profile is unambiguous, we write $\Kemeny(\sigma)$ and $\Kemeny^{\star}$.
The \KemenyScore{} decision problem takes a profile and a nonnegative integer threshold $\kappa$ and asks whether $\Kemeny_\Pi^{\star}\leq\kappa$.
Allowing negative thresholds would add only trivial no-instances.
A proposed aggregate is a polynomial-size certificate, and its Kemeny objective value can be computed in $O(q|\cC|^2)$ time.
Thus \KemenyScore{} belongs to NP\@.

The Kendall--tau center radius of $\Pi$ is
\[
\mathcal R_\Pi^\star
=
\min_\sigma\max_{r\in\{1,\dots,q\}}\KTdist(\pi_r,\sigma).
\]
The decision problem \KendallCenter{} asks, given $\Pi$ and a nonnegative integer $\rho$, whether $\mathcal R_\Pi^\star\leq\rho$.
It belongs to NP, and the optimizing center may be any aggregate order, not only one of the input rankings.

The \KemenyWinner{} problem takes a profile and a candidate $c$ and asks whether some Kemeny-optimal aggregate places $c$ first.
The \KemenyUniqueWinner{} problem asks whether $c$ is the only Kemeny winner.
Since every aggregate has exactly one first candidate, this is equivalent to asking whether every Kemeny-optimal aggregate places $c$ first.
The \KemenyPossiblePrecedence{} problem takes a profile and two distinct candidates $c,d$ and asks whether some Kemeny-optimal aggregate places $c\before d$.
The \KemenyNecessaryPrecedence{} problem asks whether every Kemeny-optimal aggregate places $c\before d$.
The term ``precedence'' distinguishes these decision problems from the search task of computing an entire Kemeny-optimal aggregate.
The \KemenyConsensusRecognition{} problem takes a profile and an aggregate order $\tau$ and asks whether $\tau$ is Kemeny-optimal.
The \UniqueKemenyConsensusRecognition{} problem asks whether $\tau$ is the unique Kemeny-optimal aggregate.
A strictly better aggregate certifies a no-instance of \KemenyConsensusRecognition{}.
For \UniqueKemenyConsensusRecognition{}, a no-certificate is an order $\sigma\neq\tau$ satisfying $\Kemeny_\Pi(\sigma)\leq\Kemeny_\Pi(\tau)$.
Thus both recognition problems belong to coNP\@.

For a profile $\Pi=(\pi_1,\dots,\pi_q)$ and distinct candidates $x,y$, let
\[
N_\Pi(x,y)=|\{r\in\{1,\dots,q\}:x\before_{\pi_r}y\}|.
\]
The \emph{pairwise support} of $x$ and $y$ is
\[
\support_\Pi(x,y)
=
\max\{N_\Pi(x,y),N_\Pi(y,x)\}.
\]
This quantity is symmetric in $x$ and $y$.
The pair is \emph{tied} when $N_\Pi(x,y)=N_\Pi(y,x)$, and a profile is \emph{tie-free} if none of its candidate pairs is tied.
When $N_\Pi(x,y)>N_\Pi(y,x)$, write $x\to y$ and call this the \emph{majority arc} of the pair.
Its margin is
\[
w_\Pi(x,y)=N_\Pi(x,y)-N_\Pi(y,x)>0.
\]
The arc $x\to y$ is \emph{backward} in $\sigma$ when $y\before_\sigma x$.
Let
\[
B_\Pi=\sum_{\{x,y\}\subseteq \cC}\min\{N_\Pi(x,y),N_\Pi(y,x)\}.
\]

\begin{lemma}\label{lem:tournament-identity}
For every aggregate order $\sigma$,
\[
\Kemeny_\Pi(\sigma)
=
B_\Pi+
\sum_{\substack{x\to y\\ y\before_\sigma x}}w_\Pi(x,y).
\]
\end{lemma}

\begin{proof}
Fix a pair with majority arc $x\to y$, and put $a=N_\Pi(x,y)$ and $b=N_\Pi(y,x)$.
If $\sigma$ places $x$ before $y$, the pair contributes $b$.
If it places $y$ before $x$, the pair contributes $a=b+(a-b)$.
Thus the pair always contributes the smaller directional count $b$ and contributes the additional margin $a-b$ exactly when its majority arc is backward.
A tied pair contributes its common directional count in either direction.
Summing over all pairs proves the identity.
\end{proof}

For distinct candidates $x,y$, define their absolute majority margin by $\delta_\Pi(x,y)=|N_\Pi(x,y)-N_\Pi(y,x)|$.
This quantity is symmetric in $x$ and $y$.
A common pairwise majority margin $\Delta$ for a profile of $q$ rankings must satisfy $\Delta\in\{0,\dots,q\}$ and $\Delta\equiv q\pmod 2$.
We say that a profile has \emph{uniform pairwise majority margin $\Delta$} when $\delta_\Pi(x,y)=\Delta$ for all distinct candidates $x,y$.
If $\Delta>0$, the majority relation is a tournament, denoted by $T_\Pi$.
The \emph{majority dimension} of a tournament is the minimum number of strict total orders whose strict majority relation is that tournament.
For a tournament $T$ and an order $\sigma$, let $b_T(\sigma)$ be the number of arcs of $T$ that are backward in $\sigma$.
For a profile $\Pi$ whose majority relation is a tournament, write $b_\Pi(\sigma)=b_{T_\Pi}(\sigma)$.

\begin{lemma}\label{lem:uniform-score}
Let $\Pi$ be a profile of $q$ rankings with common positive pairwise majority margin $\Delta$.
For every aggregate order $\sigma$,
\begin{equation}\label{eq:uniform-score}
\Kemeny_\Pi(\sigma)
=
\frac{q-\Delta}{2}\binom{|\cC|}{2}
+
\Delta b_\Pi(\sigma).
\end{equation}
\end{lemma}

\begin{proof}
Every pair contributes the smaller directional count $(q-\Delta)/2$ and an additional $\Delta$ exactly when its majority arc is backward.
The formula follows from \cref{lem:tournament-identity}.
\end{proof}

For any order, its backward arcs form a feedback arc set of the tournament.
Conversely, if deleting a set of arcs makes the tournament acyclic, every backward arc in a topological order of the remaining digraph belongs to the deleted set.
Thus, for positive uniform margins, Kemeny optimization is exactly unweighted Minimum Feedback Arc Set in the majority tournament, up to an additive constant and a positive scale.
The minimizing orders are the Slater orders of $T_\Pi$~\cite{Slater1961,Alon2006}.
The \SlaterWinner{} problem asks whether some Slater order places a designated vertex first.
The \SlaterUniqueWinner{} problem asks whether this vertex is the only Slater winner; equivalently, every Slater order places it first.
The \SlaterPossiblePrecedence{} and \SlaterNecessaryPrecedence{} problems ask whether some or every Slater order, respectively, places one designated vertex before another.
The \SlaterConsensusRecognition{} problem asks whether a supplied order is a Slater order, and the \UniqueSlaterConsensusRecognition{} problem asks whether it is the unique Slater order.
For a tournament $T$, write $\fas(T)=\min_\sigma b_T(\sigma)$.
The threshold problem \TournamentFAS{} asks whether $T$ has an order with at most a supplied number of backward arcs.

For three rankings, an all\nobreakdash-$2$\nobreakdash-to\nobreakdash-$1$ profile has uniform pairwise majority margin $1$.
The specialization of \cref{eq:uniform-score} is
\begin{equation}\label{eq:all-two-to-one-score}
\Kemeny_\Pi(\sigma)=\binom{|\cC|}{2}+b_\Pi(\sigma).
\end{equation}
It therefore suffices to control the number of backward majority arcs in the main reduction.

\section{The three-ranking Kemeny reduction}\label{sec:score-reduction}

This section proves \cref{thm:main} by reducing from \MaxCut{} on simple $4$-regular graphs.
The reduction translates a cut of the source graph into an aggregate order of three rankings.
For each graph vertex, the construction creates six vertex blocks and six padding blocks, each appearing consecutively in every input ranking.
Two specified relative orders of the six vertex blocks represent the two possible sides of the cut, while the padding blocks do not encode the cut.
For each graph edge, one additional candidate makes the score depend on the choices made at both endpoints.
The use of a $4$-regular source graph lets these two kinds of score contributions reproduce the cut size.

Two additional facts are needed for the reduction to be correct.
First, every aggregate order can be transformed, without increasing its Kemeny score, so that each block is contiguous, the vertex groups occur in index order, and the six vertex blocks for each graph vertex use one of the two encoding orders.
The transformed aggregate therefore encodes a cut.
Second, for each encoded cut, the minimum score over all positions and relative orders of the edge candidates lies in an interval whose width is smaller than the score decrease caused by increasing the cut size by one.
A single threshold can therefore distinguish cuts meeting the threshold of the Max-Cut instance from all smaller cuts.

We first establish the restriction of \MaxCut{} used in the reduction and derive the cut-size formula needed by the construction.
We then define the three input rankings and prove that every candidate pair is split exactly $2$\nobreakdash-to\nobreakdash-$1$.
The next two subsections normalize arbitrary aggregates and determine the best position of each edge candidate.
The resulting bounds prove \cref{thm:main}; the final subsection strengthens them to exact value and solution recovery.

\subsection{Source problem}\label{sec:source}

We prove that \MaxCut{} restricted to simple $4$-regular graphs is NP-complete.
For a graph $H$ and a subset $X\subseteq V(H)$, let $\cut_H(X)$ be the number of edges with exactly one endpoint in $X$.
The sets $X$ and $V(H)\setminus X$ are the two sides of the same cut.
Let $\maxcutvalue(H)=\max_{X\subseteq V(H)}\cut_H(X)$.
The decision problem \MaxCut{} takes a graph $H$ and a positive integer $t$ and asks whether $\maxcutvalue(H)\geq t$.

\begin{lemma}\label{lem:four-regular}
\MaxCut{} is NP-complete on simple $4$-regular graphs.
\end{lemma}

\begin{proof}
Yannakakis~\cite[Theorem~13]{Yannakakis1978} proves NP-completeness of the following edge-deletion problem: given a cubic graph $H$ and an integer $d\geq0$, decide whether deleting at most $d$ edges can make $H$ bipartite.
For every $X\subseteq V(H)$, the edges counted by $\cut_H(X)$ form a bipartite spanning subgraph with bipartition $X,V(H)\setminus X$.
Conversely, if a spanning subgraph of $H$ is bipartite and $X$ is one side of a bipartition, then all of its edges are counted by $\cut_H(X)$.
Hence the largest number of edges in a bipartite spanning subgraph is $\maxcutvalue(H)$, and the minimum number of edges whose deletion makes $H$ bipartite is $|E(H)|-\maxcutvalue(H)$.

If $d\geq |E(H)|$, then the deletion instance is a yes-instance; map it to the fixed cubic yes-instance $(K_4,1)$ of \MaxCut{}.
If $d<|E(H)|$, map it to $(H,|E(H)|-d)$.
In both cases, the two instances have the same answer, and the output threshold is positive and at most the number of edges in the output graph.
Thus \MaxCut{} is NP-hard on cubic instances with $1\leq t\leq |E(H)|$.

Let $(H,t)$ be such a cubic \MaxCut{} instance, with $H=(U,F)$.
Take two disjoint copies $H_1$ and $H_2$ of $H$, and write $u_r$ for the copy of $u\in U$ in $H_r$, where $r\in\{1,2\}$.
Start with their disjoint union and add the edge $\{u_1,u_2\}$ for every $u\in U$.
These edges form a perfect matching between the two copies.
Every vertex has its three neighbors inside its copy and one neighbor in the other copy, so the resulting graph $G$ is simple and $4$-regular.

For $X\subseteq U$, set $Y_X=\{u_1:u\in X\}\cup\{u_2:u\notin X\}$.
Within $H_1$, the edges with exactly one endpoint in $Y_X$ correspond exactly to the edges counted by $\cut_H(X)$; within $H_2$, they correspond to the same cut with its sides exchanged.
Every matching edge has exactly one endpoint in $Y_X$.
Taking $X$ to represent a maximum cut of $H$ therefore gives a cut of $G$ of size $2\maxcutvalue(H)+|U|$.
Conversely, for any $Y\subseteq V(G)$, define $X_r=\{u\in U:u_r\in Y\}$ for $r\in\{1,2\}$.
The edges of $H_r$ with exactly one endpoint in $Y$ contribute $\cut_H(X_r)\leq\maxcutvalue(H)$, and at most $|U|$ matching edges have exactly one endpoint in $Y$.
Therefore $\maxcutvalue(G)=2\maxcutvalue(H)+|U|$.
It follows that $\maxcutvalue(H)\geq t$ exactly when $\maxcutvalue(G)\geq2t+|U|$.
Moreover, $1\leq2t+|U|\leq2|F|+|U|=|E(G)|$.
The map $(H,t)\mapsto(G,2t+|U|)$ is polynomial-time computable.
Since \MaxCut{} is in NP, \MaxCut{} on simple $4$-regular graphs is NP-complete.
\end{proof}

\Cref{lem:four-regular} shows that \MaxCut{} remains NP-complete on instances $(G,t)$ in which $G$ is simple and $4$-regular and $1\leq t\leq |E(G)|$.
Fix such an instance, write $G=(V,E)$, put $n=|V|$ and $m=|E|$, and denote the vertices by $v_1,\dots,v_n$.
Then $1\leq t\leq m$, and $4$-regularity gives $m=2n$.

For $X\subseteq V$, let $\chi_X:V\to\{0,1\}$ be its characteristic function and write $x_i=\chi_X(v_i)$.
Counting the four edges incident with each vertex in $X$, with an edge whose two endpoints lie in $X$ counted twice, gives
\begin{equation}\label{eq:cut-identity}
\cut_G(X)=4\sum_{i=1}^n x_i-2\sum_{\{v_i,v_j\}\in E}x_ix_j.
\end{equation}
The construction will make the objective depend on $X$ through the two sums in \cref{eq:cut-identity}.
The first depends on the choices made separately at the vertices, while the second couples the choices made at the two endpoints of each edge.
Their precise contributions are computed in the next subsection.

\subsection{Construction of the three rankings}\label{sec:construction}

We first analyze a profile of three rankings on six local candidates.
Its majority tournament has exactly two optimal orders.
After each local candidate is replaced by a block of candidates, these two orders will encode whether a graph vertex lies in the chosen side $X$ of the cut.
We then add padding blocks and edge candidates and define the three input rankings.

Consider the six-candidate local profile
\begin{equation}\label{eq:local-votes}
\nu_1=123456,\qquad \nu_2=632541,\qquad \nu_3=452163.
\end{equation}
Here, for example, $123456$ denotes the order $1\before2\before3\before4\before5\before6$.
Every pair of local candidates is split $2$\nobreakdash-to\nobreakdash-$1$.
\Cref{fig:local-tournament} shows the majority tournament twice, with different arcs highlighted for the proof of the next lemma.

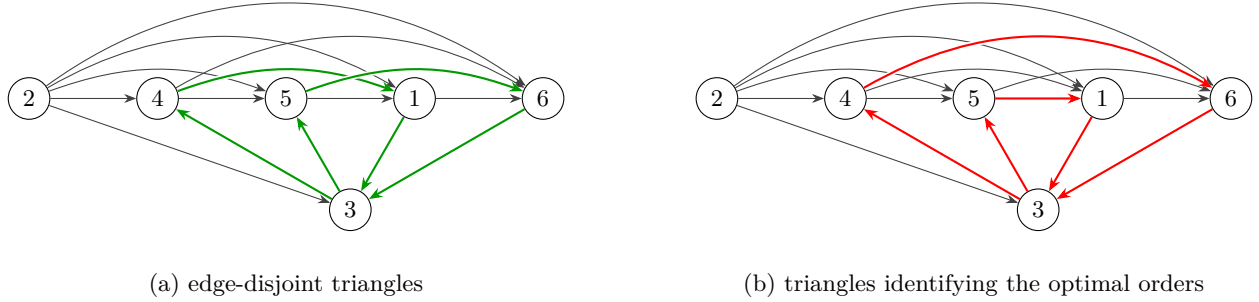
\begin{figure}[H]
\centering
\begin{tikzpicture}[
  vertex/.style={circle,draw,fill=white,inner sep=0pt,minimum size=5.5mm,font=\small},
  basearc/.style={-{Stealth[length=1.6mm]},thin,black!75},
  lowerarc/.style={preaction={draw=white,line width=2pt},-{Stealth[length=1.8mm]},thick,green!60!black},
  forcingarc/.style={preaction={draw=white,line width=2pt},-{Stealth[length=1.8mm]},thick,red},
  lowerarcplain/.style={-{Stealth[length=1.8mm]},thick,green!60!black},
  forcingarcplain/.style={-{Stealth[length=1.8mm]},thick,red}
]
\begin{scope}[xshift=-4.55cm]
\coordinate (two) at (-3.40,.75);
\coordinate (four) at (-1.70,.75);
\coordinate (five) at (0,.75);
\coordinate (one) at (1.70,.75);
\coordinate (six) at (3.40,.75);
\coordinate (three) at (.85,-.72224);
\foreach \name/\coord in {2/two,4/four,5/five,1/one,6/six,3/three}
  \node[vertex] (L\coord) at (\coord) {\name};
\draw[basearc] (Ltwo) -- (Lfour);
\draw[basearc] (Lfour) -- (Lfive);
\draw[basearc] (Lfive) -- (Lone);
\draw[basearc] (Lone) -- (Lsix);
\draw[basearc] (Ltwo) to[bend left=20] (Lfive);
\draw[basearc] (Ltwo) to[bend left=30] (Lone);
\draw[basearc] (Lfour) to[bend left=30] (Lsix);
\draw[basearc] (Ltwo) to[bend left=36] (Lsix);
\draw[basearc] (Ltwo) -- (Lthree);
\draw[lowerarcplain] (Lone) -- (Lthree);
\draw[lowerarcplain] (Lthree) -- (Lfour);
\draw[lowerarc] (Lfour) to[bend left=20] (Lone);
\draw[lowerarcplain] (Lthree) -- (Lfive);
\draw[lowerarc] (Lfive) to[bend left=20] (Lsix);
\draw[lowerarcplain] (Lsix) -- (Lthree);
\node[align=center,font=\small] at (0,-1.72) {(a) edge-disjoint triangles};
\end{scope}
\begin{scope}[xshift=4.55cm]
\coordinate (two) at (-3.40,.75);
\coordinate (four) at (-1.70,.75);
\coordinate (five) at (0,.75);
\coordinate (one) at (1.70,.75);
\coordinate (six) at (3.40,.75);
\coordinate (three) at (.85,-.72224);
\foreach \name/\coord in {2/two,4/four,5/five,1/one,6/six,3/three}
  \node[vertex] (R\coord) at (\coord) {\name};
\draw[basearc] (Rtwo) -- (Rfour);
\draw[basearc] (Rfour) -- (Rfive);
\draw[basearc] (Rone) -- (Rsix);
\draw[basearc] (Rtwo) to[bend left=20] (Rfive);
\draw[basearc] (Rfour) to[bend left=20] (Rone);
\draw[basearc] (Rfive) to[bend left=20] (Rsix);
\draw[basearc] (Rtwo) to[bend left=30] (Rone);
\draw[basearc] (Rtwo) to[bend left=36] (Rsix);
\draw[basearc] (Rtwo) -- (Rthree);
\draw[forcingarcplain] (Rone) -- (Rthree);
\draw[forcingarcplain] (Rthree) -- (Rfive);
\draw[forcingarcplain] (Rfive) -- (Rone);
\draw[forcingarcplain] (Rthree) -- (Rfour);
\draw[forcingarc] (Rfour) to[bend left=30] (Rsix);
\draw[forcingarcplain] (Rsix) -- (Rthree);
\node[align=center,font=\small] at (0,-1.72) {(b) triangles identifying the optimal orders};
\end{scope}
\end{tikzpicture}
\caption{Two drawings of the same majority tournament for the profile in \cref{eq:local-votes}.
The green triangles in (a) give the lower bound in \cref{lem:local-states}; the red triangles in (b) identify the possible backward-arc sets when that bound is attained.}
\label{fig:local-tournament}
\end{figure}

\begin{lemma}\label{lem:local-states}
The minimum number of backward majority arcs for the profile in \cref{eq:local-votes} is $2$, and the only two orders attaining this minimum are $\omega_0=234516$ and $\omega_1=245163$.
\end{lemma}

\begin{proof}
The two green triangles $1\to3\to4\to1$ and $3\to5\to6\to3$ in \cref{fig:local-tournament}(a) are edge-disjoint, so every order has at least two backward arcs.
Suppose that an order has exactly two.
It then has one backward arc in each green triangle and no backward arc outside them.

The red triangles $1\to3\to5\to1$ and $3\to4\to6\to3$ in \cref{fig:local-tournament}(b) imply that at least one of $1\to3,3\to5$ is backward and at least one of $3\to4,6\to3$ is backward.
If $4\to1$ were the backward arc in the first green triangle, then $1\to3$ and $3\to4$ would be forward, so the two red triangles would require both $3\to5$ and $6\to3$ to be backward in the second green triangle, a contradiction.
Hence the backward arc in the first green triangle is either $1\to3$ or $3\to4$.
In the first case, the second red triangle can contain a backward arc only if $6\to3$ is backward; in the second case, the first red triangle can contain a backward arc only if $3\to5$ is backward.
Thus the complete backward-arc set is either $\{1\to3,6\to3\}$ or $\{3\to4,3\to5\}$.
If the backward arcs are $1\to3$ and $6\to3$, the forward arcs $2\to3\to4\to5\to1\to6$ force the unique order $234516=\omega_0$.
If they are $3\to4$ and $3\to5$, the forward arcs $2\to4\to5\to1\to6\to3$ force the unique order $245163=\omega_1$.
Both orders have exactly two backward arcs.
Conversely, the preceding argument shows that every order with exactly two backward arcs has one of these two complete backward-arc sets, and the remaining forward arcs determine the order uniquely.
Thus no other order attains the lower bound.
\end{proof}

Set $M=4(n+m+1)^2$.
For each $v_i\in V$, create six mutually disjoint sets of new candidates $B_{i,1},\dots,B_{i,6}$ with $|B_{i,1}|=M+2$ and $|B_{i,2}|=\cdots=|B_{i,6}|=M$.
Call these sets the \emph{vertex blocks} for $v_i$, and write $\mathcal B_i=(B_{i,1},\dots,B_{i,6})$ for the resulting ordered list of six vertex blocks.
In input ranking $\pi_r$, the relative order of these six blocks will be $\nu_r$ from \cref{eq:local-votes}; hence the majority direction between every pair of vertex blocks is the same as that between the corresponding local candidates in the profile.
The two orders $\omega_0$ and $\omega_1$ from \cref{lem:local-states} will represent $x_i=0$ and $x_i=1$, respectively.

When the six vertex blocks are ordered by $\omega_0$, the backward majority arcs between them correspond to $1\to3$ and $6\to3$; under $\omega_1$, they correspond to $3\to4$ and $3\to5$.
For each such arc $a\to b$, all $|B_{i,a}|\cdot|B_{i,b}|$ candidate pairs with one member in each block are backward.
Consequently,
\begin{equation}\label{eq:switch-costs}
\begin{array}{c|c}
\text{order of the six blocks}&\text{number of backward arcs with endpoints in different blocks}\\ \hline
\omega_0&(M+2)M+M^2=2M^2+2M\\
\omega_1&M^2+M^2=2M^2
\end{array}
\end{equation}
Thus the six vertex blocks for $v_i$ contribute $2M^2+2M-2Mx_i$ when ordered by $\omega_{x_i}$, and their contribution depending on $x_i$ is exactly $-2Mx_i$.

Fix $X\subseteq V$ and suppose that, for every $v_i$, the six vertex blocks are ordered as $\omega_{x_i}$.
By \cref{eq:switch-costs}, the part of their total backward-arc contribution that depends on $X$ is $-2M\sum_{i=1}^n x_i$.
The main $X$-dependent contribution of the edge candidates will be $M\sum_{\{v_i,v_j\}\in E}x_ix_j$; the possible additional contribution is bounded later.
Together, these two displayed terms equal $-M\cut_G(X)/2$ by \cref{eq:cut-identity}.
For every edge $e=\{v_i,v_j\}$, with $i<j$, create one additional candidate $p_e$, called the \emph{edge candidate} for $e$.
We also add padding blocks and then complete the definitions of the three rankings.

For each $v_i$, also create six \emph{padding blocks} $D_{i,1},\dots,D_{i,6}$, each of size $M$.
The six vertex blocks and six padding blocks form the \emph{vertex group} for $v_i$.
Candidates in vertex blocks or padding blocks are called \emph{block candidates}.
The padding blocks do not encode $x_i$.
They will ensure that the gap chosen for an edge candidate in the third ranking lies strictly between the six vertex blocks of its two endpoints.

For finite sequences of blocks or candidates, juxtaposition denotes concatenation.
To number the possible positions of the edge candidates in the third ranking, define a block sequence $\mathcal Q$ as follows.
For $i=1,\dots,n$ in increasing order, append
\[
B_{i,3}B_{i,6}B_{i,1}B_{i,2}B_{i,5}B_{i,4}D_{i,1}\cdots D_{i,6}.
\]
Write the resulting sequence as $\mathcal Q=Q_1Q_2\cdots Q_{12n}$.
For each $i$, the first six blocks listed for $v_i$ are the blocks of $\mathcal B_i$ in the reverse of $\nu_3$, and the next six are $D_{i,1},\dots,D_{i,6}$.
Thus the vertex groups occur in increasing index order in $\mathcal Q$.

A \emph{block gap} of $\mathcal Q$ is a position before the first block, between consecutive blocks, or after the last block.
Number the block gaps by the number of blocks preceding them: block gap $0$ is before $Q_1$, block gap $r$ is after $Q_r$, and block gap $12n$ is after $Q_{12n}$.

For $e=\{v_i,v_j\}$ with $i<j$, the block gap immediately after the six blocks of $\mathcal B_i$ has index $12i-6$, and the block gap immediately before the six blocks of $\mathcal B_j$ has index $12j-12$.
Their average is
\begin{equation}\label{eq:midpoint}
\mu_e=\frac{(12i-6)+(12j-12)}{2}=6(i+j)-9.
\end{equation}
The formula shows that $\mu_e$ is an integer, and its distance from either endpoint-gap index is $6(j-i)-3>0$.
Thus $\mu_e$ is the index of the midpoint block gap of $\mathcal Q$ between the two endpoint block gaps.
When $j=i+1$, the six padding blocks following $\mathcal B_i$ make the endpoint gaps six positions apart; without those padding blocks, the two gaps would coincide.

Fix an arbitrary internal order of the candidates in every block and call it the \emph{forward order} of that block.
We now define the three input rankings $\pi_1$, $\pi_2$, and $\pi_3$.

For each $i$, let $P_i^{\mathrm L}$ denote the list of edge candidates $p_e$ for edges $e=\{v_i,v_j\}$ with $i<j$, ordered by increasing $j$.
Let $P_i^{\mathrm R}$ denote the list of edge candidates $p_e$ for edges $e=\{v_h,v_i\}$ with $h<i$, ordered by increasing $h$.
Either list may be empty; an empty list contributes the empty sequence.
Define $\pi_1$ by concatenating the following sequence for each $i=1,\dots,n$:
\[
B_{i,1}B_{i,2}B_{i,3}P_i^{\mathrm L}B_{i,4}B_{i,5}B_{i,6}D_{i,1}\cdots D_{i,6}.
\]
Define $\pi_2$ by concatenating the following sequence for each $i=1,\dots,n$:
\[
B_{i,6}B_{i,3}B_{i,2}P_i^{\mathrm R}B_{i,5}B_{i,4}B_{i,1}D_{i,1}\cdots D_{i,6}.
\]
Within every vertex block and padding block, both rankings use the forward order.

The third ranking reverses the block sequence $\mathcal Q$.
Consequently, it lists the vertex groups in decreasing index order and places the blocks of each $\mathcal B_i$ in the order $\nu_3$.
For each $r\in\{0,\dots,12n\}$, let $P_r^{\mathrm{mid}}$ denote the edge candidates $p_e$ for which $e=\{v_i,v_j\}$, $i<j$, and $\mu_e=r$. Order this list lexicographically by the index pair $(i,j)$; the list may be empty.
For $r\geq1$, place $P_r^{\mathrm{mid}}$ immediately before $Q_r$ in the reversed block sequence, and place $P_0^{\mathrm{mid}}$ at the end.
Thus
\[
\pi_3=P_{12n}^{\mathrm{mid}}Q_{12n}P_{12n-1}^{\mathrm{mid}}Q_{12n-1}\cdots P_1^{\mathrm{mid}}Q_1P_0^{\mathrm{mid}}.
\]
Within each vertex block and padding block in this definition, $\pi_3$ uses the reverse of the forward order.
The blocks lying after the block gap with index $\mu_e$ in $\mathcal Q$ precede $p_e$ in $\pi_3$, while the blocks lying before that block gap follow $p_e$.
This completes the definitions of the three rankings.
For every graph vertex $v_i$, deleting the edge candidates leaves the blocks of $\mathcal B_i$ in the orders $\nu_1$, $\nu_2$, and $\nu_3$ under $\pi_1$, $\pi_2$, and $\pi_3$, respectively.

The next subsection verifies that every candidate pair is split $2$\nobreakdash-to\nobreakdash-$1$.
The following subsection then defines the block orders that encode subsets $X\subseteq V$ and proves that every aggregate can be transformed to one of them without increasing the objective.
After this normalization, we determine the minimum contribution of each edge candidate.

\subsection[Pairwise relations and the all-2-to-1 property]{Pairwise relations and the all\nobreakdash-$2$\nobreakdash-to\nobreakdash-$1$ property}\label{sec:pairwise-relations}

Before analyzing aggregate orders, we prove the all\nobreakdash-$2$\nobreakdash-to\nobreakdash-$1$ condition stated in \cref{thm:main}.
Once this is established, \cref{eq:all-two-to-one-score} allows the rest of the reduction to use backward majority arcs in place of Kemeny scores.

\begin{lemma}\label{lem:all-two-to-one-profile}
Every unordered pair of candidates in the constructed profile is split exactly $2$\nobreakdash-to\nobreakdash-$1$.
\end{lemma}

\begin{proof}
There are three kinds of candidate pairs: two block candidates, one block candidate and one edge candidate, or two edge candidates.

Pairs inside one block are ordered forward in $\pi_1,\pi_2$ and backward in $\pi_3$.
For two different vertex blocks in one vertex group, the assertion follows from \cref{eq:local-votes}.
For every other pair of blocks, the first two rankings agree on their order and the third ranking reverses it.
This covers two blocks from different vertex groups, a vertex block and a padding block in the same group, and two padding blocks in the same group.

We next consider a pair consisting of one edge candidate and one block candidate.
Fix an edge $e=\{v_i,v_j\}$ with $i<j$. We consider the block candidates in order of their vertex-group indices, from groups preceding $v_i$ through groups following $v_j$.
By \cref{eq:midpoint}, $12i-6<\mu_e<12j-12$.
For a block in a vertex group with index smaller than $i$, both $\pi_1$ and $\pi_2$ place the block before $p_e$, whereas $\pi_3$ places $p_e$ before the block.
For each of the six vertex blocks in $\mathcal B_i$, the ranking $\pi_2$ places the block before $p_e$, whereas $\pi_3$ places $p_e$ before the block. Hence $\pi_1$ determines the majority orientation of each such pair.
For every padding block in the group for $v_i$ and every block in a group indexed strictly between $i$ and $j$, the ranking $\pi_1$ places $p_e$ before the block, while $\pi_2$ places the block before $p_e$; the position in $\pi_3$ determines the majority orientation.
For each of the six vertex blocks in $\mathcal B_j$, the ranking $\pi_1$ places $p_e$ before the block, whereas $\pi_3$ places the block before $p_e$. Hence $\pi_2$ determines the majority orientation of each such pair.
For every padding block in the group for $v_j$ and every block in a group with index larger than $j$, both $\pi_1$ and $\pi_2$ place $p_e$ before the block, whereas $\pi_3$ places the block before $p_e$.
These cases prove that every pair consisting of an edge candidate and a block candidate is split $2$\nobreakdash-to\nobreakdash-$1$.

It remains to consider pairs of edge candidates. Let $e=\{v_i,v_j\}$ and $f=\{v_k,v_\ell\}$ be distinct edges, written with $i<j$ and $k<\ell$.
The first ranking orders the two edge candidates lexicographically by $(i,j)$, while the second orders them lexicographically by $(j,i)$.
If these two rankings disagree, the third ranking makes one orientation occur twice and the other once, so the pair is split $2$\nobreakdash-to\nobreakdash-$1$.
Suppose both place $p_e$ before $p_f$.
Then $i\leq k$ and $j\leq\ell$, with at least one strict inequality, so $\mu_e<\mu_f$.
In the third ranking, every candidate in $P_r^{\mathrm{mid}}$ precedes every candidate in $P_s^{\mathrm{mid}}$ when $r>s$; hence $p_f$ precedes $p_e$.
The case in which the first two rankings both place $p_f$ before $p_e$ is symmetric.
If $\mu_e=\mu_f$, the equality $i+j=k+\ell$ for distinct endpoint pairs forces the first two lexicographic orders to disagree.
Thus every pair of edge candidates is also split $2$\nobreakdash-to\nobreakdash-$1$.
\end{proof}

\subsection{Normalization of aggregate orders}\label{sec:fixed-normalization}

To associate an aggregate with a subset $X\subseteq V$, we require its blocks to have the following form: every block is contiguous, the vertex groups occur in increasing index order, the six vertex blocks precede the padding blocks within each group, and those six vertex blocks use either $\omega_0$ or $\omega_1$. Under these conditions, define $X=\{v_i:\text{the six vertex blocks for }v_i\text{ use }\omega_1\}$.
An arbitrary aggregate need not have this form: it may split a block, interleave vertex groups, place padding blocks before vertex blocks, or order the six vertex blocks differently from both $\omega_0$ and $\omega_1$. We now give a polynomial-time transformation that does not increase the number of backward majority arcs and produces an order satisfying these four conditions. It first makes every block contiguous and forward-ordered, and then changes the relative order of the blocks.

Making a block contiguous uses two properties of the input rankings.
Each block is consecutive in every ranking, so every candidate outside it lies entirely before or entirely after it in that ranking; all members of the block therefore have the same majority comparison with the outside candidate.
Inside a block, the first two rankings use the forward order and the third uses the reverse order, so every majority arc between two members of the block points forward in the forward order.
These two properties allow each block to be made contiguous and internally forward-ordered without increasing the number of backward arcs.

\begin{lemma}\label{lem:block-normalization}
Every aggregate order can be changed in polynomial time, without increasing the number of backward majority arcs, so that each block is contiguous and uses its forward order.
\end{lemma}

\begin{proof}
Fix one block $Q$ and temporarily fix the order of all candidates outside $Q$.
Every member of $Q$ has the same majority comparison with every outside candidate.
For each gap $\gamma$ of the fixed outside order, let $f(\gamma)$ be the number of backward arcs between the outside candidates and one member of $Q$ inserted at $\gamma$.
If the members of $Q$ occupy gaps $\gamma_1,\dots,\gamma_{|Q|}$, the total number of such backward arcs is $\sum_{r=1}^{|Q|}f(\gamma_r)$.
Moving all members to a gap minimizing $f$ cannot increase that number, and ordering them in the forward order makes every majority arc inside $Q$ forward.

Apply this operation successively to all blocks.
A minimizing gap can be chosen not to split a block $R$ treated in an earlier iteration.
While one member of $Q$ is moved past the consecutive members of $R$, its contribution changes by the same signed amount at every step because all members of $R$ have the same majority comparison with it.
Thus $f$ is monotone across the internal gaps of $R$.
If $f$ attains its minimum at a gap inside $R$, then one of the two gaps immediately outside $R$ has the same minimum value.
Repeating this adjustment for every block treated in an earlier iteration yields a minimizing gap that splits none of them. Those blocks can therefore remain contiguous throughout the procedure.
All scans, moves, and internal sorts are polynomial in the number of candidates.
\end{proof}

Once the blocks are contiguous, placing two blocks opposite their common majority direction creates at least $M^2$ backward arcs.
To see this, suppose every candidate in a block $Q$ has a majority arc to every candidate in a block $R$, but an aggregate places $R$ before $Q$. Then all $|Q|\cdot|R|$ arcs with one endpoint in each block are backward.
Since every block has at least $M$ candidates, this number is at least $M^2$.

We next define the block order $\Gamma_X$ for each subset $X\subseteq V$ and then show that every aggregate can be transformed to one of these block orders.
After applying \cref{lem:block-normalization}, every vertex block and padding block is contiguous and internally forward-ordered.
Define $\Gamma_X$ as follows.
Process $i=1,\dots,n$ in increasing order; within the vertex group for $v_i$, arrange the six vertex blocks $B_{i,1},\dots,B_{i,6}$ according to $\omega_{x_i}$ and then list $D_{i,1},\dots,D_{i,6}$, where $x_i=\chi_X(v_i)$.
The definitions of the three input rankings imply the following candidate-level majority relations. Every candidate in a smaller-index vertex group has a majority arc to every candidate in a larger-index vertex group. Within one vertex group, every candidate in a vertex block has a majority arc to every candidate in a padding block, and, for $r<s$, every candidate in $D_{i,r}$ has a majority arc to every candidate in $D_{i,s}$. All of these arcs are forward in $\Gamma_X$.
By \cref{eq:switch-costs}, the number of backward arcs between block candidates in $\Gamma_X$ is
\begin{equation}\label{eq:fixed-cost}
n(2M^2+2M)-2M\sum_{i=1}^{n}x_i.
\end{equation}
The first term is the total contribution from the six vertex blocks when every vertex group uses $\omega_0$; every other arc between block candidates is forward in $\Gamma_X$.

\begin{lemma}\label{lem:fixed-normalization}
Every aggregate order can be changed in polynomial time, without increasing the number of backward majority arcs, so that every vertex block and padding block is contiguous and forward-ordered and deleting the edge candidates leaves the block order $\Gamma_X$ for some $X\subseteq V$.
\end{lemma}

\begin{proof}
First apply \cref{lem:block-normalization}.
There are $n(12M+2)$ block candidates.
Temporarily delete the edge candidates.
We first change the remaining block sequence into the form $\Gamma_X$ and count the decrease in backward arcs whose endpoints lie in different blocks. We then reinsert the edge candidates and bound how many arcs incident with an edge candidate can change from forward to backward.

For every $i<j$, every candidate in the vertex group for $v_i$ has a majority arc to every candidate in the vertex group for $v_j$.
Reorder the blocks by nondecreasing vertex-group index while preserving the relative order of blocks that belong to the same group. This makes every majority arc between different vertex groups forward.
Within each group, move the six vertex blocks before the padding blocks without changing their internal order, and then sort the padding blocks as $D_{i,1},\dots,D_{i,6}$.
These operations also make all affected majority arcs forward.
If any of these operations changes the block sequence, then it changes at least one pair of blocks from the reverse of their common majority direction to that majority direction. It therefore decreases the number of backward arcs between block candidates by at least $M^2$.

Only the order of the six vertex blocks in each group remains to be fixed.
Retain a group whose six vertex blocks are already ordered as $\omega_0$ or $\omega_1$; otherwise reorder them as $\omega_1$.
By \cref{lem:local-states}, the previous order has at least three backward majority arcs among the six labels and therefore contributes at least $3M^2$ backward arcs between the corresponding blocks, whereas $\omega_1$ contributes exactly $2M^2$.
This replacement changes no comparison with a block outside the group and decreases the number of backward arcs within the group by at least $M^2$.
Define $X=\{v_i:\text{the six vertex blocks for }v_i\text{ use }\omega_1\}$.
The resulting block sequence is then $\Gamma_X$.

If the block sequence obtained after \cref{lem:block-normalization} already equals this $\Gamma_X$, leave the current aggregate unchanged.
Otherwise, the preceding operations have decreased the number of backward arcs between block candidates by at least $M^2$; use the new block sequence $\Gamma_X$ and reinsert all edge candidates, in their previous relative order, in the gap before the first block.
The comparisons between edge candidates are then unchanged.
There are $m$ edge candidates and $n(12M+2)$ block candidates, so at most $mn(12M+2)$ arcs between an edge candidate and a block candidate can change from forward to backward.
Since $m=2n$, we have $M=4(3n+1)^2$ and $mn=2n^2<2(3n+1)^2/9=M/18$.
Therefore $mn(12M+2)<\frac{M}{18}(12M+2)=\frac23M^2+\frac19M<M^2$.
Thus the possible increase in backward arcs incident with edge candidates is smaller than the decrease of at least $M^2$ among block candidates.
The total number of backward majority arcs does not increase.
All block reorderings and edge-candidate reinsertions are polynomial-time computable.
\end{proof}

\subsection{Insertion cost of an edge candidate}\label{sec:insertion-position}

After applying \cref{lem:fixed-normalization}, deleting the edge candidates leaves the block order $\Gamma_X$, but the edge candidates themselves may still occur anywhere in the aggregate.
For each edge candidate, we now determine the minimum number of backward arcs between it and the block candidates.
A \emph{block gap} is a position before the first block of $\Gamma_X$, between two consecutive blocks, or after the last block.
For an edge $e$ and a block gap $\gamma$, let $\Phi_e^X(\gamma)$ be the number of backward arcs between $p_e$ and the block candidates when $p_e$ is inserted at $\gamma$, and put $\phi_e(X)=\min_\gamma\Phi_e^X(\gamma)$.
The proof analyzes the block gaps from left to right.
Arcs between two edge candidates are bounded separately in \cref{sec:score}.

\begin{lemma}\label{lem:insertion-contribution}
For every edge $e=\{v_i,v_j\}$ there is a polynomial-time computable integer $c_e$ such that, for every $X\subseteq V$, $c_e+Mx_ix_j\leq\phi_e(X)\leq c_e+Mx_ix_j+2$.
\end{lemma}

\begin{proof}
Fix an edge $e=\{v_i,v_j\}$ with $i<j$, put $k_e=j-i$, and keep the block order $\Gamma_X$ fixed.
Every member of a block $Q$ has the same majority relation with $p_e$.
Thus, if $p_e$ is moved one position at a time through $Q$, the number of backward arcs between $p_e$ and the members of $Q$ changes in the same direction at every step.
Consequently, allowing positions inside blocks cannot improve on the minimum over block gaps.
Moving $p_e$ from immediately before a block $Q$ to immediately after it changes $\Phi_e^X$ by $-|Q|$ if the majority arc is $Q\to p_e$ and by $+|Q|$ if the majority arc is $p_e\to Q$.

We consider five consecutive ranges of block gaps: before $\mathcal B_i$, while $p_e$ is moved past the six blocks of $\mathcal B_i$, between $\mathcal B_i$ and $\mathcal B_j$, while $p_e$ is moved past the six blocks of $\mathcal B_j$, and after $\mathcal B_j$.
We compute the minimum insertion cost attained while $p_e$ passes the blocks of $\mathcal B_i$ or $\mathcal B_j$ and show that every other block gap has larger cost.

Let $A_e$ be the insertion cost immediately before the first block of $\mathcal B_i$.
At this gap, every block in a vertex group with index smaller than $i$ precedes $p_e$, and every block in a vertex group with index at least $i$ follows $p_e$.
This partition of the block candidates does not depend on the internal orders $\omega_{x_h}$, so $A_e$ is independent of $X$.
For every block in a group with index smaller than $i$, both $\pi_1$ and $\pi_2$ place the block before $p_e$, whereas $\pi_3$ places $p_e$ before the block.
The majority arc therefore points from the block to $p_e$, so the insertion cost decreases as $p_e$ is moved from left to right through these groups and reaches $A_e$ immediately before $\mathcal B_i$.

At $v_i$, the rankings $\pi_2$ and $\pi_3$ disagree on every pair consisting of $p_e$ and a candidate in $\mathcal B_i$, so $\pi_1$ determines the majority relation.
Hence moving past $B_{i,1},B_{i,2},B_{i,3}$ decreases the insertion cost, whereas moving past $B_{i,4},B_{i,5},B_{i,6}$ increases it.
Measured from the gap immediately before $\mathcal B_i$, the minimum cumulative change and the total change after all six blocks are
\[
\begin{array}{c|cc}
\text{block order}&\text{minimum cumulative change}&\text{total change after six blocks}\\ \hline
\omega_0&-2M&-2\\
\omega_1&-M&-2
\end{array}
\]
These entries are obtained by reading the chosen block order from left to right and recording the cumulative sum of the signed block sizes.
For example, under $\omega_0$ at $v_i$, the cumulative changes are $-M,-2M,-M,0,-M-2,-2$.
Thus the minimum attained while $p_e$ is moved past $\mathcal B_i$ is $A_e-2M+Mx_i$, and the insertion cost immediately after $\mathcal B_i$ is $A_e-2$.

We next analyze the block gaps strictly between $\mathcal B_i$ and $\mathcal B_j$.
For each intervening block, $\pi_1$ places $p_e$ before the block and $\pi_2$ places the block before $p_e$, so $\pi_3$ determines the majority relation.
The midpoint index $\mu_e$ is used here only to determine the orientation in $\pi_3$.
Because $\pi_3$ reads $\mathcal Q$ from right to left, a block before the block gap with index $\mu_e$ follows $p_e$ in $\pi_3$, so moving $p_e$ past that block in $\Gamma_X$ increases the insertion cost.
A block after that block gap precedes $p_e$ in $\pi_3$, so moving past it decreases the cost.
There are $12k_e-6$ intervening blocks, and the definition of $\mu_e$ divides them into two sets of $6k_e-3$ blocks.

First treat every intervening block as having size $M$, temporarily ignoring the two extra candidates in blocks $B_{h,1}$.
The total increase contributed by the blocks before the block gap with index $\mu_e$ then equals the total decrease contributed by the blocks after that gap.

Suppose first that $k_e=2r+1$ for some $r\geq0$.
Then $\mu_e=12(i+r)-3$, which is the gap after $D_{i+r,3}$ in $\mathcal Q$.
The order $\Gamma_X$ lists the vertex groups by increasing index and places the six vertex blocks before the padding blocks in each group.
Therefore every intervening block whose passage increases the cost occurs before every intervening block whose passage decreases it.
Starting from $A_e-2$ immediately after $\mathcal B_i$, the insertion cost therefore first increases and then decreases back to $A_e-2$; it never falls below $A_e-2$.

Now suppose that $k_e=2r$ for some $r\geq1$.
Then $\mu_e=12(i+r)-9$, the gap after $B_{i+r,3},B_{i+r,6},B_{i+r,1}$ in the corresponding twelve-block sequence of $\mathcal Q$.
Before $p_e$ reaches the six vertex blocks of this intermediate group in $\Gamma_X$, it has passed $12r-6$ intervening blocks, all of which increase the cost.
Under this temporary assumption that every intervening block has size $M$, the insertion cost has therefore increased by $(12r-6)M\geq6M$.
Within this group, moving past $B_{i+r,3},B_{i+r,6},B_{i+r,1}$ increases the cost, whereas moving past $B_{i+r,2},B_{i+r,5},B_{i+r,4}$ decreases it.
In the order $\omega_0=234516$, the six successive changes have signs $-,+,-,-,+,+$ and minimum cumulative change $-2M$.
In the order $\omega_1=245163$, they have signs $-,-,-,+,+,+$ and minimum cumulative change $-3M$.
Thus throughout these six blocks the insertion cost remains at least $3M$ above $A_e-2$.
Every remaining intervening block decreases the cost.
Because the total increase and total decrease are equal under this assumption, the cost then returns to $A_e-2$ immediately before $\mathcal B_j$ and never falls below that value.

It remains to restore the two extra candidates in every intervening block $B_{h,1}$.
When $k_e=2r+1$, there are $r$ such blocks before the block gap with index $\mu_e$ and $r$ after it.
Their extra contributions cancel, and the cost-increasing extra contributions occur before the cost-decreasing ones in $\Gamma_X$.
When $k_e=2r$, there are $r$ such blocks before the block gap with index $\mu_e$ and $r-1$ after it; the unmatched block is $B_{i+r,1}$ before the gap.
Again, every cost-increasing extra contribution is encountered before every cost-decreasing one.
Consequently, the extra candidates cannot create an intervening value below $A_e-2$.
The insertion cost immediately before $\mathcal B_j$ is $A_e-2$ when $k_e$ is odd and $A_e$ when $k_e$ is even.
Set $\varepsilon_e=-2$ in the odd case and $\varepsilon_e=0$ in the even case, so this value is $A_e+\varepsilon_e$.

At $v_j$, the rankings $\pi_1$ and $\pi_3$ disagree on every pair consisting of $p_e$ and a candidate in $\mathcal B_j$, so $\pi_2$ determines the majority relation.
Hence moving past $B_{j,6},B_{j,3},B_{j,2}$ decreases the insertion cost, whereas moving past $B_{j,5},B_{j,4},B_{j,1}$ increases it.
Measured from the gap immediately before $\mathcal B_j$, the corresponding values are
\begin{equation}\label{eq:right-endpoint-costs}
\begin{array}{c|cc}
\text{block order}&\text{minimum cumulative change}&\text{total change after six blocks}\\ \hline
\omega_0&-2M&+2\\
\omega_1&-M&+2
\end{array}
\end{equation}
By \cref{eq:right-endpoint-costs}, the minimum attained while $p_e$ is moved past $\mathcal B_j$ is $A_e+\varepsilon_e-2M+Mx_j$.
Immediately after $\mathcal B_j$, the insertion cost is $A_e+\varepsilon_e+2$.
For every remaining block, both $\pi_1$ and $\pi_2$ place $p_e$ before the block, whereas $\pi_3$ places the block before $p_e$.
The majority arc points from $p_e$ to the block, so the insertion cost increases after $\mathcal B_j$.

The minimum attained while $p_e$ is moved past $\mathcal B_i$ is at most $A_e-M$.
Since $M>2$, this is smaller than every value before $\mathcal B_i$, between $\mathcal B_i$ and $\mathcal B_j$, or after $\mathcal B_j$.
Hence no block gap before $\mathcal B_i$, strictly between $\mathcal B_i$ and $\mathcal B_j$, or after $\mathcal B_j$ is globally minimizing.
A global minimum is attained at a block gap within $\mathcal B_i$ or $\mathcal B_j$.
Therefore
\begin{equation}\label{eq:global-insertion-minimum}
\phi_e(X)=A_e-2M+\min\{Mx_i,\varepsilon_e+Mx_j\}.
\end{equation}

Set $c_e=\phi_e(\emptyset)$, which can be computed by scanning the block gaps of $\Gamma_{\emptyset}$.
Subtracting $c_e$ from \cref{eq:global-insertion-minimum} gives
\[
\begin{array}{c|cccc}
(x_i,x_j)&(0,0)&(0,1)&(1,0)&(1,1)\\ \cline{2-5}
k_e\text{ even}&0&0&0&M\\
k_e\text{ odd}&0&2&0&M
\end{array}
\]
In the first three columns, $x_ix_j=0$ and the displayed difference is either $0$ or $2$; in the last column, $x_ix_j=1$ and the difference is $M$.
This proves the inequalities in the lemma.
\end{proof}

\subsection{Score bounds and proof of the main theorem}\label{sec:score}

We now combine the three disjoint parts of the backward-arc count.
The contribution from candidate pairs belonging to two blocks is given exactly by \cref{eq:fixed-cost}.
For each edge candidate, \cref{lem:insertion-contribution} places its minimum contribution against the block candidates between $c_e+Mx_ix_j$ and $c_e+Mx_ix_j+2$.
The only remaining arcs are those between two edge candidates; there are $\binom m2$ such pairs, so their total contribution lies between $0$ and $\binom m2$.
Define
\[
\beta=n(2M^2+2M)+\sum_{e\in E}c_e,
\qquad
\eta=2m+\binom m2.
\]
The constant term $\beta$ consists of the exact block-candidate contribution at $X=\emptyset$ and the constants $c_e$ from the edge-candidate bounds; it does not include contributions between edge candidates.
In the aggregate constructed for the upper bound, placing each edge candidate in a block gap minimizing its contribution against the block candidates makes the total insertion-cost excess at most $2m$, while the contribution between edge candidates is at most $\binom m2$; their sum is bounded by $\eta$.

For $X\subseteq V$, let $b^\star(X)$ be the minimum of $b_\Pi(\sigma)$ over aggregates $\sigma$ in which every vertex block and padding block is contiguous and forward-ordered and for which deleting the edge candidates leaves $\Gamma_X$.
By \cref{eq:fixed-cost} and the lower bound in \cref{lem:insertion-contribution}, every such aggregate $\sigma$ satisfies
\begin{equation}\label{eq:score-lower}
b_\Pi(\sigma)
\geq
\beta-2M\sum_{i=1}^n x_i+M\sum_{\{v_i,v_j\}\in E}x_ix_j.
\end{equation}
This is a lower bound because each edge candidate contributes at least its minimum against the block candidates, and the number of backward arcs between edge candidates is nonnegative.

Conversely, place each edge candidate in a block gap minimizing its contribution against the block candidates.
Several edge candidates may occupy the same gap and may then be ordered arbitrarily.
The upper bound in \cref{lem:insertion-contribution} contributes at most $2m$ above $\sum_{e\in E}(c_e+Mx_ix_j)$, and at most $\binom m2$ arcs between edge candidates are backward.
Thus there is an aggregate whose block order is $\Gamma_X$ and
\begin{equation}\label{eq:score-upper}
b_\Pi(\sigma)
\leq
\beta-2M\sum_{i=1}^n x_i+M\sum_{\{v_i,v_j\}\in E}x_ix_j+\eta.
\end{equation}

Combining \cref{eq:score-lower,eq:score-upper,eq:cut-identity} gives
\begin{equation}\label{eq:cut-score-bounds}
\beta-\frac{M}{2}\cut_G(X)
\leq
b^\star(X)
\leq
\beta-\frac{M}{2}\cut_G(X)+\eta.
\end{equation}
The choice of $M$ gives $\eta<M/2$:
\begin{equation}\label{eq:error-bound}
\eta=2m+\binom m2=2n^2+3n<2(3n+1)^2=\frac{M}{2}.
\end{equation}
Consequently, for every integer $k\geq1$,
\[
\beta-\frac M2k+\eta<\beta-\frac M2(k-1).
\]
Together with \cref{eq:cut-score-bounds}, this implies that $b^\star(X)<b^\star(Y)$ whenever $\cut_G(X)=k$ and $\cut_G(Y)\leq k-1$.
In particular, every cut of size at least $t$ yields an aggregate with at most $\beta-Mt/2+\eta$ backward arcs, whereas the minimum for every cut of size at most $t-1$ is strictly larger.

Let $N=n(12M+2)+m$ be the total number of candidates and define the backward-arc threshold by $\theta=\beta-Mt/2+\eta$.
Because $M$ is even, $\theta$ is an integer.
Set the Kemeny threshold to $\kappa=\binom N2+\theta$.
Both thresholds are nonnegative.
This follows because $c_e\geq0$ for every edge, so $\beta\geq n(2M^2+2M)$; because $t\leq m=2n$,
\[
\theta
\geq n(2M^2+2M)-Mn+\eta
=n(2M^2+M)+\eta
>0.
\]

\begin{lemma}\label{lem:correctness}
The graph $G$ has a cut of size at least $t$ if and only if the constructed all\nobreakdash-$2$\nobreakdash-to\nobreakdash-$1$ profile has an aggregate order of Kemeny score at most $\kappa=\binom N2+\beta-Mt/2+\eta$.
\end{lemma}

\begin{proof}
If $G$ has a cut of size at least $t$, choose $X\subseteq V$ with $\cut_G(X)\geq t$.
The upper bound in \cref{eq:cut-score-bounds} gives an aggregate with
\[
b_\Pi(\sigma)
\leq
\beta-\frac{M}{2}\cut_G(X)+\eta
\leq
\theta.
\]
By \cref{eq:all-two-to-one-score}, its Kemeny score is at most $\kappa$.

Conversely, suppose an aggregate $\sigma$ has Kemeny score at most $\kappa$.
By \cref{eq:all-two-to-one-score}, it has at most $\theta$ backward majority arcs.
Apply \cref{lem:fixed-normalization} to obtain an aggregate $\sigma'$ with $b_\Pi(\sigma')\leq b_\Pi(\sigma)$ such that deleting the edge candidates leaves $\Gamma_X$ for some $X\subseteq V$.
The lower bound in \cref{eq:cut-score-bounds} gives
\[
b_\Pi(\sigma')\geq\beta-\frac{M}{2}\cut_G(X).
\]
If $\cut_G(X)\leq t-1$, then
\[
b_\Pi(\sigma')
\geq
\beta-\frac{M}{2}t+\frac{M}{2}
>
\beta-\frac{M}{2}t+\eta
=
\theta,
\]
where the strict inequality follows from \cref{eq:error-bound}.
This contradicts $b_\Pi(\sigma')\leq b_\Pi(\sigma)\leq\theta$.
Therefore $\cut_G(X)\geq t$.
\end{proof}

It remains to verify that the construction is polynomial and to assemble the preceding lemmas.
The number of candidates is
\begin{equation}\label{eq:size}
N=n(12M+2)+m=O(n^3),
\end{equation}
because $m=2n$ and $M=4(n+m+1)^2=4(3n+1)^2=O(n^2)$.
Each of the three rankings is given explicitly as a list of the $N$ candidates and can be generated in polynomial time.
Each constant $c_e$ can be computed by scanning the block gaps of $\Gamma_{\emptyset}$.
Thus $\beta$, $\theta$, and $\kappa$ are computable in polynomial time.

\begin{proof}[Proof of \cref{thm:main}]
\Cref{lem:four-regular} supplies the NP-complete source problem.
The construction is polynomial by \cref{eq:size}; \cref{lem:all-two-to-one-profile} proves the all\nobreakdash-$2$\nobreakdash-to\nobreakdash-$1$ condition, and \cref{lem:correctness} proves that the source instance and the constructed instance have the same answer.
The local tournament in \cref{fig:local-tournament} contains a directed triangle.
Since the supplied three rankings realize the majority tournament and every tournament realizable by at most two rankings is transitive, the constructed tournament has majority dimension exactly $3$.
The restricted \KemenyScore{} problem is in NP: given an aggregate order, its Kemeny score and the all\nobreakdash-$2$\nobreakdash-to\nobreakdash-$1$ condition can be checked in polynomial time.
Therefore it is NP-complete.
\end{proof}

\subsection{Exact value and solution recovery}\label{sec:exact-recovery}

Let $\Pi_G=(\pi_1,\pi_2,\pi_3)$ denote the profile constructed from $G$.
The preceding argument proves that $\maxcutvalue(G)\geq t$ exactly when $\Kemeny_{\Pi_G}^{\star}\leq\kappa$, which is the decision equivalence in \cref{thm:main}.
The inequalities in \cref{eq:cut-score-bounds,eq:error-bound} separate the minimum backward-arc counts associated with different cut sizes.
Because the three constructed rankings depend only on $G$, while $t$ enters only through the Kemeny threshold, the optimal Kemeny score determines the exact maximum-cut value, and any Kemeny-optimal aggregate can be normalized to recover a maximum cut.
We record these stronger conclusions separately for use in the consequences derived later.
By \cref{eq:size}, the profile has polynomial size. The construction computes every constant $c_e$, and hence every quantity in the recovery formula below, in polynomial time.

\begin{proposition}\label{prop:maxcut-recovery}
Let $N$ be the number of candidates of $\Pi_G$ and let $\Kemeny_{\Pi_G}^{\star}$ be its optimal Kemeny score.
Then
\begin{equation}
\label{eq:maxcut-recovery}
\maxcutvalue(G)
=
\left\lceil
\frac{2}{M}
\left(
\beta+\binom N2-\Kemeny_{\Pi_G}^{\star}
\right)
\right\rceil.
\end{equation}
Moreover, a maximum cut of $G$ can be recovered in polynomial time from any Kemeny-optimal aggregate of $\Pi_G$.
\end{proposition}

\begin{proof}
Put $b^{\star}=\Kemeny_{\Pi_G}^{\star}-\binom N2$.
By \cref{eq:all-two-to-one-score}, $b^{\star}$ is the minimum number of backward majority arcs.
Choose $X^{\star}$ with $\cut_G(X^{\star})=\maxcutvalue(G)$.
The upper bound in \cref{eq:cut-score-bounds} gives
\[
b^{\star}
\leq
\beta-\frac M2\maxcutvalue(G)+\eta.
\]
Conversely, apply the transformation in \cref{lem:fixed-normalization} to a Kemeny-optimal aggregate.
The resulting aggregate remains optimal, and deleting its edge candidates leaves $\Gamma_X$ for some $X\subseteq V$.
The lower bound in \cref{eq:cut-score-bounds} gives
\[
b^{\star}
\geq
\beta-\frac M2\cut_G(X)
\geq
\beta-\frac M2\maxcutvalue(G).
\]
Consequently,
\[
\maxcutvalue(G)-\frac{2\eta}{M}
\leq
\frac{2(\beta-b^{\star})}{M}
\leq
\maxcutvalue(G).
\]
Since $2\eta/M<1$ by \cref{eq:error-bound}, taking the ceiling and substituting the definition of $b^{\star}$ proves \cref{eq:maxcut-recovery}.

It remains to recover a maximum cut.
Apply the polynomial-time transformation in \cref{lem:fixed-normalization} to the supplied Kemeny-optimal aggregate, and define
$X=\{v_i:\text{the six vertex blocks for }v_i\text{ use }\omega_1\}$
from the resulting block order.
If $\cut_G(X)\leq\maxcutvalue(G)-1$, then the lower bound gives
\[
b^{\star}
\geq
\beta-\frac M2\maxcutvalue(G)+\frac M2,
\]
whereas $X^{\star}$ gives
\[
b^{\star}
\leq
\beta-\frac M2\maxcutvalue(G)+\eta.
\]
This contradicts $\eta<M/2$.
Thus $\cut_G(X)=\maxcutvalue(G)$, so $X$ determines a maximum cut of $G$.
\end{proof}

\section{Three-ranking winner, precedence, and recognition problems}\label{sec:other-base-reductions}

This section proves the winner, precedence, and recognition results in \cref{thm:three-ranking-optimality}.
The reductions apply the construction from \cref{sec:construction} to simple graphs that need not be $4$-regular.
We first state the construction for an arbitrary simple graph and record the properties used later.

\subsection{The construction for an arbitrary simple graph}\label{sec:general-block-construction}

Let $F=(V,E)$ be a nonempty simple graph with $V=\{v_1,\dots,v_n\}$ and $m=|E|$.
For each vertex $v_i$, create six vertex blocks $B_{i,1},\dots,B_{i,6}$ and six padding blocks $D_{i,1},\dots,D_{i,6}$.
The sizes of these blocks will be specified in each reduction.
Write $\mathcal B_i=(B_{i,1},\dots,B_{i,6})$ for the ordered list of these six vertex blocks.
For each edge $e=\{v_i,v_j\}$ with $i<j$, create a nonempty edge-candidate block $P_e$.
A block with one member is identified with that member.
Fix a forward order of every block.
The first two rankings use the forward order inside every block, and the third ranking uses its reverse.

For each $i$, let $P_i^{\mathrm L}$ be the concatenation of the blocks $P_{\{v_i,v_j\}}$ with $i<j$, ordered by increasing $j$.
Let $P_i^{\mathrm R}$ be the concatenation of the blocks $P_{\{v_h,v_i\}}$ with $h<i$, ordered by increasing $h$.
Either sequence may be empty.
Define $\pi_1^F$ by concatenating, for $i=1,\dots,n$, the sequence
\[
B_{i,1}B_{i,2}B_{i,3}P_i^{\mathrm L}
B_{i,4}B_{i,5}B_{i,6}D_{i,1}\cdots D_{i,6}.
\]
Define $\pi_2^F$ by concatenating, for $i=1,\dots,n$, the sequence
\[
B_{i,6}B_{i,3}B_{i,2}P_i^{\mathrm R}
B_{i,5}B_{i,4}B_{i,1}D_{i,1}\cdots D_{i,6}.
\]

To define $\pi_3^F$, form the block sequence $\mathcal Q$ by concatenating, for $i=1,\dots,n$, the sequence
\[
B_{i,3}B_{i,6}B_{i,1}B_{i,2}B_{i,5}B_{i,4}D_{i,1}\cdots D_{i,6}.
\]
Write $\mathcal Q=Q_1\cdots Q_{12n}$ and number its block gaps by the number of preceding blocks, as in \cref{sec:construction}.
For $e=\{v_i,v_j\}$ with $i<j$, put
\begin{equation}\label{eq:general-midpoint}
\mu_e=6(i+j)-9.
\end{equation}
For each $r\in\{0,\dots,12n\}$, let $P_r^{\mathrm{mid}}$ be the concatenation of the edge-candidate blocks $P_e$ with $\mu_e=r$, ordered lexicographically by the endpoint pair $(i,j)$.
Define
\[
\pi_3^F
=
P_{12n}^{\mathrm{mid}}Q_{12n}
P_{12n-1}^{\mathrm{mid}}Q_{12n-1}
\cdots
P_1^{\mathrm{mid}}Q_1P_0^{\mathrm{mid}}.
\]
Let $\Pi_F=(\pi_1^F,\pi_2^F,\pi_3^F)$.
This is the construction in \cref{sec:construction}, except that the graph need not be regular and an edge candidate may be replaced by a block.

For $r\in\{0,1\}$, let $\omega_r(\mathcal B_i)$ denote the block sequence obtained by replacing each label $h$ in the order $\omega_r$ by the block $B_{i,h}$.
Thus
\[
\omega_0(\mathcal B_i)=B_{i,2}B_{i,3}B_{i,4}B_{i,5}B_{i,1}B_{i,6},
\qquad
\omega_1(\mathcal B_i)=B_{i,2}B_{i,4}B_{i,5}B_{i,1}B_{i,6}B_{i,3}.
\]
For $X\subseteq V$, write $x_i=\chi_X(v_i)$.
Define $\Gamma_X$ by concatenating, for $i=1,\dots,n$, the sequence
\begin{equation}\label{eq:general-block-order}
\omega_{x_i}(\mathcal B_i)D_{i,1}\cdots D_{i,6}.
\end{equation}

\begin{lemma}\label{lem:general-block-construction}
The profile $\Pi_F$ has the following properties.
\begin{enumerate}[label=\textup{(\roman*)}]
\item Every candidate pair is split $2$\nobreakdash-to\nobreakdash-$1$.
\item Its majority tournament contains a directed triangle and has majority dimension exactly $3$.
\item Each vertex block, padding block, and edge-candidate block can be made contiguous and internally forward-ordered without increasing the number of backward majority arcs.
\item Suppose each edge-candidate block is a singleton $P_e=\{p_e\}$, let every block $B_{i,1}$ have size $M+2$, and let every other vertex or padding block have size $M$, where $M$ is even and satisfies
\[
mn(12M+2)<M^2.
\]
Every aggregate order can then be transformed in polynomial time, without increasing the number of backward majority arcs, so that every vertex and padding block is contiguous and internally forward-ordered and deleting the edge candidates leaves $\Gamma_X$ for some $X\subseteq V$.
Moreover, for an edge $e$ and a block gap $\gamma$ of $\Gamma_X$, let $\Phi_e^X(\gamma)$ denote the number of backward majority arcs between $p_e$ and the block candidates when $p_e$ is inserted at $\gamma$, and put $\phi_e(X)=\min_\gamma\Phi_e^X(\gamma)$.
There is a polynomial-time computable integer $c_e$ such that
\[
c_e+Mx_ix_j
\leq
\phi_e(X)
\leq
c_e+Mx_ix_j+2
\]
for every $e=\{v_i,v_j\}\in E$ and every $X\subseteq V$.
\end{enumerate}
\end{lemma}

\begin{proof}
The proof of \cref{lem:all-two-to-one-profile} uses only the relative positions of the blocks and edge candidates, not their sizes or the degrees of the graph.
Replacing an edge candidate by a block preserves its comparison with every outside candidate, and the prescribed internal orders split every pair inside the block $2$\nobreakdash-to\nobreakdash-$1$.
This proves (i).

The three rankings realize the majority tournament.
For any $i$, choose $x_r\in B_{i,r}$ for $r\in\{1,3,4\}$.
Their majority arcs are
\[
x_1\to x_3\to x_4\to x_1,
\]
so the tournament contains a directed triangle.
The supplied three rankings realize the tournament, while every tournament realizable by at most two rankings is transitive.
Hence the majority dimension is exactly $3$, proving (ii).

The proof of \cref{lem:block-normalization} applies to every block because its members have identical comparisons with every outside candidate and the majority relation inside the block is the transitive order given by its forward order.
This proves (iii).
Finally, the insertion-cost proof in \cref{lem:insertion-contribution} uses only the relative block positions and the displayed block sizes, so it gives the bounds in (iv).
The normalization proof in \cref{lem:fixed-normalization} uses regularity only to bound the possible increase involving the edge candidates by a quantity smaller than $M^2$.
The inequality in (iv) gives the same bound for $F$ and therefore gives the stated transformation.
The $4$-regularity assumption is used later only in the cut identity \cref{eq:cut-identity}.
This proves (iv).
\end{proof}

\subsection{Kemeny winner, unique winner, and precedence}\label{sec:winner-ranking}

\begin{lemma}\label{lem:strict-alpha-compare}
Given two graphs $G$ and $H$, deciding whether $\alpha(G)>\alpha(H)$ is $\ThetaTwoP$-complete under polynomial-time many-one reductions, where $\alpha$ denotes the maximum independent-set size.
\end{lemma}

\begin{proof}
Membership follows by asking in parallel, for every integer $r$ from $0$ to $\max\{|V(G)|,|V(H)|\}$, whether $G$ or $H$ has an independent set of size at least $r$.
The answers determine both independence numbers and hence their strict comparison.

For hardness, we use the $\ThetaTwoP$-complete problem \textsc{Min-Card-Vertex-Cover-Compare} used in the \KemenyWinner{} reduction of Hemaspaandra, Spakowski, and Vogel~\cite[Section~4]{HemaspaandraSpakowskiVogel2005}.
Its instances ask whether $\tau(G)\leq\tau(H)$, where $\tau$ is the minimum vertex-cover size.
We may assume that both graphs are nonempty by adding one isolated vertex to each.
Put $g=|V(G)|$ and $h=|V(H)|$, and let
\[
G'=G\mathbin{\dot\cup} I_h,
\qquad
H'=H\mathbin{\dot\cup} I_{g-1},
\]
where $I_r$ is an independent set of size $r$.
Since $\alpha(F)=|V(F)|-\tau(F)$,
\begin{align*}
\alpha(G')&=g-\tau(G)+h,\\
\alpha(H')&=h-\tau(H)+g-1.
\end{align*}
Consequently, $\alpha(G')>\alpha(H')$ if and only if $\tau(G)\leq\tau(H)$.
\end{proof}

\begin{proof}[Proof of \cref{thm:three-ranking-optimality}(i)]
We first prove membership.
The all\nobreakdash-$2$\nobreakdash-to\nobreakdash-$1$ condition is checkable in polynomial time.
For a profile on $N$ candidates, every Kemeny score lies between $0$ and $3\binom N2$.
Ask in parallel, for every integer $t$ in this range, whether there exists an aggregate of score at most $t$.
For \KemenyWinner{}, also ask whether there exists such an aggregate with $c$ first.
For \KemenyUniqueWinner{}, ask whether there exists such an aggregate with $c$ not first.
For \KemenyPossiblePrecedence{}, ask whether there exists such an aggregate with $c\before d$, and for \KemenyNecessaryPrecedence{} ask whether one exists with $d\before c$.
Every query belongs to NP because an aggregate order is a polynomial-size certificate.
The parallel answers determine the unconstrained and constrained minimum Kemeny scores.
The possible conditions hold when the corresponding constrained minimum score equals the unrestricted minimum score, while the necessary conditions hold when the minimum score under the opposite constraint is strictly larger.
All four problems therefore belong to $\ThetaTwoP$.

We prove hardness using \cref{lem:strict-alpha-compare}.
Let $(G,H)$ be an instance.
Take disjoint graphs
\[
A=G^1\mathbin{\dot\cup}G^2\mathbin{\dot\cup}\{a\},
\qquad
B=H^1\mathbin{\dot\cup}H^2\mathbin{\dot\cup}\{v,b\},
\]
where $G^1,G^2$ are copies of $G$, $H^1,H^2$ are copies of $H$, and $a,v,b$ are isolated inside their displayed graphs.
Let $J=A\vee B$ be their join, obtained by adding every edge with one endpoint in $A$ and the other in $B$.
An independent set of $J$ is contained entirely in $A$ or entirely in $B$, and
\begin{equation}\label{eq:branch-alpha}
\alpha(A)=2\alpha(G)+1,
\qquad
\alpha(B)=2\alpha(H)+2.
\end{equation}
If $\alpha(G)>\alpha(H)$, then $\alpha(A)>\alpha(B)$, so every maximum independent set of $J$ is contained in $A$ and excludes $v$.
If $\alpha(G)\leq\alpha(H)$, then $\alpha(B)>\alpha(A)$, so every maximum independent set is contained in $B$ and contains the isolated vertex $v$.
Relabel the vertices of $J$ as $v_1,\dots,v_n$ with $v=v_1$, and write $m=|E(J)|$.

Use the construction in \cref{sec:general-block-construction} for the graph $J$ with the following block sizes.
Put
\begin{equation}\label{eq:winner-scale}
M=10(n+m+1)^2.
\end{equation}
For $i=2,\dots,n$, set
\[
|B_{i,1}|=M+2,
\qquad
|B_{i,r}|=M\quad(r=2,\dots,6).
\]
For the distinguished vertex $v_1$, set
\[
|B_{1,1}|=M+1,
\qquad
|B_{1,r}|=M\quad(r=2,\dots,6).
\]
Every padding block has size $M$.
For each edge $e\in E(J)$, let $P_e$ have three candidates.
Let $\pi_3^J$ be the third ranking produced by this construction.

Add two candidates $c,d$.
Define $\widehat\pi_1$ by listing $c\,d$ first and then concatenating, for $i=1,\dots,n$, the sequence
\[
B_{i,1}B_{i,2}B_{i,3}P_i^{\mathrm L}
B_{i,4}B_{i,5}B_{i,6}D_{i,1}\cdots D_{i,6}.
\]
Define $\widehat\pi_2$ by listing $d$ first, then listing
\[
B_{1,6}B_{1,3}B_{1,2}P_1^{\mathrm R}
B_{1,5}B_{1,4}B_{1,1}\,c\,D_{1,1}\cdots D_{1,6},
\]
and then concatenating, for $i=2,\dots,n$, the sequence
\[
B_{i,6}B_{i,3}B_{i,2}P_i^{\mathrm R}
B_{i,5}B_{i,4}B_{i,1}D_{i,1}\cdots D_{i,6}.
\]
The six vertex blocks for $v_1$ occur in $\pi_3^J$ in the order
\[
B_{1,4}B_{1,5}B_{1,2}B_{1,1}B_{1,6}B_{1,3}.
\]
Define $\widehat\pi_3$ by replacing the consecutive blocks $B_{1,5}B_{1,2}$ in $\pi_3^J$ by $B_{1,5}\,c\,B_{1,2}$ and appending $d$.
Let $\widehat\Pi=(\widehat\pi_1,\widehat\pi_2,\widehat\pi_3)$.

\begin{claim}\label{claim:winner-all-two-to-one}
Every candidate pair in $\widehat\Pi$ is split $2$\nobreakdash-to\nobreakdash-$1$.
Moreover, both $c$ and $d$ have a majority arc to every block candidate outside the six vertex blocks for $v_1$ and to every candidate in an edge-candidate block.
\end{claim}

\begin{proof}
The relations among the vertex blocks, padding blocks, and edge-candidate blocks are those in \cref{lem:general-block-construction}.
Replacing one edge candidate by a block preserves its comparison with every outside candidate.
Pairs inside any block are ordered forward in the first two rankings and backward in the third.
Pairs belonging to two different edge-candidate blocks inherit the $2$\nobreakdash-to\nobreakdash-$1$ relation between the corresponding singleton edge candidates in \cref{sec:construction}.

For a block $B$, write $c\to B$ when $c\to x$ for every $x\in B$, and use the analogous notation for arcs directed from a block to a candidate.
The pairs involving $c,d$ and the vertex blocks for $v_1$ have majority arcs
\begin{equation}\label{eq:cd-majority-arcs}
\begin{gathered}
c\to B_{1,r}\quad(r\in\{1,2,3,6\}),
\qquad
B_{1,r}\to c\quad(r\in\{4,5\}),\\
d\to B_{1,r}\quad(r=1,\dots,6),
\qquad
c\to d.
\end{gathered}
\end{equation}
Each relation follows from the specified positions of $c$ and $d$ in the three rankings.

The candidates $c,d$ precede every block candidate outside the six vertex blocks for $v_1$ in the first two rankings and follow every such candidate in the third.
For $c$ in the second ranking, this uses that the vertex group for $v_1$ occurs first and that $c$ precedes $D_{1,1}$.
Thus both candidates have a majority arc to every block candidate outside the six vertex blocks for $v_1$.

Since $v_1$ has the smallest index, it is the smaller-index endpoint of every incident edge.
No edge-candidate block incident with $v_1$ therefore occurs among the blocks $P_1^{\mathrm R}$ in the second ranking.
Hence $c,d$ precede every edge-candidate block in the first two rankings.
By \cref{eq:general-midpoint}, every edge midpoint satisfies $\mu_e\geq9$.
The third ranking therefore places every edge-candidate block before $B_{1,4}$ and hence before $c$, while $d$ is last.
Thus both $c$ and $d$ have a majority arc to every candidate in an edge-candidate block.
Finally, the pair $c,d$ is ordered as $c\before d$, $d\before c$, and $c\before d$.
This proves the claim.
\end{proof}

By \cref{eq:all-two-to-one-score}, minimizing the Kemeny objective in the constructed profile is equivalent to minimizing the number of backward majority arcs.
We use backward-arc counts for the remainder of the reduction.

We next determine the contribution of $c,d$ and the six vertex blocks for $v_1$.
For an order obtained by interleaving $c,d$ with these six blocks, count the backward arcs involving $c$ or $d$ and one of the blocks, together with the pair $c,d$.

\begin{claim}\label{claim:cd-block-contribution}
The relevant contribution values are
\[
\begin{array}{c|cc}
&\omega_0&\omega_1\\ \hline
\text{minimum}&2M&M+1\\
\text{next value}&2M+1&2M\\
\text{constrained minimum}&2M&2M
\end{array}
\]
The first two rows concern the unconstrained problem.
In the last row, the two constraints are that $c$ is first and that $c\before d$.
Under $\omega_0$, an unconstrained minimum is attained by the order $c\,d\,\omega_0(\mathcal B_1)$.
Under $\omega_1$, an unconstrained minimum is attained by
\[
d\,B_{1,2}B_{1,4}B_{1,5}\,c\,B_{1,1}B_{1,6}B_{1,3}.
\]
Under either constraint, the minimum in both columns is attained by placing $c,d$ before the six vertex blocks in the order $c\before d$.
\end{claim}

\begin{proof}
When $c$ precedes all six vertex blocks, the arcs $B_{1,4}\to c$ and $B_{1,5}\to c$ are backward, so its contribution is initially $2M$.
As $c$ is moved through the blocks in the order $\omega_0=234516$, its contributions in the seven successive gaps are
\[
2M,
3M,
4M,
3M,
2M,
3M+1,
4M+1.
\]
For the order $\omega_1=245163$, the corresponding values are
\[
2M,
3M,
2M,
M,
2M+1,
3M+1,
4M+1.
\]
Candidate $d$ has a majority arc to every candidate in each vertex block.
Its contributions in the seven gaps are the corresponding prefix weights
\[
0,
M,
2M,
3M,
4M,
5M+1,
6M+1
\]
under $\omega_0$, and
\[
0,
M,
2M,
3M,
4M+1,
5M+1,
6M+1
\]
under $\omega_1$.
The pair $c,d$ contributes $0$ when $c\before d$ and $1$ otherwise.

Under $\omega_0$, placing both candidates before the vertex blocks in the order $c\before d$ gives contribution $2M$.
If $d\before c$, the contribution of $c$ is at least $2M$, the contribution of $d$ is nonnegative, and the pair $c,d$ contributes $1$.
Thus the unconstrained minimum is $2M$, and the next-smallest unconstrained value is $2M+1$.

Under $\omega_1$, placing $d$ before the vertex blocks and $c$ after the first three blocks gives contribution $M+1$.
This is the only placement that combines the unique value $M$ for $c$ with the value $0$ for $d$, and it has $d\before c$.
Every other placement with $d\before c$ has contribution at least $2M+1$.
If $c\before d$, then placing $d$ before all six blocks forces the contribution of $c$ to be at least $2M$, while moving $d$ later only increases the total.
Hence the minimum under either constraint and the next-smallest unconstrained value are both $2M$.
\end{proof}

Let $X\subseteq V(J)$ and write $x_i=\chi_X(v_i)$.
The contribution among the six vertex blocks for $v_1$ is $2M^2+M$ under $\omega_0$ and $2M^2$ under $\omega_1$.
Together with \cref{claim:cd-block-contribution}, the contribution involving only these six blocks and $c,d$ is
\begin{equation}\label{eq:distinguished-fixed-cost}
\begin{cases}
2M^2+3M  &\text{ if }v_1\notin X,\\
2M^2+M+1 &\text{ if }v_1\in X.
\end{cases}
\end{equation}
For every $i\geq2$, \cref{eq:switch-costs} gives contributions $2M^2+2M$ under $\omega_0$ and $2M^2$ under $\omega_1$.

For the unconstrained problem, put
\[
W_0=c\,d\,\omega_0(\mathcal B_1),
\qquad
W_1=d\,B_{1,2}B_{1,4}B_{1,5}\,c\,B_{1,1}B_{1,6}B_{1,3},
\]
and define
\begin{equation}\label{eq:winner-unconstrained-block-order}
\widehat\Gamma_X
=
W_{x_1}D_{1,1}\cdots D_{1,6}
\,\omega_{x_2}(\mathcal B_2)D_{2,1}\cdots D_{2,6}
\cdots
\omega_{x_n}(\mathcal B_n)D_{n,1}\cdots D_{n,6}.
\end{equation}
For either constraint that $c$ is first or that $c\before d$, define
\begin{equation}\label{eq:winner-constrained-block-order}
\widehat\Gamma_X^{c,d}
=
c\,d\,\omega_{x_1}(\mathcal B_1)D_{1,1}\cdots D_{1,6}
\,\omega_{x_2}(\mathcal B_2)D_{2,1}\cdots D_{2,6}
\cdots
\omega_{x_n}(\mathcal B_n)D_{n,1}\cdots D_{n,6}.
\end{equation}
In $\widehat\Gamma_X^{c,d}$, the candidates $c$ and $d$ occupy the first two positions of the complete aggregate, in that order, not only the first two positions after the edge-candidate blocks are deleted.

Put
\[
\beta_{\mathrm{fix}}=2nM^2+(2n+1)M.
\]
By \cref{eq:distinguished-fixed-cost,eq:switch-costs}, the contribution among $c,d$ and the block candidates in $\widehat\Gamma_X$ is
\begin{equation}\label{eq:winner-fixed-cost}
\beta_{\mathrm{fix}}
-2M\sum_{i=1}^{n}x_i
+x_1.
\end{equation}
Under either constraint, the corresponding contribution in $\widehat\Gamma_X^{c,d}$ is
\begin{equation}\label{eq:winner-constrained-fixed-cost}
\beta_{\mathrm{fix}}
-2M\sum_{i=1}^{n}x_i
+x_1
+(M-1)x_1.
\end{equation}

\begin{claim}\label{claim:winner-normalization}
For the unconstrained problem and for either constraint that $c$ is first or that $c\before d$, every aggregate can be changed without increasing its number of backward majority arcs so that every vertex block, padding block, and edge-candidate block is contiguous and internally forward-ordered.
For some $X\subseteq V(J)$, deleting the edge-candidate blocks then leaves the order $\widehat\Gamma_X$ in the unconstrained problem and the order $\widehat\Gamma_X^{c,d}$ under either constraint.
Under either constraint, $c,d$ may be assumed to be the first two candidates of the complete aggregate.
In the unconstrained problem, the same is true whenever $v_1\notin X$.
\end{claim}

\begin{proof}
First make every vertex block, padding block, and edge-candidate block contiguous and internally forward-ordered.
For the unconstrained problem and for the constraint $c\before d$, use the argument in \cref{lem:block-normalization} while keeping all outside candidates fixed.
Under the constraint that $c$ is first, restrict the permissible insertion gaps to those after $c$.
Every member of the block already occupies such a gap, so moving all members to a permissible gap minimizing their common external contribution cannot increase the score.
The argument that avoids splitting a previously normalized block remains within this permissible interval.

There are
\[
N_{\mathrm{block}}=n(12M+2)-1<13nM
\]
block candidates.
Delete $c,d$ and the edge-candidate blocks temporarily.
Order the remaining blocks by increasing vertex-group index, place the six vertex blocks before the padding blocks in each group, and order the padding blocks by their second index.
If the six vertex blocks for some $v_i$ are ordered neither as $\omega_0$ nor as $\omega_1$, replace their order by $\omega_1$.
These operations do not increase the contribution among block candidates.
If they change the block sequence, they decrease that contribution by at least $M^2$.
A violation involving two vertex groups, a vertex block and a padding block, or two padding blocks reverses every pair between two blocks of size at least $M$.
For a vertex-block order other than $\omega_0$ and $\omega_1$, \cref{lem:local-states} gives at least three backward block arcs, whereas $\omega_1$ has two.
The resulting block sequence is $\Gamma_X$ for some $X\subseteq V(J)$.

Reinsert $c,d$ and the edge-candidate blocks while preserving their relative order.
This preserves either imposed constraint.
Each pair consisting of a block candidate and one reinserted candidate changes at most once.
Since $4n(m+1)\leq(n+m+1)^2$,
\[
13n(3m+2)
<39n(m+1)
\leq\frac{39}{4}(n+m+1)^2
<M.
\]
The possible increase is therefore at most
\[
(3m+2)N_{\mathrm{block}}<M^2.
\]
Thus a Kemeny-optimal aggregate may be assumed, after deleting $c,d$ and all edge-candidate blocks, to have block order $\Gamma_X$.

Suppose first that $v_1\notin X$.
Move $c,d$ to the first two positions in the order $c\before d$.
By \cref{claim:cd-block-contribution}, this does not increase their contribution with the six vertex blocks for $v_1$.
By \cref{claim:winner-all-two-to-one}, both candidates have a majority arc to every other block candidate and to every candidate in an edge-candidate block, so their remaining contribution cannot increase.
This move preserves either constraint.

Now suppose that $v_1\in X$ and no constraint is imposed.
If deleting the edge-candidate blocks does not leave the order $\widehat\Gamma_X$ defined in \cref{eq:winner-unconstrained-block-order}, then \cref{claim:cd-block-contribution} shows that moving $d$ to the first position and $c$ between $B_{1,5}$ and $B_{1,1}$ in the order $\omega_1(\mathcal B_1)$ decreases the contribution involving block candidates by at least $M-1$.
Moving $d$ forward cannot worsen a comparison with an edge-candidate block.
Only the comparisons between $c$ and the $3m$ candidates in those blocks can worsen.
Since $M-1>3m$, the total number of backward arcs decreases.
Thus an unconstrained Kemeny-optimal aggregate may be assumed to use $\widehat\Gamma_X$.

Finally, suppose that $v_1\in X$ and either constraint is imposed.
Move $c,d$ to the first two positions in the order $c\before d$.
By \cref{claim:cd-block-contribution}, this minimizes their contribution with the six vertex blocks for $v_1$ subject to either constraint.
Their remaining contribution cannot increase by \cref{claim:winner-all-two-to-one}.
The resulting order is $\widehat\Gamma_X^{c,d}$.
\end{proof}

We next compare the contribution of one member of an edge-candidate block with the corresponding contribution in the construction of \cref{sec:construction}.
All minima below are over gaps between the candidates $c,d$ and the contiguous vertex and padding blocks.
The monotonicity argument at the start of the proof of \cref{lem:insertion-contribution} shows that allowing positions inside a block cannot lower the minimum.

For $X\subseteq V(J)$, let $\Gamma_X^{\mathrm{ref}}$ be the block order obtained by deleting $c,d$, giving $B_{1,1}$ one additional candidate, and retaining the order $\omega_{x_i}$ of the six vertex blocks for every $v_i$.
For an edge $e$, let $\phi_e^{\mathrm{ref}}(X)$ be the minimum contribution of the corresponding single edge candidate against the candidates of $\Gamma_X^{\mathrm{ref}}$.
Let $\widehat\phi_e(X)$ be the corresponding minimum for one member of $P_e$ against $\widehat\Gamma_X$.
Let $\widehat\phi_e^{c,d}(X)$ be its minimum against $\widehat\Gamma_X^{c,d}$ when only block gaps after $d$ are allowed.

\begin{claim}\label{claim:winner-edge-block-contribution}
For every edge $e=\{v_i,v_j\}$ and every $X\subseteq V(J)$,
\begin{equation}\label{eq:winner-gap-projection}
|\widehat\phi_e(X)-\phi_e^{\mathrm{ref}}(X)|\leq3.
\end{equation}
If
\[
a_e=\widehat\phi_e(\emptyset),
\]
then
\begin{equation}\label{eq:winner-edge-block-bound}
a_e+Mx_ix_j-6
\leq
\widehat\phi_e(X)
\leq
a_e+Mx_ix_j+8.
\end{equation}
Moreover,
\begin{equation}\label{eq:winner-constrained-edge-block-bound}
\widehat\phi_e^{c,d}(X)
\geq
\widehat\phi_e(X)-1.
\end{equation}
\end{claim}

\begin{proof}
Delete $c,d$ from $\widehat\Gamma_X$ and add one candidate to $B_{1,1}$.
The resulting order is $\Gamma_X^{\mathrm{ref}}$.
Deleting $c,d$ maps every permissible gap of $\widehat\Gamma_X$ to a block gap of $\Gamma_X^{\mathrm{ref}}$, and every block gap of $\Gamma_X^{\mathrm{ref}}$ has a permissible preimage.
Adjoining the extra member inside $B_{1,1}$ does not change the block gaps.
At corresponding gaps, the contribution of one edge candidate differs only on its pairs with $c,d$ and with this extra member of $B_{1,1}$.
The pointwise difference is therefore at most $3$, and taking minima in both directions proves \cref{eq:winner-gap-projection}.

The integer $M$ is even.
Since $4mn\leq(n+m)^2<(n+m+1)^2$, \cref{eq:winner-scale} gives
\[
mn(12M+2)<13mnM<M^2.
\]
Thus \cref{lem:general-block-construction}(iv) applies to $\Gamma_X^{\mathrm{ref}}$, and \cref{lem:insertion-contribution} gives
\[
Mx_ix_j
\leq
\phi_e^{\mathrm{ref}}(X)-\phi_e^{\mathrm{ref}}(\emptyset)
\leq
Mx_ix_j+2.
\]
Write
\[
\widehat\phi_e(X)=\phi_e^{\mathrm{ref}}(X)+\varepsilon_X,
\qquad
\widehat\phi_e(\emptyset)=\phi_e^{\mathrm{ref}}(\emptyset)+\varepsilon_{\emptyset},
\]
where \cref{eq:winner-gap-projection} gives
\[
|\varepsilon_X|,|\varepsilon_{\emptyset}|\leq3.
\]
Subtracting proves \cref{eq:winner-edge-block-bound}.

When $v_1\notin X$, the orders $\widehat\Gamma_X$ and $\widehat\Gamma_X^{c,d}$ agree on $c,d$ and the block candidates, and restricting the allowed gaps cannot decrease the minimum.
Suppose that $v_1\in X$.
After deleting $c$, the two orders agree on $d$ and all block candidates.
Every gap allowed under the constraint therefore corresponds to a gap in the unconstrained order with the same relation to every candidate except $c$.
The two contributions differ by at most $1$ at corresponding gaps, which proves \cref{eq:winner-constrained-edge-block-bound}.
\end{proof}

For $X\subseteq V(J)$, let $e_J(X)=|E(J[X])|$ be the number of edges induced by $X$.
Define
\begin{equation}\label{eq:winner-objective-function}
f_J(X)
=
-2M|X|+x_1+3Me_J(X).
\end{equation}
Put
\[
\beta=\beta_{\mathrm{fix}}+3\sum_{e\in E(J)}a_e,
\qquad
L=18m,
\qquad
U=24m+9\binom m2.
\]
Let $b^\star(X)$ be the minimum backward-arc count among aggregates whose order after deleting the edge-candidate blocks is $\widehat\Gamma_X$.
Let $b_{c,d}^\star(X)$ be the corresponding minimum under either constraint, so deleting the edge-candidate blocks leaves $\widehat\Gamma_X^{c,d}$.

The three members of each $P_e$ have identical comparisons with every outside candidate and contribute three times the one-candidate quantity in \cref{eq:winner-edge-block-bound} against $c,d$ and the block candidates.
The lower bound in that equation and the nonnegativity of the contribution between distinct edge-candidate blocks give
\begin{equation}\label{eq:winner-score-lower}
b^\star(X)
\geq
\beta+f_J(X)-L.
\end{equation}
For the upper bound, place each edge-candidate block in a gap minimizing its contribution against $c,d$ and the block candidates, and order it internally forward.
The error from \cref{eq:winner-edge-block-bound} is at most $24m$, and at most $9\binom m2$ candidate pairs belonging to different edge-candidate blocks are backward.
Thus
\begin{equation}\label{eq:winner-score-upper}
b^\star(X)
\leq
\beta+f_J(X)+U.
\end{equation}
By \cref{eq:winner-constrained-fixed-cost,eq:winner-constrained-edge-block-bound},
\begin{equation}\label{eq:winner-constrained-lower}
b_{c,d}^\star(X)
\geq
\beta+f_J(X)+(M-1)x_1-L-3m.
\end{equation}
Since $n\geq1$, \cref{eq:winner-scale} gives
\begin{equation}\label{eq:winner-error-separation}
L+U+3m
=
45m+9\binom m2
<10(m+2)^2-2
\leq M-2.
\end{equation}

We now determine which subsets can occur in a Kemeny-optimal aggregate.
Suppose that $X$ is not independent, and choose $v_i\in X$ with $r\geq1$ neighbors in $X$.
Removing $v_i$ changes $f_J$ by
\[
2M-3Mr
\]
when $i\neq1$, and by
\[
2M-1-3Mr
\]
when $i=1$.
In either case the value decreases by at least $M$.
Repeatedly removing a vertex from an induced edge therefore produces an independent set whose value under $f_J$ is at least $M$ smaller.
For an independent set $S$,
\[
f_J(S)=-2M|S|+\chi_S(v_1).
\]
An independent set of cardinality less than $\alpha(J)$ has value under $f_J$ at least $2M-1$ larger.
Let $f_J^\star=\min_{X\subseteq V(J)}f_J(X)$.
Consequently, every subset that is not a maximum independent set has value under $f_J$ at least $M$ above $f_J^\star$.
By \cref{eq:winner-score-lower,eq:winner-score-upper,eq:winner-error-separation,claim:winner-normalization}, the set $X$ obtained from every Kemeny-optimal aggregate is a maximum independent set of $J$.

Let $\widehat b^{\star}$ be the minimum backward-arc count in $\widehat\Pi$.
If $\alpha(G)>\alpha(H)$, every maximum independent set of $J$ excludes $v_1$, and the order $\widehat\Gamma_X$ defined in \cref{eq:winner-unconstrained-block-order} places $c$ first and hence places $c\before d$.
If $\alpha(G)\leq\alpha(H)$, every maximum independent set contains $v_1$, and the order $\widehat\Gamma_X$ defined in \cref{eq:winner-unconstrained-block-order} places $d\before c$ with $c$ not first.
Moreover, every aggregate with $c$ first or $c\before d$ has at least $\widehat b^{\star}+2$ backward arcs.
To verify this, let $X$ be the subset obtained after applying \cref{claim:winner-normalization} to such an aggregate.
If $v_1\in X$, then $f_J(X)\geq f_J^\star$.
If $v_1\notin X$, then $X$ is not a maximum independent set because every maximum independent set of $J$ contains $v_1$ in the present case, and hence $f_J(X)\geq f_J^\star+M$.
Therefore \cref{eq:winner-constrained-lower} gives
\[
b_{c,d}^\star(X)
\geq
\begin{cases}
\beta+f_J^\star+M-1-L-3m &\text{ if }v_1\in X,\\
\beta+f_J^\star+M-L-3m   &\text{ if }v_1\notin X.
\end{cases}
\]
In either case, the constrained backward-arc count is at least $\beta+f_J^\star+M-1-L-3m$.
On the other hand, \cref{eq:winner-score-upper} gives $\widehat b^\star\leq\beta+f_J^\star+U$.
The difference is at least $2$ by \cref{eq:winner-error-separation}.

Add one candidate $z$.
Define $\widetilde\pi_1$ by inserting $z$ immediately after $c$ in $\widehat\pi_1$.
Define $\widetilde\pi_2$ by placing $z$ before $\widehat\pi_2$, and define $\widetilde\pi_3$ by placing $z$ after $\widehat\pi_3$:
\[
\widetilde\pi_2=z\,\widehat\pi_2,
\qquad
\widetilde\pi_3=\widehat\pi_3\,z.
\]
Let $\widetilde\Pi=(\widetilde\pi_1,\widetilde\pi_2,\widetilde\pi_3)$.
The pair $c,z$ has majority arc $c\to z$.
For every $x\in\widehat{\cC}\setminus\{c\}$, where $\widehat{\cC}$ is the candidate set of $\widehat\Pi$, the pair $z,x$ has majority arc $z\to x$.
Every pair remains split $2$\nobreakdash-to\nobreakdash-$1$, so $\widetilde\Pi$ is all\nobreakdash-$2$\nobreakdash-to\nobreakdash-$1$.

Write $b_{\widehat\Pi}$ and $b_{\widetilde\Pi}$ for the backward-arc counts in the two majority tournaments.
For an aggregate $\sigma$ of $\widehat{\cC}$, the minimum contribution of the new candidate is
\begin{equation}
\label{eq:z-first-position-contribution}
\min_{\tau:\,\tau|_{\widehat{\cC}}=\sigma}
\bigl(b_{\widetilde\Pi}(\tau)-b_{\widehat\Pi}(\sigma)\bigr)
=
\begin{cases}
0 &\text{ if }c\text{ is first in }\sigma,\\
1 &\text{ if }c\text{ is not first in }\sigma.
\end{cases}
\end{equation}
If $c$ is first, placing $z$ immediately after $c$ makes every new majority arc forward.
If $c$ is not first, choose $x\in\widehat{\cC}$ preceding $c$.
Placing $z$ before $c$ reverses $c\to z$, whereas placing $z$ after $c$ reverses $z\to x$.
Thus every insertion has contribution at least $1$, and placing $z$ first attains this value.

Suppose first that $\alpha(G)>\alpha(H)$.
A Kemeny-optimal aggregate of $\widehat\Pi$ with $c$ first extends by placing $z$ immediately after $c$ and has $\widehat b^{\star}$ backward arcs in $\widetilde\Pi$.
Every aggregate whose restriction to $\widehat{\cC}$ does not place $c$ first has at least $\widehat b^{\star}+1$ backward arcs by \cref{eq:z-first-position-contribution}.
Hence every Kemeny-optimal aggregate of $\widetilde\Pi$ places $c$ first.
Candidate $c$ is therefore the unique Kemeny winner, and every Kemeny-optimal aggregate places $c\before d$.
All four winner and precedence questions are yes-instances.

Now suppose that $\alpha(G)\leq\alpha(H)$.
A Kemeny-optimal aggregate of $\widehat\Pi$ with $c$ not first and $d\before c$ extends by placing $z$ first and has $\widehat b^{\star}+1$ backward arcs.
Every aggregate with $c$ first or $c\before d$ restricts to an aggregate of $\widehat\Pi$ with at least $\widehat b^{\star}+2$ backward arcs, and the contribution of $z$ is nonnegative.
No Kemeny-optimal aggregate of $\widetilde\Pi$ can therefore place $c$ first or place $c\before d$.
Thus $c$ is not a Kemeny winner and every Kemeny-optimal aggregate places $d\before c$.
All four winner and precedence questions are no-instances.

By \cref{claim:winner-all-two-to-one}, every pair is split $2$\nobreakdash-to\nobreakdash-$1$.
For every $i$, one candidate from each of $B_{i,1},B_{i,3},B_{i,4}$ forms a directed triangle.
Since the supplied three rankings realize the majority tournament and every tournament realizable by at most two rankings is transitive, its majority dimension is exactly $3$.

The construction creates
\[
n(12M+2)+2+3m=O(nM+m)
\]
candidates.
Since $M$ is polynomially bounded, the reduction is polynomial.
This proves $\ThetaTwoP$-hardness for \KemenyWinner{}, \KemenyUniqueWinner{}, \KemenyPossiblePrecedence{}, and \KemenyNecessaryPrecedence{} and completes the proof.
\end{proof}

\subsection{Kemeny consensus recognition and uniqueness}\label{sec:recognition-reductions}

Both reductions start from Maximum Independent-Set Recognition and use the construction in \cref{sec:general-block-construction}.
The first reduction constructs an aggregate that is Kemeny-optimal exactly when the supplied independent set is maximum.
The second changes the block sizes and the positions of the edge-candidate blocks so that the supplied aggregate is uniquely Kemeny-optimal exactly under the same condition.

\begin{lemma}
\label{lem:maximum-independent-set-recognition}
The following problem is coNP-complete under polynomial-time many-one reductions: given a nonempty finite simple graph $H$ and a set $S\subseteq V(H)$, decide whether $S$ is a maximum independent set of $H$.
An instance with a nonindependent set $S$ is a no-instance.
\end{lemma}

\begin{proof}
The problem belongs to coNP\@.
If $S$ is not independent, an edge with both endpoints in $S$ certifies a no-instance.
If $S$ is independent but not maximum, a larger independent set certifies a no-instance.

For hardness, reduce from Minimum Vertex Cover Recognition, which is coNP-complete~\cite[Theorem~2]{FitzsimmonsHemaspaandra2021}.
Write $\alpha(F)$ for the maximum independent-set size and $\tau(F)$ for the minimum vertex-cover size of a graph $F$.
Given a graph $G$ and a set $C\subseteq V(G)$, first check whether $C$ is a vertex cover.
If it is not, output the fixed no-instance consisting of a one-vertex graph and the empty set.
Otherwise, add one isolated vertex $u$ and put
\[
H=G\mathbin{\dot\cup}\{u\},
\qquad
S=(V(G)\setminus C)\cup\{u\}.
\]
Then $H$ is nonempty, $S$ is independent, and
\[
\alpha(H)=\alpha(G)+1,
\qquad
|S|=|V(G)|-|C|+1.
\]
Since $\alpha(G)=|V(G)|-\tau(G)$, the cover $C$ is minimum if and only if $S$ is a maximum independent set of $H$.
\end{proof}

Both recognition reductions use the following preprocessing of an instance $(H,S)$ from \cref{lem:maximum-independent-set-recognition}.
Write $H=(V,E)$, put
\[
n=|V|,
\qquad
m=|E|,
\]
and relabel the vertices as $v_1,\dots,v_n$ so that every vertex of $S$ has a smaller index than every vertex of $V\setminus S$.
If $S$ is not independent, use the cyclic profile
\[
a\before b\before c,
\qquad
b\before c\before a,
\qquad
c\before a\before b
\]
together with the order $a\before c\before b$.
Its majority tournament is a directed triangle of majority dimension $3$, and the supplied order has two backward arcs although an order with one backward arc exists.
Thus it is a fixed no-instance for both recognition problems, and every pair is split $2$\nobreakdash-to\nobreakdash-$1$.
For the remainder of this subsection, assume that $S$ is independent; since $H$ is nonempty, $n\geq1$.
Both reductions use
\begin{equation}
\label{eq:verification-scales}
A=n+1,
\qquad
P=A+2,
\qquad
a_i=A+\chi_S(v_i)\quad(i=1,\dots,n),
\qquad
A_{\mathrm{tot}}=\sum_{i=1}^{n}a_i,
\qquad
M=20(n+m+1)^5.
\end{equation}
The integer $M$ is even, and all these quantities are polynomially bounded.
Since $P=n+3$ and $A_{\mathrm{tot}}=n(n+1)+|S|\leq n(n+2)$,
\begin{equation}
\label{eq:verification-parameter-bounds}
P\leq2(n+m+1),
\qquad
A_{\mathrm{tot}}<(n+m+1)^2,
\qquad
m<n+m+1.
\end{equation}
For $X\subseteq V$, write $x_i=\chi_X(v_i)$ and let
\[
e_H(X)=|E(H[X])|
\]
be the number of edges induced by $X$.
Define
\begin{equation}
\label{eq:recognition-objective-function}
f_H(X)
=
-M\sum_{i=1}^{n}a_ix_i
+PMe_H(X).
\end{equation}

\begin{lemma}\label{lem:recognition-objective-gap}
The function $f_H$ has the following properties.
\begin{enumerate}[label=\textup{(\roman*)}]
\item If $X$ is not independent, then some $v_i\in X$ satisfies
\[
f_H(X\setminus\{v_i\})\leq f_H(X)-M.
\]
\item If $S$ is a maximum independent set, then
\begin{equation}\label{eq:recognition-objective-gap-yes}
f_H(X)\geq f_H(S)+M
\qquad
\text{for every }X\neq S.
\end{equation}
\item If $S$ is not a maximum independent set, then there is an independent set $I$ such that
\begin{equation}\label{eq:recognition-objective-gap-no}
f_H(I)\leq f_H(S)-M.
\end{equation}
\end{enumerate}
\end{lemma}

\begin{proof}
Suppose that $X$ is not independent, and choose $v_i\in X$ with $r\geq1$ neighbors in $X$.
Then
\[
f_H(X\setminus\{v_i\})-f_H(X)
=
Ma_i-PMr
\leq
M(A+1-P)
=-M,
\]
which proves (i).
Repeatedly applying (i) produces an independent subset whose value under $f_H$ is at least $M$ smaller whenever the original set is not independent.

For an independent set $I$, \cref{eq:recognition-objective-function} gives
\begin{equation}\label{eq:recognition-independent-value}
f_H(I)
=
-M\bigl(A|I|+|I\cap S|\bigr).
\end{equation}
Suppose that $S$ is maximum.
If $I\neq S$ is independent and $|I|<|S|$, then
\[
A|S|+|S|-A|I|-|I\cap S|
\geq A.
\]
If $|I|=|S|$, then $|I\cap S|\leq|S|-1$, so the same difference is at least $1$.
Together with (i), this proves (ii).

If $S$ is not maximum, choose an independent set $I$ with $|I|\geq|S|+1$.
Since $A=n+1$ and $|S|\leq n$,
\[
A|I|+|I\cap S|-A|S|-|S|
\geq A-|S|
\geq1.
\]
By \cref{eq:recognition-independent-value}, this proves (iii).
\end{proof}

\subsubsection{Kemeny consensus recognition}\label{sec:consensus-verification}

\begin{proof}[Proof of \cref{thm:three-ranking-optimality}(ii)]
A strictly better aggregate certifies a no-instance, so the problem belongs to coNP\@.
The all\nobreakdash-$2$\nobreakdash-to\nobreakdash-$1$ condition is checkable directly from the three rankings.

For coNP-hardness, retain the labeling $v_1,\dots,v_n$, the quantities $n,m$, and the parameters in \cref{eq:verification-scales}.
Apply the construction in \cref{sec:general-block-construction} with the following block sizes.
For every $i=1,\dots,n$, let
\[
|B_{i,1}|=M+a_i,
\qquad
|B_{i,r}|=M\quad(r=2,\dots,6),
\]
and let every padding block have size $M$.
For every edge $e\in E$, let the edge-candidate block $P_e$ have size $P$.
Denote the resulting profile by $\Pi_H$.
By \cref{lem:general-block-construction}, every pair is split $2$\nobreakdash-to\nobreakdash-$1$, and the majority tournament contains a directed triangle and has majority dimension exactly $3$.

By \cref{eq:all-two-to-one-score}, minimizing the Kemeny objective is equivalent to minimizing the number of backward majority arcs.
For $X\subseteq V$, let $\Gamma_X$ be the block order defined in \cref{eq:general-block-order}.
Put
\[
\beta_{\mathrm{fix}}
=
2nM^2+M\sum_{i=1}^{n}a_i.
\]
The contribution among the block candidates in $\Gamma_X$ is
\begin{equation}
\label{eq:verification-fixed-cost}
\beta_{\mathrm{fix}}
-M\sum_{i=1}^{n}a_ix_i.
\end{equation}
This follows because the six vertex blocks for $v_i$ contribute $2M^2+a_iM$ when ordered as $\omega_0$ and $2M^2$ when ordered as $\omega_1$.

We next bound the contribution of one member of an edge-candidate block against the block candidates.
Let $\phi_e(X)$ be the minimum of this contribution over the gaps of $\Gamma_X$.
Let $\phi_e^{\mathrm{ref}}(X)$ be the corresponding minimum when every vertex block, including $B_{i,1}$, has size exactly $M$.

The insertion-cost calculation in the proof of \cref{lem:insertion-contribution} simplifies in this equal-size construction.
For $e=\{v_i,v_j\}$ with $i<j$, the minimum attained while one member of $P_e$ passes the six vertex blocks for either endpoint is $-2M+Mx_i$ or $-2M+Mx_j$, relative to a common value $\eta_e$ immediately before the first endpoint group.
The insertion cost does not fall below this common value between the endpoint groups, and it changes monotonically before the first and after the second endpoint group.
Consequently,
\[
\phi_e^{\mathrm{ref}}(X)
=
\eta_e-2M+\min\{Mx_i,Mx_j\}
=
\eta_e-2M+Mx_ix_j.
\]
Subtracting the case $X=\emptyset$ gives
\begin{equation}
\label{eq:verification-equal-size-edge-contribution}
\phi_e^{\mathrm{ref}}(X)-\phi_e^{\mathrm{ref}}(\emptyset)
=
Mx_ix_j.
\end{equation}
At corresponding gaps, the actual and equal-size contributions differ only on the $A_{\mathrm{tot}}$ additional candidates in the blocks $B_{i,1}$.
Therefore
\[
|\phi_e(X)-\phi_e^{\mathrm{ref}}(X)|\leq A_{\mathrm{tot}}.
\]
Writing $c_e=\phi_e(\emptyset)$ and using \cref{eq:verification-equal-size-edge-contribution}, we obtain
\begin{equation}
\label{eq:verification-edge-contribution-bound}
c_e+Mx_ix_j-2A_{\mathrm{tot}}
\leq
\phi_e(X)
\leq
c_e+Mx_ix_j+2A_{\mathrm{tot}}.
\end{equation}

First make every vertex block, padding block, and edge-candidate block contiguous and internally forward-ordered by \cref{lem:block-normalization}.
There are
\[
N_{\mathrm{block}}
=12nM+A_{\mathrm{tot}}
<13(n+m+1)M
\]
block candidates, and \cref{eq:verification-parameter-bounds} gives
\[
Pm<2(n+m+1)^2.
\]
Temporarily delete the edge-candidate blocks.
Order the remaining blocks by increasing vertex-group index, place the six vertex blocks before the padding blocks in each group, and order the padding blocks by their second index.
If the six vertex blocks for some $v_i$ are ordered neither as $\omega_0$ nor as $\omega_1$, replace their order by $\omega_1$.
If any of these operations changes the block sequence, the contribution among block candidates decreases by at least $M^2$, by the same argument as in \cref{lem:fixed-normalization}.
The resulting sequence is $\Gamma_X$ for some $X\subseteq V$.

Reinsert the edge-candidate blocks while preserving their relative order.
At most
\[
PmN_{\mathrm{block}}
<26(n+m+1)^3M
<M^2
\]
comparisons between block candidates and edge candidates can change from forward to backward, where the final inequality follows from $n+m+1\geq2$ and \cref{eq:verification-scales}.
Thus every aggregate can be changed without increasing its backward-arc count so that deleting the edge-candidate blocks leaves the block order $\Gamma_X$ for some $X\subseteq V$.

Put
\[
\beta
=
\beta_{\mathrm{fix}}+P\sum_{e\in E}c_e,
\qquad
R
=
2PmA_{\mathrm{tot}}+P^2\binom m2.
\]
Let $b^\star(X)$ be the minimum backward-arc count among aggregates for which deleting the edge-candidate blocks leaves the block order $\Gamma_X$.
The lower bound in \cref{eq:verification-edge-contribution-bound} and the nonnegativity of the contribution between distinct edge-candidate blocks give
\begin{equation}
\label{eq:verification-score-lower}
b^\star(X)
\geq
\beta+f_H(X)-R.
\end{equation}
Conversely, place every edge-candidate block in a gap minimizing its contribution against the block candidates, order each block internally forward, and order the blocks arbitrarily when they occupy the same gap.
The upper bound in \cref{eq:verification-edge-contribution-bound}, together with the bound $P^2\binom m2$ on candidate pairs belonging to different edge-candidate blocks, gives
\begin{equation}
\label{eq:verification-score-upper}
b^\star(X)
\leq
\beta+f_H(X)+R.
\end{equation}

By \cref{eq:verification-parameter-bounds},
\[
R
<
2\bigl(2(n+m+1)\bigr)(n+m+1)(n+m+1)^2
+
\bigl(2(n+m+1)\bigr)^2\frac{(n+m+1)^2}{2}
=6(n+m+1)^4.
\]
Since $n+m+1\geq2$, \cref{eq:verification-scales} gives
\begin{equation}
\label{eq:verification-separation}
M>4R.
\end{equation}
It remains to construct in polynomial time an aggregate attaining $b^\star(S)$.
The bounds above do not by themselves determine the positions of the edge-candidate blocks when the block order is $\Gamma_S$.

For $v_j\notin S$, let $\gamma_j$ be the gap immediately after $B_{j,2},B_{j,3}$ when the six vertex blocks for $v_j$ are ordered as $\omega_0$.
In the equal-size construction, this gap gives the unique minimum contribution $-2M$, relative to the value before the six vertex blocks.
When the blocks are ordered as $\omega_1$, the corresponding minimum is $-M$.
The insertion cost between the two endpoint vertex groups never falls below its value at the beginning of that interval, and it changes monotonically before the first endpoint and after the second.
Consequently, for an edge $e=\{v_i,v_j\}$ with $i<j$, the following statements hold.
If $v_i\in S$ and $v_j\notin S$, then $\gamma_j$ is the unique minimum in the equal-size construction.
If $v_i,v_j\notin S$, then $\gamma_i$ and $\gamma_j$ are the only minima.
Every other gap has contribution at least $M$ larger for one member of $P_e$.
The first case has this orientation because the vertices of $S$ were assigned the smaller indices.

Changing from equal block sizes to the actual sizes can change the difference between two gaps by at most $2A_{\mathrm{tot}}$.
Moving the entire block $P_e$ can increase its contribution with all other edge-candidate blocks by at most $P^2(m-1)$.
The bounds in \cref{eq:verification-parameter-bounds} give
\[
2A_{\mathrm{tot}}+P(m-1)
<4(n+m+1)^2
<M.
\]
Consequently,
\[
M-2A_{\mathrm{tot}}
>
P(m-1),
\]
and hence
\begin{equation}
\label{eq:verification-endpoint-dominance}
P(M-2A_{\mathrm{tot}})>P^2(m-1).
\end{equation}
If $v_i\in S$ and $v_j\notin S$, moving $P_e$ from any other gap to $\gamma_j$ strictly decreases the total number of backward arcs.
If $v_i,v_j\notin S$, moving $P_e$ from any other gap to either $\gamma_i$ or $\gamma_j$ strictly decreases the total.
Thus, in a Kemeny-optimal aggregate with block order $\Gamma_S$, the block for an edge with one endpoint in $S$ occupies the gap of its endpoint outside $S$, while the block for an edge with both endpoints outside $S$ occupies one of its two endpoint gaps.

Only edges with both endpoints in $V\setminus S$ have two possible endpoint gaps.
For such an edge $e=\{v_i,v_j\}$ with $i<j$, let $y_e=0$ denote placement at $\gamma_i$ and let $y_e=1$ denote placement at $\gamma_j$.
The contribution of $P_e$ against the block candidates gives an explicitly computable unary cost for each value of $y_e$.
The contribution between $P_e$ and every edge-candidate block whose position is fixed is included in this unary cost.

After the endpoint gaps have been chosen, order the edge-candidate blocks at a gap $\gamma_j$ as follows.
First list the blocks $P_{\{v_i,v_j\}}$ assigned to their larger-index endpoint, in increasing order of $i$.
Then list the blocks $P_{\{v_j,v_k\}}$ assigned to their smaller-index endpoint, in increasing order of $k$.
The first two input rankings agree with this order on every pair of blocks at the same gap, so no majority arc between two such blocks is backward.

Consider two edges $e$ and $f$ with both endpoints in $V\setminus S$, and assume that $e$ precedes $f$ in the first input ranking.
For $y_e,y_f\in\{0,1\}$, let $C_{ef}(y_e,y_f)$ be the number of backward candidate pairs between $P_e$ and $P_f$.
Form the $2\times2$ matrix whose entry in row $a$ and column $b$ is $P^{-2}C_{ef}(a,b)$, where $a=y_e$ and $b=y_f$.
Write $e=\{v_i,v_j\}$ and $f=\{v_k,v_\ell\}$ with $i<j$, $k<\ell$, and $e$ preceding $f$ in the first ranking.
A case analysis of the relative order of the four endpoints gives the following matrices.
\[
\begin{array}{c|c}
\text{endpoint pattern}&\bigl(P^{-2}C_{ef}(a,b)\bigr)_{a,b\in\{0,1\}}\\ \hline
j\leq k
&\begin{pmatrix}0&0\\0&0\end{pmatrix}\\[3mm]
i=k,\ \text{or }j=\ell,\ \text{or }i<k<j<\ell
&\begin{pmatrix}0&0\\1&0\end{pmatrix}\\[3mm]
i<k<\ell<j\text{ and }P_e\before_{\pi_3^H}P_f
&\begin{pmatrix}0&0\\1&1\end{pmatrix}\\[3mm]
i<k<\ell<j\text{ and }P_f\before_{\pi_3^H}P_e
&\begin{pmatrix}1&1\\0&0\end{pmatrix}
\end{array}
\]
The first row includes $j=k$.
The second row covers a common smaller-index endpoint, a common larger-index endpoint, and $i<k<j<\ell$.
The last two rows cover $i<k<\ell<j$.
When $i+j=k+\ell$, the fixed order in the third ranking selects one of the last two rows.
Each such matrix satisfies
\[
C_{ef}(0,0)+C_{ef}(1,1)
\leq
C_{ef}(0,1)+C_{ef}(1,0).
\]
For a pairwise cost matrix $(c_{ab})$ satisfying this inequality, put
\[
w=c_{01}+c_{10}-c_{00}-c_{11}\geq0
\]
and write
\[
c_{y_ey_f}
=
c_{00}
+(c_{10}-c_{00})y_e
+(c_{11}-c_{10})y_f
+w(1-y_e)y_f.
\]
Represent $y_e=0$ by placing the variable vertex for $e$ on the source side of a cut and $y_e=1$ by placing it on the sink side.
A directed edge from the vertex for $e$ to the vertex for $f$ with capacity $w$ contributes exactly $w(1-y_e)y_f$ to the cut capacity.
The linear terms in this identity, together with the previously defined costs depending on one variable, are unary functions.
After subtracting the smaller value of each unary function, its two nonnegative values are represented by edges incident with the source and sink.
Interactions between two edge-candidate blocks whose positions are fixed contribute only an additive constant.
Summing these representations gives a network whose cut capacity differs from the contribution depending on the variables $y_e$ by a constant.
Therefore, one minimum $s$--$t$ cut yields an optimal choice of all endpoint gaps.
All capacities have polynomial bit length.

Let $\tau_S$ be the aggregate obtained from $\Gamma_S$ by placing the edge-candidate blocks according to a minimum cut, ordering the blocks at each gap by the rule in the preceding paragraph, and ordering every block internally forward.
The preceding placement argument and the minimum-cut calculation give
\begin{equation}
\label{eq:verification-exact-restricted-minimum}
b_{\Pi_H}(\tau_S)=b^\star(S).
\end{equation}

Suppose first that $S$ is a maximum independent set.
For every $X\neq S$, \cref{eq:verification-score-lower,eq:recognition-objective-gap-yes} give
\[
b^\star(X)
\geq
\beta+f_H(S)+M-R.
\]
By \cref{eq:verification-score-upper,eq:verification-separation},
\[
\beta+f_H(S)+M-R
>
\beta+f_H(S)+R
\geq
b^\star(S).
\]
Together with \cref{eq:verification-exact-restricted-minimum}, this shows that $\tau_S$ is globally Kemeny-optimal.

Conversely, suppose that $S$ is not maximum and let $I$ be as in \cref{eq:recognition-objective-gap-no}.
Then \cref{eq:verification-score-lower,eq:verification-score-upper,eq:verification-separation} give
\[
b^\star(I)
\leq
\beta+f_H(I)+R
<
\beta+f_H(S)-R
\leq
b^\star(S)
=
b_{\Pi_H}(\tau_S).
\]
Thus $\tau_S$ is not Kemeny-optimal.

The construction has $12nM+A_{\mathrm{tot}}+Pm$ candidates.
The three rankings, the minimum-cut network, and the aggregate $\tau_S$ can all be generated in polynomial time, and every network capacity has polynomial bit length.
We have therefore produced in polynomial time a three-ranking all\nobreakdash-$2$\nobreakdash-to\nobreakdash-$1$ profile $\Pi_H$ and an aggregate $\tau_S$ such that $\tau_S$ is Kemeny-optimal if and only if $S$ is a maximum independent set of $H$.
This proves coNP-hardness and completes the proof.
\end{proof}

\subsubsection{Unique Kemeny consensus recognition}\label{sec:unique-consensus-verification}

We first record when uniqueness under block-contiguity and internal-order restrictions implies unrestricted uniqueness.

\begin{lemma}
\label{lem:unique-block-normalization}
Let the candidates of a tournament be partitioned into nonempty blocks such that the members of each block have the same comparison with every candidate outside the block and the subtournament inside each block is transitive.
Suppose there is a unique order minimizing the number of backward arcs among the orders in which every block is contiguous and internally follows its transitive order.
Then this order is the unique order minimizing the number of backward arcs without the block-contiguity restriction.
\end{lemma}

\begin{proof}
Let $\tau$ be the unique minimum order among those satisfying the contiguity and internal-order conditions in the statement, and suppose that $\sigma$ is an unrestricted minimum order.
Apply the block-normalization procedure from \cref{lem:block-normalization}, one block at a time, to $\sigma$.
The procedure never increases the backward-arc count and ends in an order satisfying the contiguity and internal-order conditions.
Since $\sigma$ already has the minimum backward-arc count, every normalization step preserves that count, and the final order must be $\tau$.

Assume that some normalization step changes the order, and let the normalization of $Q$ be the last normalization step that changes the order.
Immediately after this step the order is already $\tau$, because no later step changes it.
Immediately before the step, deleting the members of $Q$ therefore leaves the order $\tau\setminus Q$.
For a gap $\gamma$ of $\tau\setminus Q$, let $f_Q(\gamma)$ be the number of backward arcs between one member of $Q$ placed in $\gamma$ and the candidates outside $Q$.
All members of $Q$ have the same external comparisons, so the same function applies to every member.

The gap occupied by $Q$ in $\tau$ is the unique minimum of $f_Q$ among gaps between blocks.
Otherwise, placing $Q$ contiguously and internally forward in another minimizing block gap would produce a second minimum order satisfying the contiguity and internal-order conditions.
Across the interior gaps of another block, $f_Q$ changes by the same nonzero amount at every crossing.
Consequently, every interior gap has larger value than one of the two adjacent block gaps and cannot be a global minimum.
Thus the gap of $Q$ in $\tau$ is the unique minimum among all gaps of $\tau\setminus Q$.

Before the last changing step, suppose that the members of $Q$ occupy gaps $\gamma_1,\dots,\gamma_{|Q|}$ of $\tau\setminus Q$.
Their external contribution is
\[
\sum_{r=1}^{|Q|}f_Q(\gamma_r),
\]
and their internal contribution is nonnegative, with equality only in the transitive internal order.
If their arrangement differs from the arrangement in $\tau$, either one member occupies a nonminimum gap or the internal order is not transitive.
Normalizing $Q$ would then decrease the backward-arc count strictly, contradicting the optimality of $\sigma$.
Hence no normalization step changes the order, and $\sigma=\tau$.
\end{proof}

\begin{proof}[Proof of \cref{thm:three-ranking-optimality}(ii)]
A distinct aggregate with no larger objective value certifies a no-instance, so the problem belongs to coNP\@.
The all\nobreakdash-$2$\nobreakdash-to\nobreakdash-$1$ condition is checkable as in the preceding proof.

For coNP-hardness, retain the labeling $v_1,\dots,v_n$, the quantities $n,m$, and the parameters in \cref{eq:verification-scales}.
In addition, put
\begin{equation}
\label{eq:unique-verification-scales}
\zeta=(n+m+1)^2,
\qquad
N_{\mathrm{add}}=A_{\mathrm{tot}}+n\zeta.
\end{equation}
The quantity $N_{\mathrm{add}}$ is the total number of candidates added to the blocks $B_{i,1}$ and $B_{i,2}$ beyond their common size $M$.
Since $n\leq n+m$ and $A_{\mathrm{tot}}<(n+m+1)^2$,
\begin{equation}
\label{eq:unique-verification-parameter-bound}
N_{\mathrm{add}}<(n+m+1)^3.
\end{equation}
Apply the construction in \cref{sec:general-block-construction} with the following block sizes.
For every $i=1,\dots,n$, let
\begin{equation}
\label{eq:unique-verification-block-sizes}
|B_{i,1}|=M+a_i,
\qquad
|B_{i,2}|=M+\zeta,
\qquad
|B_{i,r}|=M\quad(r=3,4,5,6),
\end{equation}
and let every padding block have size $M$.
For every edge $e\in E$, let the edge-candidate block $P_e$ have size $P$.
Denote the resulting profile by $\Pi_H$.
By \cref{lem:general-block-construction}, every pair is split $2$\nobreakdash-to\nobreakdash-$1$, and the majority tournament contains a directed triangle and has majority dimension exactly $3$.
By \cref{eq:all-two-to-one-score}, minimizing the Kemeny objective is equivalent to minimizing the number of backward majority arcs.

For $X\subseteq V$, let $\Gamma_X$ be the block order defined in \cref{eq:general-block-order}.
None of the two backward block arcs under $\omega_0$ or $\omega_1$ has an endpoint with label $2$, so the additional $\zeta$ candidates in $B_{i,2}$ do not change either contribution.
Put
\[
\beta_{\mathrm{fix}}
=
2nM^2+M\sum_{i=1}^{n}a_i.
\]
The contribution among the block candidates in $\Gamma_X$ is therefore
\begin{equation}
\label{eq:unique-verification-fixed-cost}
\beta_{\mathrm{fix}}
-M\sum_{i=1}^{n}a_ix_i.
\end{equation}

We next compare the contribution of one member of an edge-candidate block with the corresponding contribution when all vertex and padding blocks have size $M$.
For a block gap $\gamma$ of $\Gamma_X$, let $\Phi_e^X(\gamma)$ be the contribution of one member of $P_e$ against the block candidates, and put $\phi_e(X)=\min_{\gamma}\Phi_e^X(\gamma)$.
Let $\Phi_e^{\mathrm{ref},X}(\gamma)$ and $\phi_e^{\mathrm{ref}}(X)$ denote the corresponding quantities in the equal-size reference construction.

\begin{claim}
\label{claim:unique-verification-equal-size}
Let $e=\{v_i,v_j\}$ with $i<j$.
The minimum insertion cost in the equal-size construction is attained only within the vertex groups for $v_i$ and $v_j$.
Relative to the common value immediately before either endpoint group, the two endpoint minima are
\[
-2M+Mx_i
\qquad\text{and}\qquad
-2M+Mx_j.
\]
Every other gap has insertion cost at least $M$ larger than the global minimum.
Consequently,
\begin{equation}
\label{eq:unique-verification-equal-size-edge-contribution}
\phi_e^{\mathrm{ref}}(X)-\phi_e^{\mathrm{ref}}(\emptyset)
=
Mx_ix_j.
\end{equation}
If $x_i=x_j=0$, the two endpoint minima are the only minima.
If $(x_i,x_j)=(1,0)$, the minimum in the vertex group for $v_j$ is unique.
\end{claim}

\begin{proof}
Before the vertex group for $v_i$, moving one member of $P_e$ to the right decreases its contribution at every block crossing, and after the vertex group for $v_j$ it increases the contribution at every crossing.
Within the group for $v_i$, moving past the blocks with labels $1,2,3$ decreases the contribution, while moving past the blocks with labels $4,5,6$ increases it.
Reading the orders $\omega_0$ and $\omega_1$ gives unique minima $-2M$ and $-M$, respectively, relative to the value before the six vertex blocks, and the total change across the group is zero.
Within the group for $v_j$, moving past the blocks with labels $6,3,2$ decreases the contribution, while moving past those with labels $5,4,1$ increases it.
This gives the same two minima and again total change zero.

Between the two endpoint groups, the midpoint position divides the intervening blocks equally between those that increase the contribution and those that decrease it.
When $j-i$ is odd, every block that increases the contribution precedes every block that decreases it in $\Gamma_X$.
When $j-i$ is even, only the vertex group containing the midpoint can interleave the two types.
Before that group, the contribution is at least $6M$ above its value at the beginning of the interval, and at most three decreasing crossings can occur before the corresponding increasing crossings.
Thus the insertion cost never falls below its value at the beginning of the interval and has the same value at the end.
Every gap outside the endpoint minima is therefore at least $M$ above the global minimum.
Finally,
\[
\min\{-2M+Mx_i,-2M+Mx_j\}
=
-2M+Mx_ix_j,
\]
which proves the identity and the uniqueness statements.
\end{proof}

At corresponding gaps, the actual and equal-size contributions differ only on the $N_{\mathrm{add}}$ additional candidates in the blocks with labels $1$ and $2$.
Thus
\[
|\Phi_e^X(\gamma)-\Phi_e^{\mathrm{ref},X}(\gamma)|\leq N_{\mathrm{add}}
\]
for every gap $\gamma$, and consequently
\[
|\phi_e(X)-\phi_e^{\mathrm{ref}}(X)|\leq N_{\mathrm{add}}.
\]
Writing $c_e=\phi_e(\emptyset)$ gives
\begin{equation}
\label{eq:unique-verification-edge-contribution-bound}
c_e+Mx_ix_j-2N_{\mathrm{add}}
\leq
\phi_e(X)
\leq
c_e+Mx_ix_j+2N_{\mathrm{add}}.
\end{equation}

Make every vertex block, padding block, and edge-candidate block contiguous and internally forward-ordered by \cref{lem:block-normalization}.
Temporarily delete the edge-candidate blocks.
Order the remaining blocks by increasing vertex-group index, place the six vertex blocks before the padding blocks in each group, and order the padding blocks by their second index.
If the six vertex blocks for some $v_i$ are ordered neither as $\omega_0$ nor as $\omega_1$, replace their order by $\omega_1$.
Every changed pair of vertex groups, every changed vertex-block--padding-block pair, and every changed padding-block pair decreases the number of backward arcs by at least $M^2$.
A vertex-block order other than $\omega_0$ and $\omega_1$ has at least three backward block arcs by \cref{lem:local-states}, whereas $\omega_1$ contributes exactly $2M^2$ because its two backward arcs involve only blocks of size $M$.
Thus, if the block sequence changes, its contribution decreases by at least $M^2$, and the resulting sequence is $\Gamma_X$ for some $X\subseteq V$.

There are
\[
N_{\mathrm{block}}
=12nM+N_{\mathrm{add}}
<13(n+m+1)M
\]
block candidates, and \cref{eq:verification-parameter-bounds} gives $Pm<2(n+m+1)^2$.
Reinserting the edge-candidate blocks while preserving their relative order can increase the contribution by at most
\[
PmN_{\mathrm{block}}
<26(n+m+1)^3M
<M^2,
\]
where the final inequality follows from $n+m+1\geq2$ and \cref{eq:verification-scales}.
Consequently, the normalization strictly decreases the backward-arc count whenever deleting the edge-candidate blocks from the original aggregate does not leave $\Gamma_X$ for some $X\subseteq V$.

Put
\[
\beta
=
\beta_{\mathrm{fix}}+P\sum_{e\in E}c_e,
\qquad
R
=
2PmN_{\mathrm{add}}+P^2\binom m2.
\]
Let $b^\star(X)$ be the minimum backward-arc count among aggregates in which every vertex block, padding block, and edge-candidate block is contiguous and internally forward-ordered and deleting the edge-candidate blocks leaves $\Gamma_X$.
Then \cref{eq:unique-verification-edge-contribution-bound} gives
\begin{equation}
\label{eq:unique-verification-score-bounds}
\beta+f_H(X)-R
\leq
b^\star(X)
\leq
\beta+f_H(X)+R.
\end{equation}

By \cref{eq:verification-parameter-bounds,eq:unique-verification-parameter-bound},
\[
R
<
2\bigl(2(n+m+1)\bigr)(n+m+1)(n+m+1)^3
+
\bigl(2(n+m+1)\bigr)^2\frac{(n+m+1)^2}{2}
<5(n+m+1)^5.
\]
Thus \cref{eq:verification-scales} gives
\begin{equation}
\label{eq:unique-verification-error-separation}
M>4R.
\end{equation}

It remains to define an aggregate $\tau_S$ satisfying these contiguity and internal-order conditions and to prove that it is the unique order of minimum Kemeny score subject to those conditions when $S$ is maximum.
For each vertex $v_j\notin S$, let $\gamma_j$ be the gap immediately after $B_{j,2},B_{j,3}$ when the six vertex blocks for $v_j$ are ordered as $\omega_0$ in $\Gamma_S$.
Every edge $e=\{v_i,v_j\}$, written with $i<j$, satisfies $v_j\notin S$ because the vertices of $S$ have the smaller indices and $S$ is independent.
We will place the edge-candidate block $P_e$ in $\gamma_j$.

\begin{claim}
\label{claim:unique-verification-endpoints}
Fix $e=\{v_i,v_j\}$ with $i<j$.
Among the gaps of $\Gamma_S$, $\gamma_j$ is the unique position of $P_e$ in an order of minimum Kemeny score subject to the contiguity and internal-order conditions.
More precisely, moving $P_e$ from any other gap to $\gamma_j$ decreases its contribution with the block candidates by more than its contribution with all other edge-candidate blocks can increase.
\end{claim}

\begin{proof}
Suppose first that $v_i\in S$.
The six vertex blocks for $v_i$ and $v_j$ are ordered as $\omega_1$ and $\omega_0$, respectively.
By \cref{claim:unique-verification-equal-size}, $\gamma_j$ is the unique minimum position in the equal-size construction, and every other gap increases the contribution of one member of $P_e$ by at least $M$.
The additional candidates in the blocks $B_{h,1}$ and $B_{h,2}$ can change the difference between two gap contributions by at most $2N_{\mathrm{add}}$.
Thus, in the present construction, placing one member of $P_e$ in $\gamma_j$ instead of any other gap decreases its contribution with the block candidates by at least $M-2N_{\mathrm{add}}$.

Now suppose that $v_i\notin S$.
The six vertex blocks for both endpoints are ordered as $\omega_0$.
By \cref{claim:unique-verification-equal-size}, $\gamma_i$ and $\gamma_j$ are the only minimum positions in the equal-size construction, and every other gap has contribution at least $M$ larger.
It remains to compare $\gamma_i$ and $\gamma_j$ after the block sizes are changed.
Write $d=j-i$.
The block $B_{h,2}$ occupies position $12h-8$ in the sequence $\mathcal Q$ from \cref{sec:general-block-construction}.
Moving one member of $P_e$ from $\gamma_i$ toward $\gamma_j$ past this block increases its contribution exactly when $12h-8<\mu_e=6(i+j)-9$ and decreases it otherwise.
Moving from $\gamma_i$ to $\gamma_j$ passes the blocks $B_{h,2}$ for $h=i+1,\dots,j$.
If $d=2r+1$, exactly $r$ of these indices satisfy the displayed inequality and $r+1$ do not.
If $d=2r$, exactly $r-1$ satisfy it and $r+1$ do not.
Consequently, the additional $\zeta$ candidates in each block $B_{h,2}$ decrease the contribution at $\gamma_j$ relative to $\gamma_i$ by $\zeta$ when $d$ is odd and by $2\zeta$ when $d$ is even.
The additional candidates in the blocks $B_{h,1}$ change this difference by at most $A_{\mathrm{tot}}$ in absolute value.
Therefore
\[
\Phi_e^S(\gamma_j)
\leq
\Phi_e^S(\gamma_i)-(\zeta-A_{\mathrm{tot}}).
\]
Every gap other than $\gamma_i$ and $\gamma_j$ remains at least $M-2N_{\mathrm{add}}$ above one of these two endpoint gaps.

Put $g=\min\{M-2N_{\mathrm{add}},\zeta-A_{\mathrm{tot}}\}$.
The definitions and $A_{\mathrm{tot}}\leq n(n+2)$ give
\[
\zeta-A_{\mathrm{tot}}-P(m-1)
\geq
(n+m+1)^2-n(n+2)-(n+3)(m-1)
=
m^2+nm+n-m+4
>0.
\]
Moreover, \cref{eq:verification-scales,eq:verification-parameter-bounds,eq:unique-verification-parameter-bound} give
\[
M-2N_{\mathrm{add}}
>
20(n+m+1)^5-2(n+m+1)^3
>
2(n+m+1)^2
>
P(m-1).
\]
Hence $g>P(m-1)$.
Moving all $P$ members of $P_e$ to $\gamma_j$ decreases their contribution with the block candidates by at least $Pg$.
Their contribution with the $P(m-1)$ candidates in the other edge-candidate blocks can increase by at most $P^2(m-1)$.
Since $Pg>P^2(m-1)$, the total number of backward majority arcs decreases strictly.
\end{proof}

Define $\tau_S$ as follows.
Deleting the edge-candidate blocks from $\tau_S$ leaves the block order $\Gamma_S$.
For every edge $e=\{v_i,v_j\}$ with $i<j$, place $P_e$ in $\gamma_j$.
Among the edge-candidate blocks placed in one gap $\gamma_j$, order $P_{\{v_i,v_j\}}$ by increasing $i$, and order the candidates inside every block in their forward order.
The first two input rankings agree with this relative order of the edge-candidate blocks placed in the same gap.

\begin{claim}
\label{claim:unique-verification-block-minimum}
If $S$ is a maximum independent set of $H$, then $\tau_S$ is the unique order of minimum Kemeny score among the orders satisfying the contiguity and internal-order conditions.
\end{claim}

\begin{proof}
Let $\sigma$ minimize the Kemeny score among the orders satisfying the contiguity and internal-order conditions.
The strict normalization argument above shows that deleting the edge-candidate blocks from it must leave $\Gamma_X$ for some $X\subseteq V$.
If $X\neq S$, then \cref{eq:unique-verification-score-bounds,eq:recognition-objective-gap-yes,eq:unique-verification-error-separation} give
\[
b^\star(X)
\geq
\beta+f_H(S)+M-R
>
\beta+f_H(S)+R
\geq
b^\star(S).
\]
Hence $X=S$.

By \cref{claim:unique-verification-endpoints}, each edge-candidate block $P_e$ must occupy the gap specified in the definition of $\tau_S$.
For two such blocks in the same gap $\gamma_j$, the first two input rankings order them by increasing smaller endpoint, so every majority arc between the two blocks points in that direction.
If their relative order contains an inversion, it contains an adjacent inverted pair.
Swapping that pair changes no comparison with a block candidate or with any other edge-candidate block and strictly decreases the contribution between the two swapped blocks.
Thus the edge-candidate blocks in each gap have the unique optimal order specified above.
The transitive majority relation inside each block has a unique topological order.
Therefore $\sigma=\tau_S$.
\end{proof}

Suppose first that $S$ is a maximum independent set.
The preceding claim shows that $\tau_S$ is the unique Kemeny-optimal order among the orders satisfying the contiguity and internal-order conditions.
All blocks satisfy the hypotheses of \cref{lem:unique-block-normalization}: their members have identical external comparisons, and the majority relation inside each block is transitive.
That lemma therefore implies that $\tau_S$ is the unique unrestricted Kemeny-optimal aggregate.

Conversely, suppose that $S$ is not maximum and let $I$ be as in \cref{eq:recognition-objective-gap-no}.
Then
\[
b^\star(I)
\leq
\beta+f_H(I)+R
<
\beta+f_H(S)-R
\leq
b^\star(S)
\leq
b_{\Pi_H}(\tau_S),
\]
where the strict inequality follows from $M>2R$.
Thus $\tau_S$ is not Kemeny-optimal and in particular is not the unique Kemeny-optimal aggregate.

The construction has $12nM+N_{\mathrm{add}}+Pm$ candidates.
The three rankings and the supplied order $\tau_S$ can be generated in polynomial time.
We have proved that $\tau_S$ is the unique Kemeny-optimal aggregate of $\Pi_H$ if and only if $S$ is a maximum independent set of $H$.
This proves coNP-hardness and completes the proof.
\end{proof}

\section{Transformations of profiles and support classifications}\label{sec:support}

This section derives the pairwise-support classifications from the three-ranking hardness results.
When $q/2<s\leq2q/3$, replicating the profile and appending an order together with its reverse produce hard instances with exact support $s$; the four-ranking transformations in \cref{sec:four-ranking-transformations} handle the boundary case $s=q/2$.
Both operations apply a positive affine transformation to every Kemeny score and therefore preserve all Kemeny-optimal orders, the four winner and precedence conditions, and both recognition conditions.
The polynomial-time case $s>2q/3$ follows from transitivity of the majority tournament.

\subsection{Appending an order with its reverse and replicating a profile}

For an order $\rho$, let $\rho^{\mathrm{rev}}$ denote the reverse order.

\begin{lemma}\label{lem:add-two-rankings}
Let $\Pi$ be a profile on $\cC$, and let $\rho$ be any order on $\cC$.
Appending $\rho$ and $\rho^{\mathrm{rev}}$ increases the Kemeny objective of every aggregate by $\binom{|\cC|}{2}$ and leaves the set of Kemeny-optimal aggregate orders unchanged.
\end{lemma}

\begin{proof}
For every aggregate order $\sigma$ and every candidate pair, exactly one of $\rho$ and $\rho^{\mathrm{rev}}$ disagrees with $\sigma$.
Therefore
\[
\KTdist(\rho,\sigma)
+
\KTdist(\rho^{\mathrm{rev}},\sigma)
=
\binom{|\cC|}{2},
\]
independently of $\sigma$.
Thus appending $\rho$ and $\rho^{\mathrm{rev}}$ adds the same constant to the Kemeny objective of every aggregate order.
\end{proof}

\begin{lemma}\label{lem:replicating-profile}
Let $\Pi=(\pi_1,\dots,\pi_\ell)$ be a profile, and let $h$ be a positive integer.
Replacing every input ranking by $h$ identical copies multiplies the Kemeny objective of every aggregate by $h$ and leaves the set of Kemeny-optimal aggregate orders unchanged.
If a candidate pair is ordered in one direction by $a$ of the original rankings and in the other direction by $\ell-a$ rankings, then it is ordered in those directions by $ha$ and $h(\ell-a)$ rankings, respectively, in the replicated profile.
\end{lemma}

\begin{proof}
For every aggregate order $\sigma$, the Kemeny objective in the replicated profile is
\[
\sum_{r=1}^{\ell}
\sum_{j=1}^{h}
\KTdist(\pi_r,\sigma)
=
h\sum_{r=1}^{\ell}
\KTdist(\pi_r,\sigma)
=
h\Kemeny_\Pi(\sigma).
\]
Each of the $a$ rankings supporting one direction appears $h$ times, as does each of the $\ell-a$ rankings supporting the opposite direction.
The claimed pairwise counts follow.
\end{proof}

\begin{lemma}\label{lem:exact-support-transformation}
Fix integers $q\geq3$ and $s$ with
\[
\frac q2<s\leq\frac{2q}{3},
\]
and put
\[
\Delta=2s-q,
\qquad
r=2q-3s.
\]
Let $\Pi=(\pi_1,\pi_2,\pi_3)$ be an all\nobreakdash-$2$\nobreakdash-to\nobreakdash-$1$ profile on a candidate set $\cC$.
Replicate each input ranking $\Delta$ times and append $r$ pairs $(\rho,\rho^{\mathrm{rev}})$, each consisting of an arbitrary order $\rho$ on $\cC$ and its reverse.
The resulting profile $\Pi_{q,s}$ has exactly $q$ rankings, every candidate pair has support exactly $s$, and its majority tournament is the majority tournament of $\Pi$.
For every aggregate order $\sigma$,
\begin{equation}
\label{eq:exact-support-affine}
\Kemeny_{\Pi_{q,s}}(\sigma)
=
\Delta\Kemeny_\Pi(\sigma)
+
r\binom{|\cC|}{2}.
\end{equation}
Consequently, $\Pi$ and $\Pi_{q,s}$ have exactly the same Kemeny-optimal aggregate orders.
The transformation preserves whether a designated candidate is first in some or every Kemeny-optimal aggregate and whether a designated precedence holds in some or every Kemeny-optimal aggregate.
It also preserves whether a supplied order is Kemeny-optimal or uniquely Kemeny-optimal.
If the majority tournament of $\Pi$ has majority dimension exactly $3$, then the same is true after the transformation.
The transformation is computable in polynomial time.
\end{lemma}

\begin{proof}
The assumptions give a positive integer $\Delta$ and a nonnegative integer $r$.
The final number of rankings is
\[
3\Delta+2r
=
3(2s-q)+2(2q-3s)
=
q.
\]
Every pair is initially split $2$\nobreakdash-to\nobreakdash-$1$.
After replication and the addition of these order--reverse pairs, it is split
\[
(2\Delta+r)\text{--}(\Delta+r)
=
s\text{--}(q-s),
\]
so its support is exactly $s$.
The objective identity and preservation of the complete set of Kemeny-optimal orders follow from \cref{lem:add-two-rankings,lem:replicating-profile}.
Neither operation changes any majority direction, so the majority tournament and its majority dimension are unchanged.
\end{proof}

\subsection{Four rankings in the boundary case \texorpdfstring{$s=q/2$}{s=q/2}}\label{sec:four-ranking-transformations}

The reductions for four rankings use the following construction.
For every arc copy $a=(x,y)$ of a directed multigraph, the output majority digraph contains two arcs $x\to z_a$ and $z_a\to y$.
Hence an order that places $y$ before $x$ reverses at least one of these two arcs.
Every Kemeny-optimal aggregate produced by the construction begins with a vertex of the directed multigraph.
The first application preserves the Kemeny-optimal orders on the original candidate set, together with the possible and necessary first-position and precedence conditions.
The second repeats each tournament arc and adds one arc from each candidate of a supplied order to its successor, converting optimality of the supplied order into uniqueness after the transformation.

\subsubsection{Realizing a directed multigraph with four rankings}\label{sec:four-ranking-multigraph-realization}

For a finite loopless directed multigraph $D$ and a linear order $\sigma$ of its vertices, let $b_D(\sigma)$ denote the number of backward arc copies, counted with multiplicity.

\begin{lemma}
\label{lem:four-ranking-multigraph-realization}
Let $D=(V,A)$ be a finite loopless directed multigraph on a nonempty vertex set; parallel and antiparallel arc copies are allowed.
Fix a labeling $V=\{v_1,\dots,v_n\}$ and an order of the outgoing arc copies of each vertex.
One can construct in polynomial time a profile $\widehat\Pi$ of exactly four unweighted full rankings and, for each $i$, a sequence $C_i$ beginning with $v_i$, such that every candidate pair has support $2$ or $3$.
Let $H$ be the directed graph whose arcs are the majority arcs of $\widehat\Pi$.
The following statements hold.
\begin{enumerate}[label=\textup{(\roman*)}]
\item\label{item:four-ranking-extension}
For an order $\sigma=v_{i_1}\dots v_{i_n}$ of $V$, define
\[
\widehat\sigma=C_{i_1}\dots C_{i_n}.
\]
Then
\[
b_H(\widehat\sigma)=b_D(\sigma).
\]
The first candidate of $\widehat\sigma$ is the first vertex of $\sigma$.

\item\label{item:four-ranking-restriction}
For every order $\omega$ of the candidates of $\widehat\Pi$,
\[
b_H(\omega)\geq b_D(\omega|_V).
\]
If $\omega$ minimizes $b_H$, then its first candidate belongs to $V$ and equals the first vertex of $\omega|_V$.
\end{enumerate}
Consequently, a vertex is first in an order minimizing $b_D$ exactly when it is first in a Kemeny-optimal aggregate of $\widehat\Pi$.
The construction uses $n+2|A|$ candidates, where arc copies are counted with multiplicity.
\end{lemma}

\begin{proof}
For every arc copy $a=(x,y)\in A$, add a candidate $z_a$, called the arc candidate for $a$.
For each vertex $v_i$, list the arc candidates corresponding to its outgoing arc copies in the prescribed order as
\[
z_{i,1},\dots,z_{i,k_i}.
\]
Add candidates $u_{i,1},\dots,u_{i,k_i}$ and put
\[
C_i
=
v_iu_{i,1}z_{i,1}\dots u_{i,k_i}z_{i,k_i},
\]
with $C_i=v_i$ when $k_i=0$.
The candidates $u_{i,j}$ separate consecutive arc candidates in $C_i$; they are used below to show that an order minimizing $b_H$ cannot begin with a candidate outside $V$.
Let $\mathcal U$ be a fixed list of all candidates $u_{i,j}$.
For $v\in V$, let $I(v)$ be a fixed list of the arc candidates $z_a$ for which the head of $a$ is $v$.

Define $\widehat\pi_1$ by concatenating $C_i$ for $i=1,\dots,n$ in increasing order:
\[
\widehat\pi_1=C_1C_2\dots C_n.
\]
Define $\widehat\pi_2$ by concatenating the same sequences in decreasing order:
\[
\widehat\pi_2=C_nC_{n-1}\dots C_1.
\]
Define $\widehat\pi_3$ by concatenating $\mathcal U$ and then, for $i=1,\dots,n$, the sequence $I(v_i)v_i$:
\begin{equation}
\label{eq:four-ranking-realization-rankings}
\widehat\pi_3=\mathcal U\,I(v_1)v_1\,I(v_2)v_2\dots I(v_n)v_n.
\end{equation}
Define $\widehat\pi_4$ by reversing this order of the lists $I(v_i)v_i$, reversing each list $I(v_i)$, and appending $\mathcal U^{\mathrm{rev}}$:
\[
\widehat\pi_4
=
I(v_n)^{\mathrm{rev}}v_n\dots I(v_2)^{\mathrm{rev}}v_2\,I(v_1)^{\mathrm{rev}}v_1\,\mathcal U^{\mathrm{rev}}.
\]
Every original vertex, arc candidate, and candidate $u_{i,j}$ belongs to exactly one of $C_1,\dots,C_n$.
Each arc candidate belongs to exactly one of $I(v_1),\dots,I(v_n)$.
Thus every candidate occurs exactly once in each ranking.

The first two rankings agree exactly on pairs contained in one sequence $C_i$, and they order such pairs as in $C_i$.
The last two rankings disagree on every pair involving a candidate $u_{i,j}$ because $\mathcal U$ occurs first in $\widehat\pi_3$ and last in $\widehat\pi_4$, with its internal order reversed.
Among the remaining pairs, the last two rankings agree exactly on
\[
z_a\before y
\qquad\text{for each arc copy }a=(x,y)\in A.
\]
This follows because the sequences $I(v_i)v_i$ occur in opposite orders in the last two rankings, and every list $I(v_i)$ is reversed.
Because $D$ has no loops, $z_a$ belongs to the sequence beginning with $x$, whereas $y$ belongs to a different sequence.
The two families of pairs on which a pair of rankings agree are therefore disjoint.
Every pair in either family has support $3$, and every other pair is tied.

Let $H$ be the directed graph consisting of the majority arcs of $\widehat\Pi$.
Its arcs are precisely
\begin{equation}\label{eq:four-ranking-realization-arcs}
\bigcup_{i=1}^{n}
\{u\to v:u\before_{C_i}v\}
\quad\cup\quad
\{z_a\to y:a=(x,y)\in A\}.
\end{equation}
Every pair with a majority arc contributes $1$ to the Kemeny objective when the arc is forward and $3$ when it is backward, while every tied pair contributes $2$ independently of the aggregate order.
Hence there is a computable constant $B$ such that
\begin{equation}\label{eq:four-ranking-realization-score}
\Kemeny_{\widehat\Pi}(\omega)
=
B+2b_H(\omega)
\end{equation}
for every aggregate order $\omega$.
It therefore suffices to minimize $b_H$.

Let $\sigma=v_{i_1}\dots v_{i_n}$ be an order of $V$ and put $\widehat\sigma=C_{i_1}\dots C_{i_n}$.
Every majority arc between two candidates in the same sequence $C_i$ is forward in $\widehat\sigma$.
For an arc copy $a=(x,y)$, the comparison $z_a\to y$ is backward exactly when $y\before_\sigma x$.
Therefore
\begin{equation}\label{eq:four-ranking-extension-cost}
b_H(\widehat\sigma)=b_D(\sigma).
\end{equation}
The first candidate of $\widehat\sigma$ is the first vertex of $\sigma$, which proves \cref{item:four-ranking-extension}.

Now let $\omega$ be any order of the candidates of $\widehat\Pi$ and put $\sigma=\omega|_V$.
For every backward arc copy $a=(x,y)$ of $D$, at least one of the two comparisons
\[
x\to z_a
\qquad\text{and}\qquad
z_a\to y
\]
is backward in $\omega$.
Choose one such comparison for each backward arc copy.
The chosen comparisons are distinct because different arc copies have different arc candidates.
Consequently,
\begin{equation}\label{eq:four-ranking-restriction-cost}
b_H(\omega)\geq b_D(\sigma).
\end{equation}
Taking minima in \cref{eq:four-ranking-extension-cost,eq:four-ranking-restriction-cost} shows that the two minimum values are equal and that the restriction of every order minimizing $b_H$ minimizes $b_D$.

It remains to show that every order minimizing $b_H$ begins with a vertex of $D$.
Suppose that a candidate $r\notin V$ is first in an order $\omega$, and let $C_i$ be the sequence containing $r$.
Then $r\before_\omega v_i$.
If $r=u_{i,j}$, the comparison $v_i\to u_{i,j}$ is backward.
If $r=z_{i,j}$, consider the candidate $u_{i,j}$ immediately preceding it in $C_i$.
When $u_{i,j}\before_\omega v_i$, the comparison $v_i\to u_{i,j}$ is backward.
When $v_i\before_\omega u_{i,j}$, the comparison $u_{i,j}\to z_{i,j}$ is backward.
In each case the identified backward comparison contains a candidate $u_{i,j}$ and is distinct from all comparisons selected in the preceding paragraph.
Thus
\[
b_H(\omega)
\geq
b_D(\omega|_V)+1
\geq
\min_\sigma b_D(\sigma)+1.
\]
By \cref{eq:four-ranking-extension-cost}, if $v_{i_1}\dots v_{i_n}$ minimizes $b_D$, then $C_{i_1}\dots C_{i_n}$ has backward-arc count $\min_\sigma b_D(\sigma)$.
Hence $\omega$ does not minimize $b_H$.
Every order minimizing $b_H$ begins with a vertex in $V$, and its first candidate equals the first vertex of its restriction.
This proves \cref{item:four-ranking-restriction} and the statement about possible first vertices.

The construction uses one arc candidate and one candidate $u_{i,j}$ for every arc copy, so it has $n+2|A|$ candidates and is polynomial.
\end{proof}

\subsubsection{A four-ranking transformation preserving optimal orders}\label{sec:three-to-four}

\begin{lemma}\label{lem:three-to-four}
Given any profile of exactly three rankings on a candidate set $S$, one can construct in polynomial time a profile of exactly four rankings on a candidate set containing $S$ such that every candidate pair in the output profile has support $2$ or $3$.
After fixing a labeling $S=\{v_1,\dots,v_n\}$, the construction produces sequences $C_1,\dots,C_n$ beginning with the corresponding candidates.
Restricting a Kemeny-optimal aggregate of the output profile to $S$ gives a Kemeny-optimal aggregate of the input profile, and if $v_{i_1}\dots v_{i_n}$ is Kemeny-optimal for the input profile, then $C_{i_1}\dots C_{i_n}$ is Kemeny-optimal for the output profile.
The transformation preserves \KemenyWinner{}, \KemenyUniqueWinner{}, \KemenyPossiblePrecedence{}, and \KemenyNecessaryPrecedence{}.
For a supplied input order $v_{i_1}\dots v_{i_n}$, the output order $C_{i_1}\dots C_{i_n}$ is Kemeny-optimal exactly when the input order is Kemeny-optimal.
The construction also gives a polynomial-time many-one reduction from \KemenyScore{} with three rankings to \KemenyScore{} with four rankings.
\end{lemma}

\begin{proof}
Let $\Pi=(\pi_1,\pi_2,\pi_3)$ be a profile on a candidate set $S$ of size $n$.
For every majority arc $x\to y$, create $w_\Pi(x,y)$ labeled arc copies directed from $x$ to $y$, and let $D_\Pi$ be the resulting loopless directed multigraph.
By \cref{lem:tournament-identity},
\begin{equation}\label{eq:three-ranking-multigraph-score}
\Kemeny_\Pi(\sigma)
=
B_\Pi+b_{D_\Pi}(\sigma)
\end{equation}
for every order $\sigma$ of $S$.

Apply \cref{lem:four-ranking-multigraph-realization} to $D_\Pi$, using any labeling of $S$ and any order of the outgoing arc copies.
Let $\widehat\Pi$ be the resulting four-ranking profile and let $H$ be the directed graph consisting of its majority arcs.
Every candidate pair in $\widehat\Pi$ has support $2$ or $3$.
If $B_4$ denotes the constant in \cref{eq:four-ranking-realization-score}, then
\begin{equation}\label{eq:four-ranking-output-score}
\Kemeny_{\widehat\Pi}(\omega)
=
B_4+2b_H(\omega).
\end{equation}
The constant $B_4$ is computable from the output profile.

The extension and restriction statements in \cref{lem:four-ranking-multigraph-realization} give
\[
\min_\omega b_H(\omega)
=
\min_\sigma b_{D_\Pi}(\sigma).
\]
Together with \cref{eq:three-ranking-multigraph-score,eq:four-ranking-output-score}, they show that every Kemeny-optimal aggregate of $\widehat\Pi$ restricts to a Kemeny-optimal aggregate of $\Pi$ and that $C_{i_1}\dots C_{i_n}$ is Kemeny-optimal for $\widehat\Pi$ whenever $v_{i_1}\dots v_{i_n}$ is Kemeny-optimal for $\Pi$.
Conversely, if $C_{i_1}\dots C_{i_n}$ is Kemeny-optimal for $\widehat\Pi$, then \cref{eq:four-ranking-extension-cost} shows that $v_{i_1}\dots v_{i_n}$ is Kemeny-optimal for $\Pi$.
This proves the reduction from three-ranking to four-ranking \KemenyConsensusRecognition{}.

By \cref{item:four-ranking-restriction}, a candidate in $S$ is first in some Kemeny-optimal aggregate of $\widehat\Pi$ exactly when it is first in some Kemeny-optimal aggregate of $\Pi$, and no candidate outside $S$ is first in a Kemeny-optimal aggregate of $\widehat\Pi$.
This preserves \KemenyWinner{} and \KemenyUniqueWinner{}.
Restriction to $S$ and the orders $C_{i_1}\dots C_{i_n}$ preserve the relative order of every pair in $S$, which preserves \KemenyPossiblePrecedence{} and \KemenyNecessaryPrecedence{}.

To obtain the reduction for \KemenyScore{}, consider a threshold $\kappa$.
If $\kappa<B_\Pi$, output the fixed four-ranking no-instance on candidates $a,b$ with two rankings $a\before b$, two rankings $b\before a$, and threshold $1$; its minimum score is $2$.
Otherwise set
\[
\kappa'
=
B_4+2(\kappa-B_\Pi).
\]
The two score identities and the equality of the minimum backward-arc counts show that $\Kemeny_\Pi^\star\leq\kappa$ exactly when $\Kemeny_{\widehat\Pi}^\star\leq\kappa'$.
Finally,
\[
|A(D_\Pi)|
=
\sum_{x\to y}w_\Pi(x,y)
\leq
3\binom n2,
\]
so \cref{lem:four-ranking-multigraph-realization} gives a polynomial-size output.
\end{proof}

\subsubsection{A four-ranking transformation for unique consensus recognition}\label{sec:unique-consensus-four}

Replication and the addition of pairs consisting of an order and its reverse cannot transform a profile of three rankings into one of exactly four rankings.
The next construction therefore treats unique consensus recognition separately.
It repeats each majority arc sufficiently many times and adds one arc from each candidate of the supplied order to its successor.

\begin{lemma}
\label{lem:unique-consensus-three-to-four}
Let $\Pi=(\pi_1,\pi_2,\pi_3)$ be an all\nobreakdash-$2$\nobreakdash-to\nobreakdash-$1$ profile on a candidate set $S$, and let $\tau$ be an aggregate order.
One can construct in polynomial time a profile $\widehat\Pi$ of exactly four unweighted full rankings and an aggregate order $\widehat\tau$ such that $\widehat\tau$ is the unique Kemeny-optimal aggregate for $\widehat\Pi$ if and only if $\tau$ is Kemeny-optimal for $\Pi$.
Every unordered candidate pair in $\widehat\Pi$ has support $2$ or $3$.
\end{lemma}

\begin{proof}[Proof of \cref{lem:unique-consensus-three-to-four}]
Write $\tau=x_1x_2\dots x_n$ and let $T_\Pi$ be the majority tournament of $\Pi$.
Since $\Pi$ is all\nobreakdash-$2$\nobreakdash-to\nobreakdash-$1$, \cref{eq:all-two-to-one-score} shows that an order $\sigma$ of $S$ is Kemeny-optimal exactly when it minimizes $b_\Pi(\sigma)$.
Set $L=n+1$.

Construct a loopless directed multigraph $D_\tau$ on $S$ as follows.
For every majority arc $a=(x,y)$ of $T_\Pi$, add $L$ labeled copies $a_1,\dots,a_L$ directed from $x$ to $y$.
For every $i<n$, also add one arc $e_i=(x_i,x_{i+1})$.
At each $x_i$, order the outgoing copies of majority arcs arbitrarily and place $e_i$ last when it exists.
Apply \cref{lem:four-ranking-multigraph-realization} to $D_\tau$ with the labeling $x_1,\dots,x_n$.
Let $\widehat\Pi$ be the resulting four-ranking profile, let $H$ be the directed graph consisting of its majority arcs, and let $C_1,\dots,C_n$ be the sequences supplied by that lemma.
Write $z_{a,r}$ for the arc candidate corresponding to $a_r$ and $d_i$ for the arc candidate corresponding to $e_i$.
Every candidate pair in $\widehat\Pi$ has support $2$ or $3$.
By \cref{eq:four-ranking-realization-score}, the Kemeny-optimal aggregates of $\widehat\Pi$ are exactly the orders minimizing $b_H$.

Set
\[
\widehat\tau=C_1C_2\dots C_n.
\]
For an order $\sigma$ of $S$, let $r_\tau(\sigma)=|\{i\in\{1,\dots,n-1\}:x_{i+1}\before_\sigma x_i\}|$, the number of consecutive comparisons of $\tau$ that are reversed in $\sigma$.
By \cref{eq:four-ranking-extension-cost},
\begin{equation}
\label{eq:unique-four-supplied-cost}
b_H(\widehat\tau)
=
b_{D_\tau}(\tau)
=
L b_\Pi(\tau),
\end{equation}
because every arc $e_i$ is forward in $\tau$.
For every order $\omega$ of the candidates of $\widehat\Pi$ with restriction $\sigma=\omega|_S$, \cref{eq:four-ranking-restriction-cost} gives
\begin{equation}
\label{eq:unique-four-base-lower}
b_H(\omega)
\geq
b_{D_\tau}(\sigma)
=
L b_\Pi(\sigma)+r_\tau(\sigma).
\end{equation}

Suppose first that $\tau$ is Kemeny-optimal for $\Pi$.
If $\sigma\neq\tau$, then $b_\Pi(\sigma)\geq b_\Pi(\tau)$ and $r_\tau(\sigma)\geq1$, because an order satisfying $x_i\before x_{i+1}$ for every $i<n$ must equal $\tau$.
Thus \cref{eq:unique-four-base-lower} gives
\[
b_H(\omega)
\geq
L b_\Pi(\tau)+1
>
b_H(\widehat\tau).
\]

Now consider an order $\omega\neq\widehat\tau$ whose restriction to $S$ is $\tau$.
Every comparison between consecutive candidates in $\widehat\tau$ is a majority arc.
The consecutive comparisons inside each $C_i$ follow from \cref{eq:four-ranking-realization-arcs}.
Because $e_i$ is last among the outgoing arcs of $x_i$, its arc candidate $d_i$ is the last candidate in $C_i$, and $d_i\to x_{i+1}$ is a majority arc.
An order that follows all these consecutive comparisons must equal $\widehat\tau$.
Hence $\omega$ reverses at least one of them.

For each of the $L b_\Pi(\tau)$ copies of a majority arc that is backward in $\tau$, choose one backward comparison among $x\to z_{a,r}$ and $z_{a,r}\to y$, as in the proof of \cref{eq:four-ranking-restriction-cost}.
The selected comparisons are distinct.
A consecutive comparison inside some $C_i$ contains a candidate $u_{i,j}$ and therefore is not selected.
A consecutive comparison $d_i\to x_{i+1}$ corresponds to the arc $e_i$, which is forward in $\tau$, and is also not selected.
The reversed consecutive comparison consequently contributes in addition to the selected $L b_\Pi(\tau)$ comparisons, and
\[
b_H(\omega)
\geq
L b_\Pi(\tau)+1
>
b_H(\widehat\tau).
\]
Thus $\widehat\tau$ is the unique Kemeny-optimal aggregate of $\widehat\Pi$.

Conversely, suppose that $\tau$ is not Kemeny-optimal for $\Pi$.
Choose an order $\sigma$ of $S$ with $b_\Pi(\sigma)\leq b_\Pi(\tau)-1$, and let $\widehat\sigma$ be the extension defined in \cref{item:four-ranking-extension}.
Then
\begin{equation}
\label{eq:unique-four-extension-cost}
b_H(\widehat\sigma)
=
L b_\Pi(\sigma)+r_\tau(\sigma)
\leq
L(b_\Pi(\tau)-1)+(n-1)
<
L b_\Pi(\tau)
=
b_H(\widehat\tau),
\end{equation}
where the strict inequality uses $L=n+1$.
Thus $\widehat\tau$ is not Kemeny-optimal.

The multigraph $D_\tau$ has $L\binom n2+n-1$ arc copies.
By \cref{lem:four-ranking-multigraph-realization}, the output has $n+2\left(L\binom n2+n-1\right)=O(n^3)$ candidates and is computable in polynomial time.
The two implications prove the claimed equivalence.
\end{proof}

\subsection{Complete support classifications}

\begin{theorem}
\label{thm:support-dichotomy}
Fix integers $q\geq3$ and $s$ with
\[
\left\lceil\frac q2\right\rceil\leq s\leq q.
\]
Restrict to profiles $\Pi$ of exactly $q$ unweighted full rankings satisfying $\support_\Pi(x,y)\geq s$ for all distinct candidates $x,y$.
Then the following classifications hold.
\begin{enumerate}[label=\textup{(\roman*)}]
\item \KemenyScore{} is NP-complete when $3s\leq2q$ and polynomial-time solvable when $3s>2q$.
\item \KemenyWinner{}, \KemenyUniqueWinner{}, \KemenyPossiblePrecedence{}, and \KemenyNecessaryPrecedence{} are $\ThetaTwoP$-complete when $3s\leq2q$ and polynomial-time solvable when $3s>2q$.
\item \KemenyConsensusRecognition{} and \UniqueKemenyConsensusRecognition{} are coNP-complete when $3s\leq2q$ and polynomial-time solvable when $3s>2q$.
\end{enumerate}
If $q/2<s\leq2q/3$, the NP-hardness, $\ThetaTwoP$-hardness, and coNP-hardness conclusions persist when every pair has support exactly $s$ and the majority tournament has majority dimension exactly $3$.
In the boundary case $s=q/2$, which occurs only for even $q$, the same conclusions persist when every pair has support in $\{s,s+1\}$.
For every even $q\geq6$, they also hold with exact support $s+1$ and majority dimension exactly $3$.
When $3s>2q$, the majority relation is a transitive tournament, and its unique topological order is the unique Kemeny-optimal aggregate; the score, winner, precedence, and recognition problems can then be solved in $O(q|\cC|^2)$ time.
\end{theorem}

\begin{proof}
Membership in NP for \KemenyScore{} follows by evaluating a proposed aggregate.
For the winner and precedence problems, parallel NP threshold queries determine the unrestricted minimum Kemeny score and the minimum Kemeny score under the specified first-position or precedence constraint, which gives membership in $\ThetaTwoP$.
For \KemenyUniqueWinner{}, use the constraint that the designated candidate is not first.
A strictly better aggregate certifies a no-instance of \KemenyConsensusRecognition{}, while a distinct aggregate with no larger objective value certifies a no-instance of \UniqueKemenyConsensusRecognition{}; both recognition problems therefore belong to coNP\@.

Suppose first that $3s>2q$.
Then $s>q/2$, so every pair has a unique majority direction.
If the majority tournament contained a directed triangle $x\to y\to z\to x$, the three arcs would have total support at least $3s$.
A strict total order can agree with at most two arcs of a directed triangle, so the $q$ input rankings contribute at most $2q$ such agreements, a contradiction.
The majority tournament is therefore transitive.
Its unique topological order follows every majority direction, minimizes every pairwise contribution simultaneously, and is the unique Kemeny-optimal aggregate.
Computing the pairwise counts and this order takes $O(q|\cC|^2)$ time, after which the score, winner, precedence, and recognition problems are decided by direct inspection.

Now suppose that $q/2<s\leq2q/3$.
Apply \cref{lem:exact-support-transformation} to the corresponding three-ranking instances from \cref{thm:main,thm:three-ranking-optimality}.
The affine identity in \cref{eq:exact-support-affine} gives the threshold transformation for \KemenyScore{}, while preservation of the complete set of Kemeny-optimal orders gives the winner, precedence, and recognition reductions.
The output has exact support $s$ and the same majority tournament, so majority dimension remains exactly $3$.

It remains to treat the boundary case $s=q/2$.
For \KemenyScore{}, \KemenyWinner{}, \KemenyUniqueWinner{}, \KemenyPossiblePrecedence{}, \KemenyNecessaryPrecedence{}, and \KemenyConsensusRecognition{}, apply \cref{lem:three-to-four} to the corresponding three-ranking instance; in the recognition case, replace the supplied input order $v_{i_1}\dots v_{i_n}$ by $C_{i_1}\dots C_{i_n}$.
For \UniqueKemenyConsensusRecognition{}, apply \cref{lem:unique-consensus-three-to-four} to an instance of \KemenyConsensusRecognition{} from \cref{thm:three-ranking-optimality}.
These four-ranking constructions have supports in $\{2,3\}$ and preserve the threshold comparison, designated first-position or precedence condition, or supplied-order recognition condition used in the corresponding reduction.
If $q>4$, append $(q-4)/2$ pairs consisting of an order and its reverse; the supports become $s$ and $s+1$, and the Kemeny-optimal orders are unchanged.
For every even $q\geq6$, applying \cref{lem:exact-support-transformation} with support $s+1$ to the corresponding three-ranking instance from \cref{thm:main,thm:three-ranking-optimality} gives exact support $s+1$ and majority dimension exactly $3$.
All transformations are polynomial because $q$ and $s$ are fixed.
\end{proof}

Specializing \cref{thm:support-dichotomy} to exact support gives the following classification.

\begin{corollary}\label{cor:exact-support-trichotomy}
Fix $q\geq3$ and an integer $s$ with $\lceil q/2\rceil\leq s\leq q$, and restrict to profiles of exactly $q$ rankings in which every candidate pair has support exactly $s$.
If $q$ is even and $s=q/2$, every aggregate is Kemeny-optimal, and \KemenyScore{}, the four winner and precedence problems, and both recognition problems are polynomial-time solvable.
If $q/2<s\leq2q/3$, \KemenyScore{} is NP-complete, the four winner and precedence problems are $\ThetaTwoP$-complete, and both recognition problems are coNP-complete, even when the majority tournament has majority dimension exactly $3$.
If $s>2q/3$, the majority tournament is transitive, its unique topological order is the unique Kemeny-optimal aggregate, and \KemenyScore{}, the four winner and precedence problems, and both recognition problems are polynomial-time solvable.
\end{corollary}

\begin{proof}
At $s=q/2$, every pair is tied and contributes $q/2$ independently of the aggregate, so every aggregate is Kemeny-optimal.
The other two cases follow from \cref{thm:support-dichotomy}.
\end{proof}

\subsection{Fractional support thresholds, uniform margins, and tournament consequences}

\begin{corollary}\label{cor:two-thirds-support}
Fix a rational number $\alpha\in[1/2,1]$.
Consider profiles $\Pi$ of $q$ rankings satisfying $\support_\Pi(x,y)\geq\alpha q$ for all distinct candidates $x,y$, where $q$ is part of the input.
\KemenyScore{} is
\[
\begin{cases}
\text{NP-complete} &\text{ if }\alpha\leq2/3,\\[1mm]
\text{polynomial-time solvable} &\text{ if }\alpha>2/3,
\end{cases}
\]
\KemenyWinner{}, \KemenyUniqueWinner{}, \KemenyPossiblePrecedence{}, and \KemenyNecessaryPrecedence{} are
\[
\begin{cases}
\ThetaTwoP\text{-complete} &\text{ if }\alpha\leq2/3,\\[1mm]
\text{polynomial-time solvable} &\text{ if }\alpha>2/3,
\end{cases}
\]
and \KemenyConsensusRecognition{} and \UniqueKemenyConsensusRecognition{} are
\[
\begin{cases}
\text{coNP-complete} &\text{ if }\alpha\leq2/3,\\[1mm]
\text{polynomial-time solvable} &\text{ if }\alpha>2/3.
\end{cases}
\]
When $\alpha>2/3$, the unique Kemeny-optimal aggregate is the unique topological order of the transitive majority tournament and can be computed in $O(q|\cC|^2)$ time.
\end{corollary}

\begin{proof}
If $\alpha\leq2/3$, the all\nobreakdash-$2$\nobreakdash-to\nobreakdash-$1$ instances from \cref{thm:main,thm:three-ranking-optimality} satisfy the promise because each majority direction has support exactly $2/3$.
These instances give NP-hardness for \KemenyScore{}, $\ThetaTwoP$-hardness for the winner and precedence problems, and coNP-hardness for the recognition problems.
A proposed aggregate certifies a yes-instance of \KemenyScore{}; parallel NP threshold queries give the $\ThetaTwoP$ upper bounds; a strictly better aggregate certifies that a supplied order is not Kemeny-optimal; and a distinct aggregate with no larger score certifies that a supplied order is not uniquely Kemeny-optimal.

Now suppose that $\alpha>2/3$ and let $q$ be the profile size of an input instance.
For $q\geq3$, put $s=\lceil\alpha q\rceil$.
Then $s>q/2$ and $3s\geq3\alpha q>2q$, so the polynomial-time cases of \cref{thm:support-dichotomy} apply.
For $q=1$, the input order is the unique Kemeny-optimal aggregate.
For $q=2$, the support promise forces the two rankings to agree on every pair, so their common order is the unique Kemeny-optimal aggregate.
\end{proof}

The classical two-thirds quota theorem states that if every pairwise majority is supported by more than two thirds of the rankings, then the majority tournament is transitive~\cite{Vidu2000}.
The three-ranking constructions in \cref{thm:main,thm:three-ranking-optimality} attain the boundary value $2/3$.
Lowering the minimum support only enlarges the allowed instance class, so the same constructions prove hardness for every threshold at most $2/3$.

\begin{corollary}\label{cor:uniform-margin}
Fix an integer $q\geq1$ and an integer $\Delta\in\{0,\dots,q\}$ with $\Delta\equiv q\pmod 2$.
\KemenyScore{} restricted to profiles of exactly $q$ unweighted full rankings with uniform pairwise majority margin $\Delta$ is
\[
\begin{cases}
\text{polynomial-time solvable} &\text{ if }\Delta=0\text{ or }3\Delta>q,\\[1mm]
\text{NP-complete} &\text{ if }0<3\Delta\leq q,
\end{cases}
\]
\KemenyWinner{}, \KemenyUniqueWinner{}, \KemenyPossiblePrecedence{}, and \KemenyNecessaryPrecedence{} on the same profile class are
\[
\begin{cases}
\text{polynomial-time solvable} &\text{ if }\Delta=0\text{ or }3\Delta>q,\\[1mm]
\ThetaTwoP\text{-complete} &\text{ if }0<3\Delta\leq q,
\end{cases}
\]
and \KemenyConsensusRecognition{} and \UniqueKemenyConsensusRecognition{} are
\[
\begin{cases}
\text{polynomial-time solvable} &\text{ if }\Delta=0\text{ or }3\Delta>q,\\[1mm]
\text{coNP-complete} &\text{ if }0<3\Delta\leq q.
\end{cases}
\]
\end{corollary}

\begin{proof}
If $\Delta=0$, every pair is tied and every aggregate has the same Kemeny score, so the score, winner, precedence, and recognition problems are polynomial-time solvable.
The positive-margin cases with $q\leq2$ are also polynomial-time solvable.
Assume that $q\geq3$ and $\Delta>0$, and put $s=(q+\Delta)/2$.
A profile with uniform pairwise majority margin $\Delta$ has pairwise support exactly $s>q/2$, and $3s\leq2q$ if and only if $3\Delta\leq q$.
The result follows from \cref{thm:support-dichotomy}.
\end{proof}

The fixed-$q$ classifications in \cref{thm:support-dichotomy} are summarized in \cref{tab:support-dichotomy}.

\begin{table}[t]
\centering
\scriptsize
\caption{Complexity for fixed $q\geq3$ when every candidate pair has support at least $s$.}
\label{tab:support-dichotomy}

\begingroup
\renewcommand{\arraystretch}{1.15}
\begin{tabular}{@{}lccc@{}}
\toprule
Minimum pairwise support $s$ & Score & \shortstack{Winner and\\precedence problems} & \shortstack{Recognition\\problems}\\
\midrule
$s=q/2$ for even $q$ & NP-complete & $\ThetaTwoP$-complete & coNP-complete\\
$q/2<s\leq2q/3$ & NP-complete & $\ThetaTwoP$-complete & coNP-complete\\
$s>2q/3$ & polynomial time & polynomial time & polynomial time\\
\bottomrule
\end{tabular}
\endgroup

\medskip
\parbox{0.98\linewidth}{The winner and precedence column includes both winner problems and both precedence problems, and the recognition column includes ordinary and unique recognition. The first row concerns a minimum-support condition, not exact support. At exact support $q/2$, every pair is tied and every aggregate is Kemeny-optimal, so all three problem columns are polynomial-time solvable; hardness in the first row uses profiles whose supports belong to $\{q/2,q/2+1\}$.}
\end{table}

\begin{corollary}
\label{cor:tournament-fas}
Fix integers $q\geq3$ and $s$ with
\[
\frac q2<s\leq\frac{2q}{3}.
\]
Consider majority tournaments supplied with a $q$-ranking realization in which every candidate pair has support exactly $s$.
The following hardness results persist under the additional promise that the tournament has majority dimension exactly $3$:
\begin{enumerate}[label=\textup{(\roman*)}]
\item Minimum Feedback Arc Set in Tournaments is NP-hard;
\item \SlaterWinner{}, \SlaterUniqueWinner{}, \SlaterPossiblePrecedence{}, and \SlaterNecessaryPrecedence{} are $\ThetaTwoP$-hard;
\item \SlaterConsensusRecognition{} and \UniqueSlaterConsensusRecognition{} are coNP-hard.
\end{enumerate}
Without the promise on majority dimension, Minimum Feedback Arc Set in Tournaments belongs to NP, the four Slater winner and precedence problems belong to $\ThetaTwoP$, and the two Slater recognition problems belong to coNP.
In particular, all three conclusions hold for tournaments induced by three rankings in which every pair is split $2$\nobreakdash-to\nobreakdash-$1$.
\end{corollary}

\begin{proof}
For exact support $s$, every candidate pair has the common positive majority margin $\Delta=2s-q$.
By \cref{lem:uniform-score},
\[
\Kemeny_\Pi(\sigma)
=
(q-s)\binom{|\cC|}{2}
+
\Delta b_\Pi(\sigma).
\]
Consequently, Kemeny-optimal orders are exactly Slater orders.
The exact-support clause of \cref{thm:support-dichotomy} gives $\ThetaTwoP$-hardness for the four Slater winner and precedence problems and coNP-hardness for the two Slater recognition problems.
The standard parallel-query argument for backward-arc counts gives membership in $\ThetaTwoP$.
A strictly better order certifies that a supplied order is not Slater-optimal, while a distinct order with no larger backward-arc count certifies that it is not uniquely Slater-optimal; neither upper bound uses the promise that the majority dimension is exactly $3$.

For the feedback-arc-set decision problem, membership in NP follows by checking the backward arcs of a proposed order.
By the exact-support clause of \cref{thm:support-dichotomy}, \KemenyScore{} remains NP-hard with exact support $s$ and majority dimension exactly $3$.
For a profile with exact support $s$, the quantity $B_\Pi$ defined in \cref{lem:tournament-identity} equals $(q-s)\binom{|\cC|}{2}$.
For a Kemeny threshold $\kappa\geq B_\Pi$, set
\[
k=
\left\lfloor
\frac{\kappa-B_\Pi}{\Delta}
\right\rfloor.
\]
Then
\[
\Kemeny_\Pi(\sigma)\leq\kappa
\quad\text{if and only if}\quad
b_\Pi(\sigma)\leq k.
\]
If $\kappa<B_\Pi$, output a fixed directed-triangle no-instance with backward-arc threshold $0$ and a realization having exactly $q$ rankings and support exactly $s$, obtained from the cyclic three-ranking profile by replication and by appending pairs consisting of an order and its reverse.
This gives a polynomial-time many-one reduction.
\end{proof}

\subsection{The support boundary for four rankings}

\begin{corollary}
\label{cor:four-margin-heterogeneity}
For exactly four rankings and pairwise supports restricted to $\{2,3\}$, \KemenyScore{} is NP-complete, the four winner and precedence problems are $\ThetaTwoP$-complete, and \KemenyConsensusRecognition{} and \UniqueKemenyConsensusRecognition{} are coNP-complete.
In contrast, \KemenyScore{}, the four winner and precedence problems, and the two recognition problems are polynomial-time solvable when every pair has support at least $3$.
Consequently, for every fixed uniform pairwise majority margin, all seven problems are polynomial-time solvable on profiles of four rankings.
\end{corollary}

\begin{proof}
The transformation in \cref{lem:three-to-four} proves NP-hardness for \KemenyScore{}, $\ThetaTwoP$-hardness for the four winner and precedence problems, and coNP-hardness for \KemenyConsensusRecognition{}; it preserves the threshold comparison, the designated first-position and precedence conditions, and optimality of the supplied order while producing only support-$2$ and support-$3$ pairs.
\UniqueKemenyConsensusRecognition{} hardness follows from the uniqueness-preserving transformation in \cref{lem:unique-consensus-three-to-four} applied to the restricted \KemenyConsensusRecognition{} problem in \cref{thm:three-ranking-optimality}.
The tractability statement is the case $q=4$ and $s=3$ of \cref{thm:support-dichotomy}.
For four rankings, the possible uniform pairwise majority margins are $0,2,4$, each of which is tractable by \cref{cor:uniform-margin}.
\end{proof}

\section{Consequences and extensions}\label{sec:consequences}

The preceding results translate into permutation-median, tournament, crossing-minimization, and Mallows-model formulations.
A six-copy construction then transfers NP-hardness of \KemenyScore{} to pairwise-equidistant profiles and yields results for the center and other distance objectives.
The final subsection combines these results with the support classifications to give the fixed-profile-size boundary.

\subsection{Equivalent formulations}\label{sec:direct-consequences}

The permutation and tournament consequences below use the profile constructed in the proof of \cref{thm:main} without modification.

Identifying the candidate set with $\{1,\dots,N\}$ turns each ranking into a permutation, and Kendall--tau distance becomes inversion distance.
Biedl, Brandenburg, and Deng call the resulting optimization problem \emph{permutation crossing minimization} and denote its three-permutation version by $\textnormal{\textsc{PCM}}$-$3$~\cite{BiedlBrandenburgDeng2009}.

\begin{corollary}\label{cor:permutation-median}
For three input permutations, the threshold decision version of $\textnormal{\textsc{PCM}}$-$3$, equivalently \KemenyScore{}, is NP-complete, while computing a Kendall--tau median, equivalently an optimal $\textnormal{\textsc{PCM}}$-$3$ permutation, is NP-hard.
It is $\ThetaTwoP$-complete to decide whether a designated element is first in some or every median and whether one designated element precedes another in some or every median.
Verifying whether a supplied permutation is a median or the unique median is coNP-complete.
These statements hold even when every pair of elements appears in both relative orders among the three inputs and their majority tournament has majority dimension exactly $3$.
Equivalently, the same results hold for $1$-medians of three vertices of the permutahedron under its graph metric.
\end{corollary}

\begin{proof}
Identifying rankings with permutations preserves Kendall--tau distance, the Kemeny score, the winner and precedence conditions, and whether a supplied order is Kemeny-optimal or uniquely Kemeny-optimal.
Under this identification, the threshold decision version of $\textnormal{\textsc{PCM}}$-$3$ is exactly \KemenyScore{}, and an algorithm computing an optimal common permutation would decide the threshold problem.
The conditions that both orientations occur for every element pair and that the majority dimension is exactly $3$ are unchanged.
The graph distance in the $1$-skeleton of the permutahedron is Kendall--tau distance, so the claims follow from \cref{thm:main,thm:three-ranking-optimality}.
\end{proof}

For the instances produced by the reduction, the exact-value formula and solution-recovery procedure in \cref{prop:maxcut-recovery} apply unchanged to this permutation formulation.
The identity in \cref{eq:all-two-to-one-score} also expresses the Kemeny objective as an additive constant plus the number of backward arcs in the majority tournament.
The same profile therefore gives the following tournament formulation.

\begin{corollary}\label{cor:tournament-fas-three}
Let $\Pi$ be an all\nobreakdash-$2$\nobreakdash-to\nobreakdash-$1$ profile of three rankings on $N$ candidates, and let $T_\Pi$ be its majority tournament.
For every aggregate order $\sigma$,
\[
\Kemeny_\Pi(\sigma)=\binom N2+b_{T_\Pi}(\sigma),
\qquad
\Kemeny_\Pi^\star=\binom N2+\fas(T_\Pi).
\]
Hence the Kemeny-optimal aggregate orders of $\Pi$ are exactly the Slater orders of $T_\Pi$.
Moreover, \TournamentFAS{} is NP-complete when the input tournament is represented by an all\nobreakdash-$2$\nobreakdash-to\nobreakdash-$1$ profile of three rankings.
\end{corollary}

\begin{proof}
The first identity is \cref{eq:all-two-to-one-score}.
Minimizing it gives the second identity and shows that the Kemeny-optimal orders are exactly the Slater orders.
For the complexity statement, the reduction proving \cref{thm:main} produces an all\nobreakdash-$2$\nobreakdash-to\nobreakdash-$1$ profile on $N$ candidates and a threshold of the form $\binom N2+k$, where $k\geq0$.
The first identity shows that this threshold is attained exactly when the represented tournament has an order with at most $k$ backward arcs.
Given the three rankings, one can check the all\nobreakdash-$2$\nobreakdash-to\nobreakdash-$1$ condition, construct the majority tournament, and count the backward arcs of a proposed order in polynomial time, so the problem is in NP.
\end{proof}

For a graph $G$, let $\Pi_G$ be the profile produced by the reduction, and let $N$ be its number of candidates.
Substituting $\Kemeny_{\Pi_G}^\star=\binom N2+\fas(T_{\Pi_G})$ into \cref{eq:maxcut-recovery} gives
\[
\maxcutvalue(G)=\left\lceil\frac{2}{M}\bigl(\beta-\fas(T_{\Pi_G})\bigr)\right\rceil.
\]
Moreover, every Slater order of $T_{\Pi_G}$ is Kemeny-optimal for $\Pi_G$, so \cref{prop:maxcut-recovery} transforms any such order into a maximum cut of $G$ in polynomial time.

For a perfect matching drawn between two ordered layers, the crossing count is the Kendall--tau distance between the two endpoint orders~\cite{BiedlBrandenburgDeng2009}.
For the direct reduction, this correspondence identifies the total crossing count with the Kemeny objective, so \cref{prop:maxcut-recovery} gives the minimum total crossing count and recovers a maximum cut from any optimal common free-layer order.
The stronger pairwise-equidistant sum and maximum formulations are stated in \cref{cor:crossing-formulations}.

\subsection{A six-copy construction for equidistant profiles}\label{sec:center}

We now turn from the Kemeny objective to the center objective.
Given any profile of three rankings, the construction below produces three pairwise-equidistant rankings whose optimal Kemeny score and center radius are explicit affine functions of the original optimal Kemeny score.
If the original profile is all\nobreakdash-$2$\nobreakdash-to\nobreakdash-$1$, then the output profile has the same property.

The following three-copy specialization of the construction of Biedl, Brandenburg, and Deng already gives pairwise equidistance and the correct center radius, but it orders every pair of candidates from different copies unanimously~\cite{BiedlBrandenburgDeng2009}.
For a profile $\Pi=(\pi_1,\pi_2,\pi_3)$, take three disjoint candidate copies $C_1,C_2,C_3$ and define
\[
\rho_1=\pi_1^{(1)}\pi_2^{(2)}\pi_3^{(3)},\qquad
\rho_2=\pi_2^{(1)}\pi_3^{(2)}\pi_1^{(3)},\qquad
\rho_3=\pi_3^{(1)}\pi_1^{(2)}\pi_2^{(3)},
\]
where $\pi_i^{(r)}$ is the relabeled copy of $\pi_i$ on $C_r$.
For any aggregate order, restrict it to each copy and relabel the restriction as an order of the original candidate set.
Within each copy, the three output restrictions are $\pi_1,\pi_2,\pi_3$ in some order, so the candidate pairs within that copy contribute at least $\Kemeny_\Pi^\star$ to the sum of the three output distances.
The three copies therefore contribute at least $3\Kemeny_\Pi^\star$, and the contributions from pairs in different copies are nonnegative.
Hence every aggregate has distance at least $\Kemeny_\Pi^\star$ from at least one output ranking.
If $\sigma^\star$ is Kemeny-optimal for $\Pi$, then $(\sigma^\star)^{(1)}(\sigma^\star)^{(2)}(\sigma^\star)^{(3)}$ has distance exactly $\Kemeny_\Pi^\star$ from each output ranking.
Thus the center radius of the three-copy profile is $\Kemeny_\Pi^\star$.

The three output rankings are also pairwise equidistant.
For each pair of output rankings, their restrictions to the three copies realize the three unordered pairs from $\{\pi_1,\pi_2,\pi_3\}$ once each, so their common distance is
\[
\KTdist(\pi_1,\pi_2)+\KTdist(\pi_1,\pi_3)+\KTdist(\pi_2,\pi_3).
\]
However, all three output rankings use the same copy order $C_1C_2C_3$.
Consequently, every pair of candidates from different copies is ordered unanimously.
Thus this three-copy construction gives pairwise equidistance and center radius $\Kemeny_\Pi^\star$, but it does not preserve the property that no candidate pair is ordered unanimously.

The six-copy construction uses three copy orders such that, for every pair of copies, two output rankings use one order and the third uses the reverse order.
Hence every candidate pair from different copies is split $2$\nobreakdash-to\nobreakdash-$1$.
In the theorem below, $C_1,\dots,C_6$ denote the six candidate copies, and $\sigma^{(r)}$ denotes the relabeled copy on $C_r$ of an order $\sigma$ on the original candidate set.

\begin{theorem}\label{thm:balanced-symmetrization}
Let $\Pi$ be a profile of three rankings on $n$ candidates.
In polynomial time, one can construct a profile $\widehat\Pi$ of three pairwise-equidistant rankings on $6n$ candidates such that
\[
\Kemeny_{\widehat\Pi}^{\star}=15n^2+6\Kemeny_\Pi^{\star},
\qquad
\mathcal R_{\widehat\Pi}^{\star}=5n^2+2\Kemeny_\Pi^{\star}
=\frac13\Kemeny_{\widehat\Pi}^{\star}.
\]
If $C_1,\dots,C_6$ are the six candidate copies, then
\[
\KOpt(\widehat\Pi)
=
\bigl\{\sigma_1^{(1)}\sigma_2^{(2)}\cdots\sigma_6^{(6)}:
\sigma_r\in\KOpt(\Pi)\text{ for every }r\bigr\}.
\]
Consequently, $|\KOpt(\widehat\Pi)|=|\KOpt(\Pi)|^6$.
Every optimal center of $\widehat\Pi$ is Kemeny-optimal and is at distance $\mathcal R_{\widehat\Pi}^{\star}$ from each output ranking.
For every $\sigma\in\KOpt(\Pi)$, the order $\sigma^{(1)}\cdots\sigma^{(6)}$ is an optimal center of $\widehat\Pi$.
The profile $\Pi$ has a unique Kemeny-optimal order if and only if $\widehat\Pi$ has a unique Kemeny-optimal order, and if and only if $\widehat\Pi$ has a unique Kendall--tau center.
If $\Pi$ is all\nobreakdash-$2$\nobreakdash-to\nobreakdash-$1$, then $\widehat\Pi$ is also all\nobreakdash-$2$\nobreakdash-to\nobreakdash-$1$.
\end{theorem}

\begin{proof}
Let $\Pi=(\pi_1,\pi_2,\pi_3)$ be a profile on a candidate set $C$ of size $n$.
Create six disjoint labeled copies $C_1,\dots,C_6$ of $C$.
For an order $\pi$ on $C$, write $\pi^{(r)}$ for its relabeled copy on $C_r$.
Set $\widehat\Pi=(\rho_1,\rho_2,\rho_3)$.
Each row in the following table gives the copy order of one output ranking and the input ranking used within each copy:
\[
\begin{array}{c|c|cccccc}
 & \text{copy order} & C_1 & C_2 & C_3 & C_4 & C_5 & C_6\\
\hline
\rho_1 & C_1C_2C_5C_6C_4C_3 & \pi_1 & \pi_2 & \pi_3 & \pi_1 & \pi_2 & \pi_3\\
\rho_2 & C_1C_3C_4C_6C_5C_2 & \pi_2 & \pi_3 & \pi_1 & \pi_2 & \pi_3 & \pi_1\\
\rho_3 & C_2C_3C_4C_5C_6C_1 & \pi_3 & \pi_1 & \pi_2 & \pi_3 & \pi_1 & \pi_2
\end{array}
\]

For $j\in\{1,2,3\}$, let $I_j$ be the inversion set of the copy order of $\rho_j$ relative to $C_1C_2\cdots C_6$.
Thus $I_j$ contains $(r,s)$ with $r<s$ exactly when $C_s$ precedes $C_r$ in $\rho_j$.
The three inversion sets are
\begin{equation}\label{eq:center-inversion-partition}
\begin{aligned}
I_1&=\{(3,4),(3,5),(3,6),(4,5),(4,6)\},\\
I_2&=\{(2,3),(2,4),(2,5),(2,6),(5,6)\},\\
I_3&=\{(1,2),(1,3),(1,4),(1,5),(1,6)\}.
\end{aligned}
\end{equation}
These sets partition the $\binom 62=15$ pairs of copy indices, and each has size five.
Hence the symmetric difference of any two inversion sets has size ten, so any two output rankings disagree on the relative order of ten pairs of copies.
Within the six copies, the restrictions of any two output rankings realize each unordered pair from $\{\pi_1,\pi_2,\pi_3\}$ exactly twice.
Writing
\[
S_\Pi=\KTdist(\pi_1,\pi_2)+\KTdist(\pi_1,\pi_3)+\KTdist(\pi_2,\pi_3),
\]
we obtain
\[
\KTdist(\rho_a,\rho_b)=10n^2+2S_\Pi
\]
for $1\leq a<b\leq3$.
Thus the three output rankings are pairwise equidistant.
If $\Pi$ is all\nobreakdash-$2$\nobreakdash-to\nobreakdash-$1$, then each candidate pair contributes two to $S_\Pi$, because one input ranking disagrees with the other two on that pair.
Hence $S_\Pi=2\binom n2=n(n-1)$, and the common output distance is $12n^2-2n$.

We next derive lower bounds for the Kemeny objective and the center radius.
Fix an aggregate order $\tau$ of the $6n$ output candidates, and let $\tau_r$ be its restriction to $C_r$, relabeled as an order of $C$.
For each copy, the restrictions of $\rho_1,\rho_2,\rho_3$ are $\pi_1,\pi_2,\pi_3$ in some order.
The contribution from pairs within the six copies is therefore at least
\[
\sum_{r=1}^6\sum_{i=1}^3\KTdist(\pi_i,\tau_r)
\geq 6\Kemeny_\Pi^\star.
\]
For a candidate pair with one candidate in each of two distinct copies, \cref{eq:center-inversion-partition} shows that two output rankings use one orientation and the third uses the opposite orientation.
Whichever orientation $\tau$ chooses, at least one output ranking disagrees with it.
There are $\binom 62n^2=15n^2$ such candidate pairs.
Adding the two contributions gives
\begin{equation}\label{eq:center-sum-lower}
\sum_{j=1}^3\KTdist(\rho_j,\tau)
\geq 15n^2+6\Kemeny_\Pi^\star.
\end{equation}
At least one of the three distances is at least one third of their sum, so
\begin{equation}\label{eq:center-radius-lower}
\max_{j\in\{1,2,3\}}\KTdist(\rho_j,\tau)
\geq 5n^2+2\Kemeny_\Pi^\star.
\end{equation}

Let $\sigma^\star$ be a Kemeny-optimal order of $\Pi$, and let $\widehat\sigma=(\sigma^\star)^{(1)}\cdots(\sigma^\star)^{(6)}$.
Each output copy order has five inversions relative to $C_1C_2\cdots C_6$, so pairs from different copies contribute $5n^2$ to the distance from $\widehat\sigma$.
Within the six copies, each output ranking uses each of $\pi_1,\pi_2,\pi_3$ exactly twice, so pairs within copies contribute $2\Kemeny_\Pi^\star$.
Therefore
\begin{equation}\label{eq:center-upper}
\KTdist(\rho_j,\widehat\sigma)=5n^2+2\Kemeny_\Pi^\star
\end{equation}
for $j=1,2,3$.
The order $\widehat\sigma$ attains both lower bounds, and hence
\begin{equation}\label{eq:balanced-values}
\Kemeny_{\widehat\Pi}^\star
=15n^2+6\Kemeny_\Pi^\star,
\qquad
\mathcal R_{\widehat\Pi}^\star
=5n^2+2\Kemeny_\Pi^\star.
\end{equation}

It remains to characterize equality in \cref{eq:center-sum-lower}.
An aggregate order is Kemeny-optimal for $\widehat\Pi$ exactly when every pair of candidates from different copies contributes one to the sum of the three distances and the restriction to every copy contributes $\Kemeny_\Pi^\star$.
Fix $r<s$ and a candidate pair with one candidate in $C_r$ and the other in $C_s$.
Two output rankings place the candidate from $C_r$ first, while the third places the candidate from $C_s$ first.
This pair contributes one if the aggregate places the candidate from $C_r$ first and two otherwise.
Thus equality forces every candidate of $C_r$ to precede every candidate of $C_s$ whenever $r<s$.
It also forces the restriction to each copy to be Kemeny-optimal for $\Pi$.
Conversely, every order $\sigma_1^{(1)}\cdots\sigma_6^{(6)}$ with $\sigma_r\in\KOpt(\Pi)$ has contribution $15n^2$ from pairs in different copies and contribution $6\Kemeny_\Pi^\star$ from pairs within copies.
This proves the stated formula for $\KOpt(\widehat\Pi)$ and the identity $|\KOpt(\widehat\Pi)|=|\KOpt(\Pi)|^6$.

Now let $\tau$ be an optimal center of $\widehat\Pi$.
Each of its three distances is at most $\mathcal R_{\widehat\Pi}^\star$, so their sum is at most $3\mathcal R_{\widehat\Pi}^\star=15n^2+6\Kemeny_\Pi^\star$.
By \cref{eq:center-sum-lower}, equality holds.
Thus $\tau$ is Kemeny-optimal, and all three distances equal $\mathcal R_{\widehat\Pi}^\star$.
For every $\sigma\in\KOpt(\Pi)$, the order $\sigma^{(1)}\cdots\sigma^{(6)}$ satisfies \cref{eq:center-upper} and is therefore an optimal center.
If $\Pi$ has a unique Kemeny-optimal order, the product formula gives a unique Kemeny-optimal output order and hence a unique center.
If $\Pi$ has two distinct Kemeny-optimal orders, using either order in all six copies gives two distinct optimal centers.
This proves the uniqueness equivalences.

Finally, suppose that $\Pi$ is all\nobreakdash-$2$\nobreakdash-to\nobreakdash-$1$.
For two candidates in the same copy, the three output restrictions are precisely the three input rankings, so their comparison is split $2$\nobreakdash-to\nobreakdash-$1$.
For candidates in different copies, the same conclusion follows from \cref{eq:center-inversion-partition}.
Hence $\widehat\Pi$ is all\nobreakdash-$2$\nobreakdash-to\nobreakdash-$1$.
\end{proof}

\subsection{Consequences for the center and other distance objectives}\label{sec:center-applications}

We apply \cref{thm:balanced-symmetrization} to the hard profiles from \cref{thm:main}.
The output profiles have three pairwise-equidistant rankings, remain all\nobreakdash-$2$\nobreakdash-to\nobreakdash-$1$, and have an optimal Kemeny score and center radius that are explicit affine functions of the original optimal Kemeny score.

\begin{proof}[Proof of \cref{thm:center}]
Apply \cref{thm:balanced-symmetrization} to a three-ranking \KemenyScore{} instance $(\Pi,\kappa)$ on $n$ candidates.
For the Kemeny objective, use threshold $15n^2+6\kappa$; for the center objective, use threshold $5n^2+2\kappa$.
By \cref{eq:balanced-values}, either output threshold is met exactly when $\Kemeny_\Pi^\star\leq\kappa$.
Pairwise equidistance and the all\nobreakdash-$2$\nobreakdash-to\nobreakdash-$1$ condition can be checked in polynomial time, and an aggregate order certifies either threshold bound.
Thus both restricted problems are in NP.

It remains to determine the common input distance.
Put $N=6n$ and $Q=\binom N2$.
For any three rankings, each candidate pair contributes either zero or two to the sum of their three pairwise distances, so this sum is at most $2Q$.
If the three distances have common value $D$, then $3D\leq2Q$.
In an all\nobreakdash-$2$\nobreakdash-to\nobreakdash-$1$ profile, every candidate pair contributes two, so equality holds and $D=2Q/3$.
The identity between the optimal Kemeny score and the center radius follows from \cref{thm:balanced-symmetrization}.
\end{proof}

After identifying the candidates with $\{1,\dots,N\}$, the two threshold problems in \cref{thm:center} are the threshold decision versions of $\textnormal{\textsc{PCM}}$-$3$ and $\textnormal{\textsc{PCM}}_{\max}$-$3$ in the terminology of Biedl, Brandenburg, and Deng~\cite{BiedlBrandenburgDeng2009}.

The exact-value and solution-recovery properties of the main reduction also transfer to the center objective.
Let $\Pi_G$ be the profile on $n_G$ candidates constructed from a graph $G$, and let $\widehat\Pi_G$ be its six-copy output.
Then
\[
\Kemeny_{\Pi_G}^{\star}
=\frac{\mathcal R_{\widehat\Pi_G}^{\star}-5n_G^2}{2}.
\]
Substituting this value into \cref{eq:maxcut-recovery} determines $\maxcutvalue(G)$.
Moreover, the restriction of any optimal center of $\widehat\Pi_G$ to any copy is Kemeny-optimal for $\Pi_G$, so \cref{prop:maxcut-recovery} recovers a maximum cut of $G$ in polynomial time.

The next lemma records the distance inequalities used for the remaining objectives.

\begin{lemma}\label{lem:center-majorization}
Let $\Pi$ be a profile of three rankings on $n$ candidates, let $\widehat\Pi=(\rho_1,\rho_2,\rho_3)$ be the output of \cref{thm:balanced-symmetrization}, and put
\[
r^\star=\mathcal R_{\widehat\Pi}^\star=5n^2+2\Kemeny_\Pi^\star.
\]
For an aggregate order $\sigma$, let
\[
d_1^\downarrow(\sigma)\geq d_2^\downarrow(\sigma)\geq d_3^\downarrow(\sigma)
\]
be its three distances from the output rankings in nonincreasing order.
Then, for $k=1,2,3$,
\[
\sum_{j=1}^k d_j^\downarrow(\sigma)\geq kr^\star.
\]
The order constructed in \cref{eq:center-upper} attains equality for all three values of $k$.
\end{lemma}

\begin{proof}
By \cref{eq:center-sum-lower}, the sum of the three distances is at least $3r^\star$.
For a nonincreasing triple, the average of the first $k$ entries is at least the average of all three entries.
Hence the sum of the first $k$ entries is at least $kr^\star$ for $k=1,2,3$.
The order in \cref{eq:center-upper} has distance vector $(r^\star,r^\star,r^\star)$.
\end{proof}

The next result concerns the sum-of-powers objective.

\begin{corollary}\label{cor:power-sum}
Fix an integer $p\geq1$.
Given three permutations $(\rho_1,\rho_2,\rho_3)$ and a nonnegative integer $L$, it is NP-complete to decide whether some aggregate order $\sigma$ satisfies
\[
\sum_{j=1}^3\KTdist(\rho_j,\sigma)^p\leq L.
\]
Hardness persists when the three input permutations satisfy the restrictions in \cref{thm:center}.
For the six-copy output of a profile $\Pi$ on $n$ candidates, the minimum value of this sum-of-powers objective is
\begin{equation}\label{eq:power-sum-value}
3\bigl(5n^2+2\Kemeny_\Pi^\star\bigr)^p.
\end{equation}
\end{corollary}

\begin{proof}
Let $d_j=\KTdist(\rho_j,\sigma)$ and retain $r^\star$ from \cref{lem:center-majorization}.
Convexity gives
\[
\sum_{j=1}^3d_j^p
\geq3\left(\frac{d_1+d_2+d_3}{3}\right)^p
\geq3(r^\star)^p.
\]
The order in \cref{eq:center-upper} has all three distances equal to $r^\star$, so equality is attained and \cref{eq:power-sum-value} follows.
For an input threshold $\kappa$, use $L=3(5n^2+2\kappa)^p$.
Since $x\mapsto x^p$ is strictly increasing on the nonnegative reals, the minimum output value is at most $L$ exactly when $\Kemeny_\Pi^\star\leq\kappa$.
Because $p$ is fixed, the threshold has polynomial encoding length.
Membership in NP follows by evaluating a proposed aggregate order.
\end{proof}

The case $p=1$ is the Kemeny objective.
For $p=2$, the objective is the Squared Kemeny rule, also called Kemeny's mean rule, and its minimizers are the empirical Fr\'echet means for Kendall--tau distance~\cite{LedererPetersWas2024}.

We next consider an ordered weighted sum of the distances.

\begin{corollary}\label{cor:weighted-sorted-distance}
Fix rational weights $w_1\geq w_2\geq w_3\geq0$ that are not all zero.
Given three permutations and a nonnegative rational threshold $L_w$, it is NP-complete to decide whether some aggregate order $\sigma$ satisfies
\[
\sum_{j=1}^3w_jd_j^\downarrow(\sigma)\leq L_w,
\]
where $d_1^\downarrow(\sigma)\geq d_2^\downarrow(\sigma)\geq d_3^\downarrow(\sigma)$ are the three input distances in nonincreasing order.
Hardness persists under the restrictions in \cref{thm:center}.
Writing $W=w_1+w_2+w_3$, the minimum ordered weighted sum of the distances on the six-copy output of a profile $\Pi$ on $n$ candidates is
\[
W\bigl(5n^2+2\Kemeny_\Pi^\star\bigr).
\]
This family includes the maximum distance, the sum of the two largest distances, and the total distance.
\end{corollary}

\begin{proof}
Retain $r^\star$ and the sorted distances from \cref{lem:center-majorization}.
Then
\[
w_1d_1^\downarrow+w_2d_2^\downarrow+w_3d_3^\downarrow
=(w_1-w_2)d_1^\downarrow+(w_2-w_3)(d_1^\downarrow+d_2^\downarrow)+w_3(d_1^\downarrow+d_2^\downarrow+d_3^\downarrow)
\geq(w_1+w_2+w_3)r^\star.
\]
The order in \cref{eq:center-upper} attains equality.
For an input threshold $\kappa$, use $L_w=W(5n^2+2\kappa)$.
Since $W>0$, the minimum output value is at most $L_w$ exactly when $\Kemeny_\Pi^\star\leq\kappa$.
Multiplying the fixed weights and the threshold by a common denominator gives an equivalent integer comparison of polynomial encoding length.
Membership in NP follows by evaluating a proposed aggregate order.
\end{proof}

The crossing-distance correspondence identifies the sum and maximum of the crossing counts in three drawings with the Kemeny and center objectives, respectively.

\begin{corollary}\label{cor:crossing-formulations}
Consider three one-sided two-layer drawings whose edge sets are perfect matchings, whose fixed-layer orders are specified by three permutations of $N$ elements, and whose free layer must use the same order in all three drawings.
It is NP-complete to decide whether the sum of the three crossing counts is at most a given nonnegative integer, and it is NP-complete to decide whether their maximum is at most a given nonnegative integer.
Hardness persists when the three fixed-layer permutations are pairwise equidistant with common distance $\frac23\binom N2$ and every element pair appears in both relative orders among them.
Equivalently, the common and maximum monochromatic one-sided crossing problems are NP-complete for three edge-coloured perfect matchings when crossings are counted only between edges of the same colour.
\end{corollary}

\begin{proof}
For each perfect matching, the crossing count against the common free-layer order is its Kendall--tau distance from that order.
The two claims therefore follow from \cref{thm:center}.
\end{proof}

In the Mallows model with dispersion parameter $\phi\in(0,1)$, the probability of an observed ranking at Kendall--tau distance $d$ from the central ranking is proportional to $\phi^d$~\cite{Mallows1957}.

\begin{corollary}\label{cor:mallows}
Fix $\phi\in(0,1)$ in the Mallows model under Kendall--tau distance.
For three observed rankings, computing a maximum-likelihood central ranking is NP-hard.
It is $\ThetaTwoP$-complete to decide whether a designated candidate is first in some or every maximum-likelihood central ranking and whether one designated candidate precedes another in some or every such ranking.
Verifying whether a supplied ranking is a maximum-likelihood central ranking or the unique such ranking is coNP-complete.
These statements hold when both orientations of every candidate pair occur among the observations and their majority tournament has majority dimension exactly $3$.
The NP-hardness of computing a central ranking also persists when the observations are pairwise equidistant with common distance $\frac23\binom N2$.
\end{corollary}

\begin{proof}
Let $Z_N(\phi)$ be the normalizing constant for one observation.
For a central ranking $\sigma$, the likelihood of independent observations $\pi_1,\pi_2,\pi_3$ is
\[
Z_N(\phi)^{-3}
\phi^{\KTdist(\pi_1,\sigma)+\KTdist(\pi_2,\sigma)+\KTdist(\pi_3,\sigma)}.
\]
The factor $Z_N(\phi)^{-3}$ is independent of $\sigma$.
Since $0<\phi<1$, maximizing the likelihood is equivalent to minimizing the Kemeny objective.
The winner, precedence, and recognition complexity statements follow from \cref{thm:three-ranking-optimality}, while the pairwise-equidistant computing hardness follows from \cref{thm:center}.
\end{proof}

\begin{remark}[Scope of the crossing and support consequences]\label{rem:center-limits}
The crossing corollary counts the three drawings separately, or counts only monochromatic crossings after they are merged.
It does not imply hardness for ordinary uncoloured one-sided crossing minimization on all degree-$3$ bipartite graphs.
The hard profiles in \cref{thm:center} have exact support $2/3$.
The transitivity argument above the two-thirds threshold applies to the Kemeny objective because that objective reduces to the majority tournament; it does not give the same reduction for the center objective, the sum-of-powers objective, or an ordered weighted sum of the distances.
No support-threshold classification for those objectives is claimed here.
\end{remark}

\subsection{Classification by the Number of Rankings}\label{sec:fixed-profile-sizes}

The preceding results complete the complexity classification for every fixed number of input rankings.

\begin{theorem}\label{thm:fixed-profile-size}
Fix an integer $q\geq1$.
\begin{enumerate}[label=\textup{(\roman*)}]
\item If $q\leq2$, \KemenyScore{}, \KendallCenter{}, the four winner and precedence problems, and both recognition problems are polynomial-time solvable.
\item If $q\geq3$, \KemenyScore{} and \KendallCenter{} are NP-complete, the four winner and precedence problems are $\ThetaTwoP$-complete, and \KemenyConsensusRecognition{} and \UniqueKemenyConsensusRecognition{} are coNP-complete.
\item If $q=3$ or $q\geq5$, \KemenyScore{} remains NP-complete, the four winner and precedence problems remain $\ThetaTwoP$-complete, and both recognition problems remain coNP-complete when every pair has support $\lfloor q/2\rfloor+1$; the hard instances may be chosen with majority dimension exactly $3$.
\item For $q=4$, \KemenyScore{}, the four winner and precedence problems, and both recognition problems are polynomial-time solvable when every pair has support at least $3$.
When every support belongs to $\{2,3\}$, \KemenyScore{} is NP-complete, the four winner and precedence problems are $\ThetaTwoP$-complete, and both recognition problems are coNP-complete.
\end{enumerate}
\end{theorem}

\begin{proof}
For one ranking, its order is the unique Kemeny-optimal aggregate and the center radius is zero.
For two rankings, a pair on which they disagree contributes one to every Kemeny aggregate, while the agreed comparisons form a partial order whose linear extensions are exactly the Kemeny-optimal aggregates.
Possible and necessary first positions and precedence relations are therefore decided by this partial order.
A supplied order is Kemeny-optimal exactly when it is a linear extension, and it is uniquely Kemeny-optimal exactly when it is a linear extension and the partial order is total.
For the center objective, if the two rankings have distance $D$, the triangle inequality gives radius at least $\lceil D/2\rceil$, and an order halfway along a shortest adjacent-swap path attains this value.
This proves (i).

For the Kemeny problems in item (ii), every profile satisfies support at least $\lceil q/2\rceil$, and $3\lceil q/2\rceil\leq2q$ for every $q\geq3$.
The hard case of \cref{thm:support-dichotomy} therefore applies.
The center case $q=3$ is \cref{thm:center}, while Biedl, Brandenburg, and Deng proved NP-completeness for every fixed $q\geq4$~\cite{BiedlBrandenburgDeng2009}.

For item (iii), put $s=\lfloor q/2\rfloor+1$.
When $q$ is odd and at least $3$, or even and at least $6$, one has $3s\leq2q$.
The exact-support clause of \cref{thm:support-dichotomy} proves the hardness assertions in item~(iii) under the promise that the majority dimension is exactly $3$.
Finally, item (iv) is \cref{cor:four-margin-heterogeneity}.
\end{proof}

The fixed-$q$ threshold classifications for \KemenyScore{} and \KendallCenter{} translate to the threshold decision versions of $\textnormal{\textsc{PCM}}$-$q$ and $\textnormal{\textsc{PCM}}_{\max}$-$q$, respectively.
Under the matching representation in \cref{cor:crossing-formulations}, they also give the fixed-$q$ classifications for common one-sided crossing minimization on $q$ perfect matchings with the total or maximum crossing count and with repeated fixed-layer permutations allowed.

Restricting the Kemeny objective to tie-free profiles changes the classification only for $q=4$.

\begin{corollary}\label{cor:tie-free-fixed-q}
Fix an integer $q\geq1$.
On tie-free profiles of exactly $q$ rankings, \KemenyScore{} is polynomial-time solvable for $q\in\{1,2,4\}$ and NP-complete for $q=3$ and every $q\geq5$.
Thus $q=4$ is the unique fixed profile size at least three for which the tie-free restriction is polynomial-time solvable.
\end{corollary}

\begin{proof}
The cases $q\leq2$ are immediate.
The case $q=3$ is \cref{thm:main}.
For $q=4$, every pair in a tie-free profile has support at least $3>2q/3$, so \cref{thm:support-dichotomy} applies.
For odd $q\geq5$, choose exact support $(q+1)/2$; for even $q\geq6$, choose exact support $q/2+1$.
In both cases the exact-support hard clause of \cref{thm:support-dichotomy} applies.
\end{proof}

\section{Conclusion}\label{sec:conclusion}

The three-ranking hardness results and the profile transformations establish the fixed-profile-size classification for the Kemeny score, the four winner and precedence problems, and the two recognition problems.
For every fixed $q\leq2$, these problems are polynomial-time solvable.
For every fixed $q\geq3$, \KemenyScore{} is NP-complete, the four winner and precedence problems are $\ThetaTwoP$-complete, and the two recognition problems are coNP-complete.
The center problem \KendallCenter{} is polynomial-time solvable for $q\leq2$ and NP-complete for every fixed $q\geq3$.
The hard instances used for the three-ranking Kemeny results have exactly three all\nobreakdash-$2$\nobreakdash-to\nobreakdash-$1$ rankings and induce tournaments of majority dimension exactly $3$.

The pairwise-support classifications give the same sharp boundary for the score, winner, precedence, and recognition problems.
For fixed $q$, minimum pairwise support $s$ permits NP-hardness for \KemenyScore{}, $\ThetaTwoP$-hardness for the four winner and precedence problems, and coNP-hardness for the two recognition problems exactly when $3s\leq2q$; support strictly above two thirds forces a transitive majority tournament and a unique Kemeny-optimal aggregate.
Exact-support constructions establish sharpness when $s>q/2$, while two adjacent support values suffice in the boundary case $s=q/2$.
For four rankings, support at least $3$ makes \KemenyScore{}, the four winner and precedence problems, and the two recognition problems polynomial-time solvable, whereas allowing supports in $\{2,3\}$ yields NP-completeness, $\ThetaTwoP$-completeness, and coNP-completeness, respectively.

The reduction proving \cref{thm:main} preserves more than the answer to the threshold decision problem: its optimal Kemeny score determines the exact maximum-cut value, and a Kemeny-optimal aggregate recovers a maximum cut.
The identity in \cref{eq:all-two-to-one-score} transfers the score, winner, precedence, and recognition results to Slater orders, permutation medians, and fixed-dispersion Mallows central rankings.
The six-copy construction gives pairwise-equidistant hard profiles, an exact affine relation between the optimal Kemeny score and the center radius, and a description of all Kemeny-optimal orders in terms of their restrictions to the six copies.
It also shows that the original profile has a unique Kemeny-optimal order exactly when the output profile has a unique Kemeny-optimal order and a unique center.

\section*{Acknowledgement}

This research has been implemented with the support provided by the Ministry of Innovation and Technology NRDI Office within the framework of the Artificial Intelligence National Laboratory Program; by the Ministry of Innovation and Technology of Hungary from the National Research, Development and Innovation Fund, financed under the ELTE TKP 2021-NKTA-62 funding scheme; and by the Lend\"{u}let Programme of the Hungarian Academy of Sciences --- grant number LP2021-1/2021.

\medskip\noindent
ChatGPT 5.5 and 5.6 Sol were used for literature-search assistance, language editing, presentation refinement, and informal proof checking.
The author is responsible for all claims and proofs.

\begingroup
\bibliographystyle{plain}
\bibliography{bibliography}
\endgroup

\end{document}